\documentclass[12pt]{article}
\usepackage{amsmath}
\usepackage{graphicx}
\usepackage{enumerate}
\usepackage{natbib}
\usepackage{url} 
\usepackage{multibib}

\newcites{Supp}{References for the Supplement}

\usepackage{bm, commath}
\usepackage{caption}
\usepackage{graphicx}
\usepackage{subcaption}
\usepackage{amsmath, amsfonts, amsthm}
\usepackage{float}
\usepackage{enumitem}
\usepackage{booktabs,siunitx}
\usepackage{url}
\usepackage{multirow}
\usepackage{amssymb}
\usepackage{bm}
\usepackage{mathrsfs}
\usepackage{authblk}
\usepackage[normalem]{ulem}
\allowdisplaybreaks

\usepackage{listings}
\DeclareMathOperator*\lowlim{\underline{lim}}
\DeclareMathOperator*\uplim{\overline{lim}}

\DeclareMathOperator*{\argmin}{arg\,min}

\usepackage{bbm}
\usepackage{textcomp}
\usepackage{authblk}
\usepackage[ruled]{algorithm2e}
\SetKwInput{KwParam}{Parameter}
\SetAlgoCaptionLayout{centerline}

\usepackage{sectsty}
\usepackage[onehalfspacing]{setspace}

\usepackage[usenames,dvipsnames]{color}
\usepackage[utf8]{inputenc}
\usepackage{colortbl}
\usepackage{float}
\usepackage{tikz}
\usepackage{bbding}
\usepackage{longtable,booktabs,array,pdflscape}

\usetikzlibrary{trees}
\usetikzlibrary{shapes,decorations,arrows,calc,arrows.meta,fit,positioning}
\usetikzlibrary{shapes.multipart}
\tikzset{
    -Latex,auto,node distance =1 cm and 1 cm,semithick,
    state/.style ={ellipse, draw, minimum width = 0.7 cm},
    state1/.style ={ draw, minimum width = 0.7 cm},
    point/.style = {circle, draw, inner sep=0.04cm,fill,node contents={}},
    bidirected/.style={Latex-Latex,dashed},
    el/.style = {inner sep=2pt, align=left, sloped}
}

\newtheorem{theorem}{Theorem}

\newtheorem{corollary}{Corollary}

\newtheorem{example}{Example}
\newtheorem{assumption}{Assumption}
\newtheorem{proposition}{Proposition}

\newtheorem{lemma}{Lemma}

\usepackage{mathtools}

\usepackage{xcolor}

\makeatletter
\renewcommand{\algocf@captiontext}[2]{#1\algocf@typo. \AlCapFnt{}#2} 
\def\@algocf@capt@plain{top}
\renewcommand{\algocf@makecaption}[2]{%
  \addtolength{\hsize}{\algomargin}%
  \sbox\@tempboxa{\algocf@captiontext{#1}{#2}}%
  \ifdim\wd\@tempboxa >\hsize
  \hskip .5\algomargin%
  \parbox[t]{\hsize}{\algocf@captiontext{#1}{#2}}
  \else%
  \global\@minipagefalse%
  \hbox to\hsize{\box\@tempboxa}
  \fi%
  \addtolength{\hsize}{-\algomargin}%
}
\makeatother

\definecolor{dblue}{HTML}{0072B2}

\def\bmZ{\bm{Z}}
\def\bmX{\bm{X}}
\def\bmU{\bm{U}}
\def\red{\color{red}}
\def\bgamma{\bm{\gamma}}
\def\bGamma{\bm{\Gamma}}
\def\bbeta{\bm{\beta}}
\def\var{{\rm var}}
\def\cov{{\rm cov}}
\newcommand{\probP}{\text{I\kern-0.15em P}}

\newcommand{\blind}{1}

\usepackage{etoolbox}
\newcommand{\zerodisplayskips}{%
  \setlength{\abovedisplayskip}{3pt}%
  \setlength{\belowdisplayskip}{3pt}%
  \setlength{\abovedisplayshortskip}{3pt}%
  \setlength{\belowdisplayshortskip}{3pt}}
\appto{\normalsize}{\zerodisplayskips}
\appto{\small}{\zerodisplayskips}
\appto{\footnotesize}{\zerodisplayskips}

\begin{document}

\def\spacingset#1{\renewcommand{\baselinestretch}%
{#1}\small\normalsize} \spacingset{1}


\if1\blind
{
	 \begin{center} 
	\spacingset{1.5} 	{\LARGE\bf  A More Robust Approach to Multivariable Mendelian Randomization} \\ \bigskip \bigskip
		\spacingset{1} 
		{\large Yinxiang Wu$ ^1 $, Hyunseung Kang$ ^2 $, and Ting Ye$ ^1 $ } \\ \bigskip
	 {$ ^1 $Department of Biostatistics, University of Washington, Seattle, Washington, U.S.A. \\ 
    $ ^2 $Department of Statistics, University of Wisconsin-Madison, Madison, Wisconsin, U.S.A.\\}
	\end{center}
} \fi

\if0\blind
{
  \bigskip
  \bigskip
  \bigskip
  \begin{center}
  \spacingset{1.5}
    {\LARGE\bf A More Robust Approach to Multivariable Mendelian Randomization}
\end{center}
  \medskip
} \fi

\bigskip
\begin{abstract}
Multivariable Mendelian randomization (MVMR) uses genetic variants as instrumental variables to infer the direct effects of multiple exposures on an outcome. However, unlike univariable Mendelian randomization, MVMR often faces greater challenges with many weak instruments, which can lead to bias not necessarily toward zero and inflation of type I errors. In this work, we introduce a new asymptotic regime that allows exposures to have varying degrees of instrument strength, providing a more accurate theoretical framework for studying MVMR estimators. Under this regime, our analysis of the widely used multivariable inverse-variance weighted method shows that it is often biased and tends to produce misleadingly narrow confidence intervals in the presence of many weak instruments. To address this, we propose a simple, closed-form modification to the multivariable inverse-variance weighted estimator to reduce bias from weak instruments, and additionally introduce a novel spectral regularization technique to improve finite-sample performance. We show that the resulting spectral-regularized estimator remains consistent and asymptotically normal under many weak instruments. Through simulations and real data applications, we demonstrate that our proposed estimator and asymptotic framework can enhance the robustness of MVMR analyses.
\end{abstract}

\noindent%
\vfill

\newpage
\spacingset{1.8} 
\section{Introduction}
\label{sec: intro}

\subsection{Weak instruments in multivariable Mendelian randomization}
\label{subsec: mr}

Multivariable Mendelian randomization (MVMR) is an instrumental variable (IV) method that uses genetic variants, typically single-nucleotide polymorphisms (SNPs), as instruments to estimate the direct causal effects of multiple exposures on an outcome in the presence of unmeasured confounding \citep{ Davey-Smith:2003aa, 
sanderson2022Mendelian}. MVMR requires a set of SNPs that is greater than or equal to the number of exposures, and these SNPs must satisfy the following three IV assumptions (also illustrated in Figure \ref{fig:mvmr}) \citep{sanderson2019examination}:  (i) Relevance: the SNPs must be conditionally associated with each exposure, given the other exposures; (ii) Independence: the SNPs must be independent of all unmeasured confounders of the exposures and the outcome; (iii) Exclusion restriction: the SNPs must affect the outcome only through the exposures. 

In this paper, we focus on MVMR using summary data from genome-wide association studies and study the problem of weak IVs in this setup. In MVMR, weak IVs are SNPs that are weakly associated with some exposures, conditional on the other exposures. Weak IVs present greater challenges in MVMR compared to univariable MR \citep{patel2024weak}. First, even when the SNPs have strong associations with each exposure individually (which implies strong IVs in univariable MR), it is possible that their effects on each exposure are highly correlated, making them weakly associated with some exposures, conditional on the others \citep{sanderson2021testing}. This also means that the common practice of selecting SNPs marginally associated with each exposure, or even all exposures, using a p-value cutoff is insufficient to mitigate the weak IV problem in MVMR. Second, unlike in univariable MR, weak IVs in MVMR can bias the estimated direct effects of different exposures either toward or away from zero \citep{sanderson2021testing}, and can inflate type I error rates when testing the null hypothesis of no direct effect for an exposure. Third, while measuring IV strength in univariable MR using the first-stage $F$-statistic is well established \citep{Stock2001:ivtest, sanderson2022Mendelian}, there is little work on how to measure IV strength in MVMR with summary data in a way that accurately reflects the potential for bias or inflated type I error rates.

\begin{figure}[ht]
	\centering
		\begin{tikzpicture}
			\node[color=dblue] (1) { SNP 2};
   			\node[] (12) [below= of 1, yshift=1.25cm] { $\vdots$};
			\node[color=dblue, below=of 1,yshift=5mm] (6) { SNP $p$};
			\node[color=dblue, above=of 1,yshift=-5mm] (7) { SNP 1};
			\node[right= of 1, yshift=-1cm] (2) { Exposure $K$};
			\node[] (4) [right= of 1, yshift=1cm] { Exposure 1};
			\node[] (10) [below= of 4, yshift=.7cm] {\footnotesize $\vdots$};
			\node[] (3) [right =of 1,xshift=3cm] {Outcome};
			\path (2) edge[bend left=30,dashed] node[el,above] {} (4);
			\node[gray, align=left] (5) [below =of 2, xshift=1.55 cm, yshift=6mm] 			{ Unmeasured\\  confounders}; 
			\path[gray] (5) edge node[el,above] {} (2);
			\path[gray] (5) edge node[el,above] {} (3);
			\path[gray] (5) edge node[el,above] {} (4);
			\path[color=dblue] (1) edge node[el,above] {} (2);
			\path[color=dblue] (6) edge node[el,above] {} (2);
			\path[color=dblue] (1) edge node[el,above] {} (4);
			\path[color=dblue] (7) edge node[el,above] {} (4);
			\path[color=red] (2) edge node[el,above] {{\color{red}$ \beta_{0K} $}} (3);
			\path[color=red] (4) edge node[el,above] {{\color{red}$ \beta_{01} $}} (3);
	\end{tikzpicture}	
	\caption{   Directed acyclic graph (DAG) of IV assumptions in MVMR. The SNPs must be (i) conditionally associated with each exposure given the other exposures, (ii) independent of all unmeasured confounders, and (iii) affect the outcome only through the exposures.}\label{fig:mvmr}
\end{figure}

\subsection{Related work and our contributions}

In the econometrics literature, when individual-level data are available, the three most relevant approaches to studying bias from weak IVs and type I error inflation are (a) \emph{weak IV asymptotics} \citep{Staiger:1997aa}, (b) \emph{many weak IV asymptotics} \citep{Chao:2005aa}, and (c) \emph{many weak moment asymptotics} \citep{Newey:2009aa}. Weak IV asymptotics assume a finite number of IVs, with all IV-exposure associations being local to zero (i.e., in an $n^{-\frac{1}{2}}$ neighborhood of 0, with $n$ being the sample size). Many weak IV asymptotics, on the other hand, take the number of IVs to infinity as a function of $n$, but require that all exposures have the same degree of IV strength. Many weak moment asymptotics, developed for the generalized method of moments, allow linear combinations of the target parameters to have different degrees of identification. Under weak IV asymptotics, \cite{Stock2001:ivtest} proposed a test for weak IVs based on the Cragg-Donald statistic \citep{cragg1993testing}. \cite{sanderson2016weak} then proposed a conditional $F$-test under a different form of weak IV asymptotics, where the matrix of true IV-exposure associations is local to a rank reduction of one. This method was extended to summary-data MVMR by \cite{sanderson2021testing} to test whether each exposure is strongly predicted by genetic variants, conditional on the other exposures. An in-depth survey of weak IVs can be found in \cite{andrews2005inference} and \cite{andrews2019weak}.
 
 The work most closely related to this paper is that of \cite{sanderson2021testing}, which proposed the conditional $F$-statistic. As we discuss in Section 3, the conditional $F$-statistic for an exposure measures IV strength for a specific linear combination of exposures, but a large conditional $F$-statistic may still be insufficient to ensure reliable inference regarding the direct causal effect of that exposure.

Due to the increasing popularity of MVMR, many estimators have been proposed. Examples include the multivariable inverse-variance weighted (MV-IVW) method \citep{Burgess2015mvmr}, MV-Median   \citep{grant2021pleiotropy}, MV-Egger \citep{rees2017extending}, GRAPPLE \citep{Wang2021grapple}, and the minimizer of the heterogeneity $Q$-statistic  \citep{sanderson2021testing}. Among these, the MV-IVW method is the most widely used, and it is asymptotically equivalent to the two-stage least squares estimator with individual-level data \citep{sanderson2021testing}. However, the MV-IVW method is known to be biased in the presence of weak IVs \citep{sanderson2019examination}. \cite{lorincz2024mrbee} proposed MRBEE, a method that uses bias-corrected estimation equations to reduce the bias from weak IVs in the MV-IVW estimator. However, MRBEE assumes that genetic effects follow a random-effects model, implicitly requiring that the IV strengths for different exposures be of the same magnitude. \cite{sanderson2021testing}, \cite{Wang2021grapple}, and \cite{lin2023robust} proposed more robust estimators against weak IVs, but they can be computationally intensive and numerically unstable when applied to more than a few exposures. When there is a large number of potentially highly correlated exposures, \cite{zuber2020selecting} and \cite{grant2022efficient} proposed selecting exposures using a Bayesian approach and Lasso-type regularization, respectively. However, these methods still apply MV-IVW after selecting the exposures and thus can be biased due to weak IVs.

This paper makes three key contributions toward improving the robustness of MVMR analysis. First, we propose a general asymptotic regime that allows for varying degrees of IV strength across exposures, offering a more accurate asymptotic framework for approximating the finite-sample behavior of MVMR estimators with many weak IVs. Our framework can be viewed as a summary-data adaptation of the many weak moment asymptotics of \cite{Newey:2009aa}. 
Unlike the many weak IV framework of \citet{Chao:2005aa}, we do not make the restrictive assumption that all exposures have identical IV strength. 
Our numerical studies confirm that our asymptotic framework accurately predicts the behavior of several MVMR estimators in finite samples. Second, under our new framework, we formally show that the popular MV-IVW method often produces misleadingly narrow confidence intervals centered around biased estimates, which may lead to spurious findings. The magnitude and direction of this bias depend on the true causal effects of the exposures and the overall IV strength, for which we propose a scalar metric to quantify in practice. Third, we propose a novel spectral-regularized estimator of the causal effects of the exposures and show that it is consistent and asymptotically normal under our framework.

\section{Notation and setup}
\label{sec: notation and setup}

We are interested in estimating the direct causal effects of $K$ potentially correlated exposures, denoted as $\bmX = (X_1,\dots,X_K)^T$, on an outcome $Y$, which are confounded by unmeasured variables $\bmU$. Suppose there are $p$ mutually independent SNPs, denoted as $\bmZ = (Z_1,\dots,Z_p)^T$, obtained through linkage disequilibrium (LD) pruning or clumping \citep{Hemani:2018ab}. Following \cite{Wang2021grapple}, we adopt the following structural equation models 
\begin{align}
	  &  E(X_k|\bmZ, \bmU)  = f_k(\bmZ, \bm U), \quad k=1,\dots, K, \label{eq: exp}  \\ 
	  &  Y = \bm X^T \bbeta_0 + g(\bm U, E_Y),\label{eq: out}
\end{align}
where $\bm\beta_0$ is the vector of causal effects of interest, $E_Y$ represents random noise,
and $f_1,\dots,f_K, g$ are unspecified functions. The exposure model in \eqref{eq: exp} is unspecified, allowing full flexibility in how $\bmZ$ and $\bmU$ affect $\bmX$. For example, this model accommodates scenarios where (a) each SNP has null effects on certain exposures, (b) there are non-linear associations between $(\bmZ, \bmU)$ and discrete $\bmX$, (c) there are interactions between $\bmZ$ and $\bmU$, and (d) there are direct effects among the $K$ exposures. On the other hand, the outcome model in \eqref{eq: out} assumes linear causal effects of the exposures on the outcome, and from the third IV assumption, $\bmZ$ has no direct effect on $Y$ once we have conditioned on all the exposures. Additionally, from the second IV assumption, $\bm Z$ is independent of $(E_Y, \bmU)$.

In our main setting of two-sample summary-data MVMR, summary-level data on the SNP-exposure and SNP-outcome associations are obtained from two separate genome-wide association studies. Formally, we observe $ \{ \hat\bgamma_j, \hat \Gamma_j, \Sigma_{Xj}, \sigma_{Yj}^2; j=1,\dots, p \}$, where $\hat\bgamma_j \in \mathbb{R}^K$ is a vector of ordinary least squares estimates from the marginal regression of the exposures on $Z_j$, $\Sigma_{Xj} \in \mathbb{R}^{K\times K}$ is the variance-covariance matrix of $\hat\bgamma_j$, $\hat\Gamma_j \in \mathbb{R}$ is the ordinary least squares estimate from the marginal regression of the outcome on $Z_j$, and $\sigma_{Yj} \in \mathbb{R}$ is the standard error of $\hat\Gamma_j$. When $\Sigma_{Xj}$ is not available, we show how to estimate it based on summary-level data in Section S5.3 of the Supplement. 

A key implication of models \eqref{eq: exp}--\eqref{eq: out} is that $\Gamma_j =\bm\gamma_j^T \bm\beta_0$ holds for $j = 1,\dots,p$, where $\Gamma_j$ and $\bgamma_j$ represent the true marginal SNP-outcome and SNP-exposure associations, respectively. This relationship is derived from the properties of least squares, where $\Gamma_j = \var(Z_j)^{-1} \cov (Y, Z_j)$ and $\bgamma_j = \var(Z_j)^{-1} \cov (\bmX, Z_j)$. Then, based on the outcome model in \eqref{eq: out} and the independence between $\bmZ$ and $(E_Y, \bmU)$, we can show that 
\begin{align*}
    \Gamma_j = \var(Z_j)^{-1} \cov(Y, Z_j) = \var(Z_j)^{-1} \cov(\bmX^T \bbeta_0, Z_j) =  \var(Z_j)^{-1}\cov(\bmX^T , Z_j)  \bbeta_0  = \bgamma_j^T\bbeta_0. 
\end{align*}
We remark that the relationship $\Gamma_j = \bgamma_j^T \bbeta_0$ holds regardless of whether the SNP-exposure relationship is linear.
 Also, \cite{Wang2021grapple} showed that the relationship $\Gamma_j = \bgamma_j^T \bbeta_0$ remains approximately correct when $Y$ is binary and follows a logistic regression model.

We denote the sample sizes for the outcome and each exposure by $n_Y$ and $n_{Xk}$, for $k=1,\dots, K$, respectively. The samples for different exposures may have no overlap, partial overlap, or complete overlap, but they are assumed to be independent of the outcome sample. 

Let \( {\mathrm {diag}}(a_1, \dots, a_K) \) denote a \( K \times K \) diagonal matrix with entries \( a_1, \dots, a_K \). The following assumption formalizes the basic setup of two-sample summary-data MVMR.

\begin{assumption}\label{assump: 1}
   (i) The sample sizes $n_Y$ and $n_{Xk}$ for $k=1,\dots, K$, tend to infinity at the same rate. Let $n=\min (n_Y, n_{X1}, \dots,  n_{XK}) $. The number of exposures $K$ is a finite constant. The number of SNPs $p$ grows to infinity as $n $ increases. \\
   (ii) The $2p$ random elements in the set 
    $\{{\hat\bgamma_j, \hat\Gamma_j}, j=1,\dots, p \} $ are mutually  independent. For each $j$, $\hat\bgamma_j \sim N(\bgamma_j, \Sigma_{Xj}), \hat \Gamma_j \sim N(\Gamma_j, \sigma_{Yj}^2) $, and $  \Gamma_j =\bm\gamma_j^T \bm\beta_0  $, with $\bgamma_j = (\gamma_{j1}, \dots, \gamma_{jK})^T$. The $K\times p$ matrix $(\bgamma_1, \dots, \bgamma_{p})$ is of full rank and $K\leq p$. \\
    (iii)   The variance-covariance matrix $\Sigma_{Xj}  = 
	 {\mathrm {diag}}(\sigma_{Xj1}, \dots, \sigma_{XjK}) \ \Sigma \  {\mathrm {diag}}(\sigma_{Xj1}, \dots, \sigma_{XjK})$ and $\sigma_{Yj}^2$ are known, where $\sigma_{Xjk}$ is the standard error of $\hat \gamma_{jk}$, and $\Sigma $ represents a shared positive definite correlation matrix. Moreover, the variance ratios $\sigma_{Xjk}^2/\sigma_{Yj}^2$ are bounded away from zero and infinity for all $j =1,\dots, p$ and $k = 1,\dots,K$.
\end{assumption}
The normality assumption in Assumption 1(ii) is plausible in many MVMR studies, given the large sample sizes in modern genome-wide association studies.
The independence between the $\hat\bgamma_j$'s and $\hat \Gamma_j$'s in Assumption 1(ii) is guaranteed by the two-sample design, where the exposure and outcome samples are non-overlapping.  The independence among the $\hat\bgamma_j$'s and among the $\hat \Gamma_j$'s (as well as the boundedness of variance ratios in Assumption 1(iii)) are reasonable because the SNPs are mutually independent from LD pruning or clumping and each SNP accounts for only a small proportion of the total variance in the exposure and outcome variables; see Section S3.1 of the Supplement for further details. The conditions that $K\leq p$ and that the matrix $(\bgamma_1, \dots, \bgamma_{p})$ is of full rank are assumed for the identifiability of $\bbeta_0$, where the full rank condition corresponds to the relevance assumption described in Section \ref{sec: intro}.  

In Assumption \ref{assump: 1}(iii), we assume that the variances are known, meaning they can be estimated with negligible error. This is a standard assumption and is reasonable given large sample sizes \citep{Ye:2019}. We also assume the existence of a shared correlation matrix $\Sigma$ across SNPs, which is reasonable when each genetic effect is small. See Section S5.3 of the Supplement for details.  The condition that $\Sigma$ is positive definite holds in the absence of perfect collinearity among the exposure variables.

Throughout the paper, $\bgamma_j, \Gamma_j, \Sigma_{Xj}$, and $\sigma_{Yj}$, for  $j=1,\dots,p$, are sequences of real numbers or matrices that can depend on $n$ in order to model many weak IVs, but we omit the $n$ index for ease of notation. For example, many IVs with small effects can be modeled by letting $p\rightarrow\infty$ while shrinking the $\bgamma_j$ vectors toward zero as $n$ increases. We write $a = \Theta(b)$ if there exists a constant $c > 0$ such that $c^{-1} b \leq |a| \leq c b$. When only common variants (i.e., those with minor allele frequencies greater than 0.05 \citep{Gibson:2012aa, Cirulli:2010aa}) are used as IVs, $\var (Z_j) = \Theta(1)$, and each entry of $\Sigma_{Xj}$ and $ \sigma_{Yj}^2$ for $j=1,\dots, p$ is $\Theta (1/n)$ (see Section S3.1 of the Supplement for details). When rare variants (i.e., those with minor allele frequencies less than 0.05) are used as IVs, $\var (Z_j)$ can be very small, and the entries of $\Sigma_{Xj}$ and $ \sigma_{Yj}^2$ may converge to zero at a rate slower than $1/n$. 

\section{A general asymptotic regime for many weak IVs}
\label{sec: asymptotics}
\subsection{Asymptotic regime}

Our proposed asymptotic regime, stated in Assumption \ref{assump: 2}, is motivated by the insight that identifying the full vector $\bbeta_0$ is equivalent to identifying all linear
combinations of $\bbeta_0$. The latter requires sufficient IV strength for all linear combinations of the
exposures, and Assumption \ref{assump: 2} uses a particular rotation of the $K$-dimensional space of exposures to capture a complete, multi-dimensional summary of IV strength.

\begin{assumption}\label{assump: 2}
    There exists a non-random  $K\times K$ matrix $S_n = \tilde S_n {\mathrm {diag}}(\sqrt{\mu_{n1}}, \dots, \sqrt{\mu_{nK}})$ that satisfies the following conditions:
    (i) $\tilde S_n$ is bounded element-wise, and the smallest eigenvalue of $\tilde S_n \tilde S_n^T$ is bounded away from 0; and
    (ii) $S_n^{-1}\big(\sum_{j=1}^{p} \bgamma_j \bgamma_j^T \sigma_{Yj}^{-2}\big)S_n^{-T}$ is bounded element-wise, and its smallest eigenvalue is bounded away from 0.
\end{assumption}

The $S_n$ defined in Assumption 2 always exists; see Section S3.2 of the Supplement. The term $S_n^{-1}\big(\sum_{j=1}^{p} \bgamma_j \bgamma_j^T \sigma_{Yj}^{-2}\big)S_n^{-T}$ can be rewritten as 
$$
{\mathrm {diag}}(\mu_{n1}^{-\frac{1}{2}}, \dots, \mu_{nK}^{-\frac{1}{2}} )
 \bigg(\sum_{j=1}^{p} \tilde S_n^{-1} \bgamma_j \bgamma_j^T \tilde S_n^{-T} \sigma_{Yj}^{-2}\bigg){\mathrm {diag}}(\mu_{n1}^{-\frac{1}{2}}, \dots, \mu_{nK}^{-\frac{1}{2}} ), 
$$
where  $\tilde S_n^{-1} \bgamma_j = \tilde S_n^{-1} \var(Z_j)^{-1} \cov (\bmX, Z_j)  = \var(Z_j)^{-1} \cov (\tilde S_n^{-1} \bmX, Z_j) $ represents the marginal association between $Z_j$ and the rotated exposures $\tilde S_n^{-1} \bmX$. Then, the matrix $\sum_{j=1}^{p} \tilde S_n^{-1} \bgamma_j \bgamma_j^T \tilde S_n^{-T} \sigma_{Yj}^{-2}$ measures the IV strength for the rotated exposures $\tilde S_n^{-1} \bmX$, with strength rates $ \mu_{n1}, \dots, \mu_{nK}$. The vector $\mu_{n1},\dots,\mu_{nK}$  represents a complete, multi-dimensional summary of overall IV strength, and the scalar $\mu_{n,\min} = \min (\mu_{n1},\dots, \mu_{nK})$ represents the slowest rate among those linear combinations of exposures. Also, Section S3.5 of the Supplement shows that the minimum eigenvalue of the matrix $\sum_{j=1}^{p} \bgamma_j \bgamma_j^T \sigma_{Yj}^{-2}$ is $\Theta(\mu_{n,\min})$. Notably, the theorems we develop later (Theorem \ref{theo: ivw}-\ref{theo: divw}) will require  $\mu_{n,\min}$ to satisfy certain conditions because we are interested in estimating the entire vector $\bbeta_0$, which, by the Cramer-Wold device, is equivalent to accurately estimating the effects of all linear combinations of the exposures. Consequently, we consider $\mu_{n,\min}$ as a scalar measure of overall IV strength for summary-data MVMR. 

Assumption \ref{assump: 2} accommodates various weak IV scenarios in MVMR. Here, we highlight one example:
\begin{example}
Consider a simple scenario with two exposures, $X_1$ and $X_2$. For $j=1,\dots, p$, let $\gamma_{j1}$ and $\gamma_{j2}$ denote the marginal associations of SNP $j$ with $X_1$ and $X_2$, respectively. 
Suppose $\gamma_{j2} = \gamma_{j1} + c/\sqrt{n}$, where $c$ is a constant, $n=n_Y$, and $\gamma_{j1} \sim N(0, 1/p)$. Then,   $\sigma_{Yj}^2$ is approximately 
$ \var (Y)/n$, and the marginal IV strength for $X_1$ and $X_2$, measured by $\sum_{j=1}^{p}\gamma_{j1}^2\sigma_{Yj}^{-2} $ and $ \sum_{j=1}^{p}\gamma_{j2}^2\sigma_{Yj}^{-2}$, is on the order of $\Theta(n)$. 
However, since $\gamma_{j1}$ and $\gamma_{j2}$ are very similar for all $j$, 
the conditional IV strength for one exposure given the other is much weaker, and the conditional $F$-statistics for both $X_1$ and $X_2$ are $\Theta(1)$. Importantly, Assumption \ref{assump: 2} includes this scenario. By letting $\tilde S_n^{-1}$ be
$\begin{pmatrix}
-1 & 1\\
1 & 0
\end{pmatrix}$, where each row representing a linear combination of the exposures, 
and defining $\mu_{n1} = p, \mu_{n2} = n$ to represent the rates at which the IV strength for $X_2-X_1$ and $X_1$ grows, we have
\begin{align*}
    S_n^{-1} \bigg\{ \sum_{j=1}^{p} \bgamma_j \bgamma_j^T \sigma_{Yj}^{-2}\bigg\} S_n^{-T} 
 & = \frac{n}{\var(Y)}  \sum_{j=1}^{p} 
        \begin{pmatrix}
        \frac{\gamma_{j2}-\gamma_{j1}}{\sqrt{\mu_{n1}}}\\
        \frac{\gamma_{j1}}{\sqrt{\mu_{n2}}}
        \end{pmatrix}        \begin{pmatrix}
         \frac{\gamma_{j2}-\gamma_{j1}}{\sqrt{\mu_{n1}}} & 
        \frac{\gamma_{j1}}{\sqrt{\mu_{n2}}}
        \end{pmatrix} =
        \frac{1}{\var(Y)}
    \begin{pmatrix}
        \frac{c^2p }{\mu_{n1}} & \frac{c \sqrt{n} \sum_{j=1}^{p} \gamma_{j1}}{\sqrt{\mu_{n1}\mu_{n2}}} \\
        \frac{c\sqrt{n}  \sum_{j=1}^{p} \gamma_{j1}}{\sqrt{\mu_{n1}\mu_{n2}}} & \frac{n \sum_{j=1}^{p}\gamma_{j1}^2}{\mu_{n2}}
    \end{pmatrix}.
\end{align*}
This matrix is element-wise bounded, with its smallest eigenvalue bounded away from 0, satisfying Assumption \ref{assump: 2}; see Section S3.3 of the Supplement for more details.
\end{example}

\subsection{Connections to the existing literature}

The matrix $\sum_{j=1}^{p} \bgamma_j \bgamma_j^T \sigma_{Yj}^{-2}$ is related to the conditional $F$-statistic of \cite{sanderson2021testing}, which is commonly used in summary-data MVMR to test for weak instruments. Specifically, consider the oracle version of the conditional $F$-statistic for the $k$-th exposure, denoted as $F_k$:
\begin{align*}
    F_k = \frac{1}{p-(K-1)}\sum_{j=1}^p \frac{\delta_k^T\bgamma_j\bgamma_j^T\delta_k}{\delta_k^T\Sigma_{Xj}\delta_k}.
\end{align*}
The term $\delta_k$ is a vector in $\mathbb{R}^K$ whose $k$-th entry is -1. Its remaining entries are given by $\tilde \delta_{-k} = \argmin_{\delta_{-k} \in \mathbb{R}^{K-1}} \sum_{j=1}^p \sigma_{Xjk}^{-2} (\gamma_{jk} - \gamma_{j(-k)}^T \delta_{-k})^2$, where $\gamma_{j(-k)}$ denotes the vector
 $\bgamma_j$ with the $k$-th entry excluded. $ \tilde  \delta_{-k}$ is the coefficient from the weighted linear projection of the vector $(\gamma_{1k},\dots,\gamma_{pk})^T$ onto the column space of the matrix with rows $\gamma_{j(-k)}^T$, for $j=1,\dots,p$. Under Assumption 1(iii), we have $\delta_k^T \Sigma_{Xj} \delta_k = \Theta(\sigma_{Yj}^{2})$, and hence $p F_k = \Theta\big(\delta_k^T (\sum_{j=1}^p \bgamma_j \bgamma_j^T \sigma_{Yj}^{-2})\delta_k\big)$. In other words, the conditional $F$-statistic measures the IV strength for a particular linear combination of exposures along the direction defined by $\delta_k$. In Section S3.3 of the Supplement, we show that $\min_k F_k = \Theta(\mu_{n,\min}/p)$, meaning that the smallest conditional $F$-statistic across $K$ exposures has the same rate as $\mu_{n,\min}/p$. However, as demonstrated in Example 2, a large conditional $F$-statistic for an exposure may still be insufficient to ensure reliable inference for that exposure.

\begin{example}
    Consider three exposures, $X_1$, $X_2$, and $X_3$. For $j = 1,\dots,p$, let $\gamma_{j2} = \gamma_{j1} + c/\sqrt{n}$, $\gamma_{j1} \sim N(0, 1/p)$, and $\gamma_{j3} = 0.5\gamma_{j1} + 0.5\gamma_{j2} + \epsilon_j$, where $\epsilon_j \sim N(0, 1/p)$. All other notation remains the same as in Example 1. In this example, the marginal IV strength of the three exposures is on the order of $\Theta(n)$. However, due to the high correlation between $\gamma_{j1}$ and $\gamma_{j2}$, 
    the conditional IV strength for $X_1$ and $X_2$ is much weaker. In Section S3.3 of the Supplement, we show  that Assumption 2 holds for this example with $\mu_{n,\min}/p = \Theta(1)$, and that the conditional $F$-statistics for $X_1$, $X_2$, and $X_3$ are  $\Theta(1)$, $\Theta(1)$, and $\Theta(n/p)$, respectively. Thus, the conditional $F$-statistic for $X_3$ grows to infinity when $n/p \to \infty$. However, according to Corollary 1 (presented later in Section \ref{sec: IVW}), the coverage probability for the true causal effect of $X_3$ using the MV-IVW method can still approach 0, because $\mu_{n,\min}$ does not grow sufficiently fast; see Section S3.3 of the Supplement for more details.
\end{example} 

 Additionally, we note that calculating conditional $F$-statistics in practice requires the consistent estimation of the $\delta_k$'s, which may not be achievable when IV strength is very weak. However, as we demonstrate in the next subsection, estimators of $\mu_{n,\min}$ can be easily computed without the need to estimate the $\delta_k$ vectors.

Finally, we remark that the \emph{concentration parameter} (the noncentrality parameter of the first-stage $F$-statistic)  is commonly used to quantify IV strength based on individual-level data
 \citep{Stock:2002aa, Chao:2005aa}. As shown in Section S3.7 of the Supplement,  $\sum_{j=1}^{p} \bgamma_j \bgamma_j^T \sigma_{Yj}^{-2}$ has the same rate as the concentration parameter (denoted as $H_n$), in the sense that $S_n^{-1} H_n S_n^{-T}$ is also bounded, and its smallest eigenvalue is bounded away from 0 almost surely. This result suggests that $\mu_{n,\min}$ can be interpreted as the rate at which the minimum eigenvalue of the concentration parameter grows as $n \rightarrow \infty$. The minimum eigenvalue of the concentration parameter is also known as the Cragg-Donald statistic and has been used for assessing under-identification and weak IVs in IV regressions \citep{Stock2001:ivtest}.

\subsection{Measuring IV strength in practice}

In practical applications, however, the matrix $\sum_{j=1}^{p} \bgamma_j \bgamma_j^T \sigma_{Yj}^{-2}$ may not be ideal for assessing weak IVs for two reasons. First, it is not scale-invariant, meaning that it changes if the units of the exposures are changed. Second, it depends on the outcome. In Sections S3.4-S3.5 of the Supplement, we show that under mild assumptions, the scale-invariant matrix $\sum_{j=1}^{p} \Omega_{j}^{-1} \bgamma_j \bgamma_j^T \Omega_{j}^{-T}$ is asymptotically equivalent to $\sum_{j=1}^{p} \bgamma_j \bgamma_j^T \sigma_{Yj}^{-2}$, where $\Omega_{j} = {\mathrm {diag}}(\sigma_{Xj1},\dots,\sigma_{XjK})\Sigma^{\frac{1}{2}}$ and $A^{\frac{1}{2}}$ is the symmetric square root of a positive semi-definite matrix $A$.   The matrix $\sum_{j=1}^{p} \Omega_j^{-1} \bgamma_j \bgamma_j^T \Omega_j^{-T}$, referred to as the \emph{IV strength matrix} in the remainder of this paper, is a multivariable extension of the IV strength measure used in the univariable MR \citep{zhao2018statistical, Ye:2019}.  Notably, the minimum eigenvalue of 
$\sum_{j=1}^{p} \Omega_j^{-1} \bgamma_j \bgamma_j^T \Omega_j^{-T}$ is also $\Theta(\mu_{n,\min})$. Moreover, 
we can estimate $\sum_{j=1}^{p} \Omega_j^{-1} \bgamma_j \bgamma_j^T \Omega_j^{-T}$ unbiasedly using its sample version, $\sum_{j=1}^{p} \Omega_j^{-1} \hat \bgamma_j \hat \bgamma_j^T \Omega_j^{-T} - p I_K$, where  $I_{K}$ is a $K\times K$ identity matrix.  This allows us to use $\hat \lambda_{\min}$, the minimum eigenvalue of this sample IV strength matrix, to evaluate $\mu_{n,\min}$ and quantify the IV strength.

\section{IVW method in MVMR}
\label{sec: IVW}

Next, we theoretically study the widely used MV-IVW method using the asymptotic regime proposed in Section \ref{sec: asymptotics} and show that it tends to produce biased estimates and misleadingly narrow confidence intervals centered around these biased estimates, which may lead to spurious findings.

The MV-IVW estimator, proposed by \cite{Burgess2015mvmr}, is given by:
$
	\hat\bbeta_{\rm IVW}  = \big\{ \sum_{j=1}^{p} \hat\bgamma_j \hat\bgamma_j^T  \sigma_{Yj}^{-2} \big\}^{-1} \big\{ \sum_{j=1}^{p} \hat\bgamma_j \hat\Gamma_j \sigma_{Yj}^{-2} \big\}. $
Theorem \ref{theo: ivw} below describes the theoretical properties of the IVW estimator (with proof in Section S3.9 of the Supplement). In the following, $\xrightarrow{P}$ denotes convergence in probability, and $\xrightarrow{D}$ denotes convergence in distribution.

\begin{theorem} \label{theo: ivw}
    Under Assumptions \ref{assump: 1}-\ref{assump: 2}, and $\max_j (\gamma_{jk}^2{\sigma_{Xjk}^{-2}} ) /(\mu_{n,\min} + p) \rightarrow 0$ for every $k = 1,\dots, K$, as $n\to\infty$,
    \begin{equation}
        \mathbb{V}^{-\frac{1}{2}}\bigg\{ \sum_{j=1}^{p}(M_j+V_j) \bigg\}( \hat \bbeta_{\rm IVW} - \tilde{\bbeta}_{\rm  IVW}) \xrightarrow{D} N(\bm0, I_{K}),\label{eq: mvmr normal}
    \end{equation} 
    where $\tilde{\bbeta}_{\rm  IVW}= \left\{ \sum_{j=1}^{p} \hat M_j \right\}^{-1} \left\{ \sum_{j=1}^{p} (\hat M_j- V_j) \right\}  \bm{\beta}_0$,   $M_j = \bgamma_{j}\bgamma_{j}^T\sigma_{Yj}^{-2}$, $\hat M_j = \hat \bgamma_{j}\hat \bgamma_{j}^T\sigma_{Yj}^{-2}$, $V_j = \Sigma_{Xj}\sigma_{Yj}^{-2}$, and $\mathbb{V} = \sum_{j=1}^{p}\big\{  (1+ \bm{\beta}_0^TV_j\bm{\beta}_0) (M_j+V_j)+V_j\bm{\beta}_0\bm{\beta}_0^TV_j\big\}$.
\end{theorem}

 The condition $\max_j (\gamma_{jk}^2{\sigma_{Xjk}^{-2}} ) /(\mu_{n,\min} + p) \rightarrow 0$ for all $k$ is a mild requirement used to verify the Lindeberg's condition for the central limit theorem, and is satisfied when no single SNP dominates the IV strength. Unless $\bbeta_0=\bm 0$, Theorem \ref{theo: ivw} indicates that the MV-IVW estimator is not well-centered around the true causal effect $\bbeta_0$, but instead around a random target $\tilde {\bbeta}_{\rm IVW}$. Consequently, the asymptotic bias of $\hat \bbeta_{\rm IVW}$ is the probability limit of  $\tilde {\bbeta}_{\rm IVW} - \bbeta_0$ as $n \rightarrow \infty$. The extent of the bias depends on
whether the weighted sum of the estimated SNP-exposure effects, $ \sum_{j=1}^{p} \hat M_j$, is sufficiently large relative to its variance, $\sum_{j=1}^{p} V_j $, and whether $\bbeta_0 = \bm 0$. This is formally established in Proposition \ref{prop: ivw}.
\begin{proposition} \label{prop: ivw}
    Under Assumptions \ref{assump: 1}-\ref{assump: 2}, when $\bbeta_0 \neq \bm 0$, as $n\rightarrow \infty$, (i) $\tilde \bbeta_{\rm IVW} - \big\{ \sum_{j=1}^{p}(M_j+V_j) \big\}^{-1} \big\{ \sum_{j=1}^{p}M_j \big\} \bm{\beta}_0 \xrightarrow[]{P} \bm 0$ if $\mu_{n,\min}/p \rightarrow [0,\infty)$; (ii) $\tilde \bbeta_{\rm IVW} - \bbeta_0 \xrightarrow[]{P} \bm 0$ if $\mu_{n,\min}/p \rightarrow \infty$; and (iii) $\mathbb{V}^{-\frac{1}{2}}\big\{ \sum_{j=1}^{p}(M_j+V_j) \big\}( \tilde \bbeta_{\rm IVW} - \bbeta_0) \xrightarrow{P} \bm 0$ if $\mu_{n,\min}/p^2 \rightarrow \infty$. 
\end{proposition}
When $\bbeta_0 \neq \bm 0$, Proposition 1 indicates that the asymptotic bias of the IVW estimator is nonzero unless $\mu_{n, \min}/p \rightarrow \infty$, and remains non-negligible compared to the asymptotic standard deviation unless a stronger condition $\mu_{n, \min}/p^2 \rightarrow \infty$ holds. The proofs can be found in Section S3.9 of the Supplement. These rate conditions on $\mu_{n, \min}$ are generally unrealistic in typical MVMR studies, except in cases involving a small number of exposures with weak genetic correlations. In most MR applications, a more realistic condition is $\mu_{n,\min}/p \rightarrow [0,\infty)$, under which the IVW estimator is asymptotically biased and converges to $\big\{ \sum_{j=1}^{p}(M_j+V_j) \big\}^{-1} \big\{ \sum_{j=1}^{p}M_j \big\} \bm{\beta}_0 $ as $n\to\infty$.

Unlike in two-sample univariable MR, where the IVW estimator is always biased toward zero in the presence of weak IVs, Theorem \ref{theo: ivw} and Proposition 1 suggest that the MV-IVW estimator is more adversely affected by weak IVs. First, the asymptotic bias of MV-IVW exists unless either $\bbeta_0 = \bm 0$ or $\mu_{n,\min}/p\rightarrow\infty$. This implies that $\hat \bbeta_{\rm IVW}$ remains biased when testing the composite null hypothesis that one or some exposures have no causal effects, and is only unbiased when testing the global null hypothesis that \emph{all} exposures have no causal effect, which is of less practical interest. Second, certain elements of $\big\{ \sum_{j=1}^{p}(M_j+V_j) \big\}^{-1} \big\{ \sum_{j=1}^{p}M_j \big\} \bm{\beta}_0$ may be amplified, while others are attenuated. This explains why the IVW estimator in MVMR can exhibit bias either toward or away from zero \citep{sanderson2021testing,sanderson2022estimation}; see a theoretical example in Section S3.9 of the Supplement and simulation results in Table S3.

Furthermore, Corollary \ref{coro: ivw} shows that the MV-IVW estimator tends to produce misleadingly narrow confidence intervals, which can result in spurious findings.
\begin{corollary} \label{coro: ivw}
    Under Assumptions 1-2, unless $\bbeta_0 = \bm 0$ or $\mu_{n,\min}/p^2 \rightarrow \infty$, $\Pr(\exists k \in \{1,\dots,K\} \text{ such that }\frac{|\hat \beta_{{\rm IVW},k} - \beta_{0k}|}{\sqrt{\mathcal{V}_{{\rm IVW},kk}}} > C) \rightarrow 1$ as $n \rightarrow \infty$, where $C$ is any finite constant and  $\mathcal{V}_{{\rm IVW},kk}$ is the $k$th diagonal element of the asymptotic variance matrix   ${\mathcal{V}}_{\rm IVW}  =  \big\{ \sum_{j=1}^{p}(M_j+V_j) \big\}^{-1}  \mathbb{V}\  \big\{ \sum_{j=1}^{p}(M_j+V_j) \big\}^{-1}$.
\end{corollary}
As stated in Corollary 1, even in the oracle case where the true variance is known, the $(1-\alpha)$ confidence interval $(\hat{\beta}_{{\rm IVW}, k} \pm z_{1-\alpha/2}\sqrt{\mathcal{V}_{{\rm IVW},kk}})$ may not cover the true causal effect $\beta_{0k}$ with probability approaching 1 for a given significance level $\alpha$. Note that this is not due to multiplicity in hypothesis testing, but rather because the bias of $\hat \beta_{{\rm IVW},k}$ is non-negligible relative to its standard deviation. This is particularly concerning when $\beta_{0k}=0$, i.e., the $k$th exposure has no causal effect, as the confidence interval may fail to include 0 with high probability, potentially inflating the type I error rate and leading to spurious conclusions. Details of the derivations are provided in Section S3.9 of the Supplement.

\section{Spectral-regularized IVW estimator in MVMR}
\label{sec: dIVW}

\subsection{IVW with spectral regularization}

In this section, we propose a spectral-regularized inverse-variance weighted (SRIVW) estimator, which retains a closed-form expression similar to the MV-IVW estimator but achieves improved performance and remains consistent and asymptotically normal in the presence of many weak IVs. 

As shown in Theorem \ref{theo: ivw}, the MV-IVW estimator fails due to non-negligible measurement errors when $\sum_{j=1}^{p} \hat\bgamma_j \hat\bgamma_j^T \sigma_{Yj}^{-2}$ is used to approximate $\sum_{j=1}^{p} \bgamma_j \bgamma_j^T \sigma_{Yj}^{-2}$ in the presence of many weak IVs. A natural solution is to replace $\sum_{j=1}^{p} \hat\bgamma_j \hat\bgamma_j^T \sigma_{Yj}^{-2}$ in the MV-IVW formulation with an unbiased estimator of $\sum_{j=1}^{p} \bgamma_j \bgamma_j^T \sigma_{Yj}^{-2}$, such as $\sum_{j=1}^{p} (\hat\bgamma_j \hat\bgamma_j^T -\Sigma_{Xj}) \sigma_{Yj}^{-2}$, also denoted as $\sum_{j=1}^{p} \hat M_j - V_j$.

However, in our preliminary investigation, we observed that this simple modification can become highly unstable when IVs are very weak.
The instability arises from inverting the matrix $\sum_{j=1}^{p} \hat M_j - V_j $, which can become nearly singular when SNPs are weakly associated with one or more exposures conditional on the others, resulting in unstable estimates. To stabilize the inversion, we propose the following spectral regularization approach:
\begin{align}\label{eq: adjustment}
R_{\phi}\bigg(\sum_{j=1}^{p} \hat M_j - V_j\bigg) = \sum_{j=1}^{p} \hat M_j - V_j + \phi \bigg(\sum_{j=1}^{p} \hat M_j - V_j\bigg)^{-1}, 
\end{align}
assuming no eigenvalues are exactly zero (an event with probability zero), where  $R_{\phi}$ denotes a matrix operator for regularization and $\phi\geq 0$ is a tuning parameter that may depend on the observed data. With this regularization, our SRIVW estimator takes the following form: 
\begin{align}\label{eq: SRIVW}
    \hat\bbeta_{{\rm SRIVW}, \phi} &= \bigg\{ R_{\phi}\bigg(\sum_{j=1}^{p} \hat M_j - V_j\bigg) \bigg\}^{-1}\bigg\{\sum_{j=1}^{p} \hat \bgamma_j \hat \Gamma_j \sigma_{Yj}^{-2} \bigg\}. 
\end{align}

The spectral regularization is motivated by the spectral decomposition of $    \sum_{j=1}^{p}  \hat M_j - V_j = \hat U\hat \Lambda \hat U^T $, 
where $\hat \Lambda$ is a diagonal matrix with eigenvalues $\hat \lambda_1,\dots,\hat \lambda_K$, and $\hat U$ is the matrix of eigenvectors. This leads to the following expression:
\begin{align}
R_{\phi}\bigg(\sum_{j=1}^{p} \hat M_j - V_j\bigg)= \sum_{j=1}^{p}  \hat M_j - V_j + \phi\bigg(\sum_{j=1}^{p}  \hat M_j - V_j\bigg)^{-1} = \hat U(\hat \Lambda + \phi \hat \Lambda^{-1})\hat U^T , \nonumber
\end{align} 
where $\hat \Lambda + \phi \hat \Lambda^{-1}$ is a diagonal matrix with entries  $\hat \lambda_1 + \phi\hat\lambda_1^{-1}, \dots, \hat \lambda_K +{\phi}\hat\lambda_K^{-1}$. By appropriately selecting $\phi$, the regularization can damp the explosive effect of the inversion of small eigenvalues, and this damping effect adapts based on eigenvalue magnitudes, with small eigenvalues affected most and large eigenvalues nearly unaffected. In any case, the regularization ensures that all eigenvalues of the matrix $\sum_{j=1}^{p}  \hat M_j - V_j + \phi(\sum_{j=1}^{p}  \hat M_j - V_j)^{-1}$ have absolute values bounded below by  $2\sqrt{\phi}$. Therefore, 
by appropriately choosing $\phi$, the SRIVW estimator is expected to have reduced bias and improved stability. 
 We remark that this type of regularization is related to the Tikhonov regularization used in solving integral equations \citep{kress1989linear}. A similar regularization was considered by \cite{van2016ridge} in the context of estimating high-dimensional precision matrices, and our regularization can be viewed as a first-order Taylor expansion of their estimator; see Section S3.8 of the Supplement for details.

\subsection{Theoretical properties of the SRIVW estimator}

In Theorem \ref{theo: divw}, we provide a sufficient condition for selecting $\phi$ and establish the consistency and asymptotic normality of the SRIVW estimator.

\begin{theorem} \label{theo: divw}
Assume Assumptions 1-2, ${\mu_{n,\min}}/{\sqrt{p}}\rightarrow \infty$, and $\max_j \gamma_{jk}^2\sigma_{Xjk}^{-2}/(\mu_{n,\min} + p) \rightarrow 0$ for all $k$ as $n\to\infty$. If $\phi = O_P(\mu_{n,\min}+p)$, where $a_n = O_P (b_n)$ means $a_n/b_n$ is bounded in probability,
    then $\hat{\bbeta}_{\rm SRIVW, \phi} $ in \eqref{eq: SRIVW} is consistent and asymptotically normal, i.e., as $n\to \infty$,
\begin{align}
    \mathbb{V}^{-\frac{1}{2}}\bigg\{ \sum_{j=1}^{p}M_j \bigg\}( \hat{\bbeta}_{{\rm SRIVW}, \phi} - \bbeta_0) \xrightarrow[]{d} N(\bm0, I_{K}), \label{eq: dmvmr normal}
\end{align} 
where $M_j = \bgamma_{j}\bgamma_{j}^T\sigma_{Yj}^{-2}$, $V_j = \Sigma_{Xj}\sigma_{Yj}^{-2}$,   and $\mathbb{V} = \sum_{j=1}^{p}\big\{  (1+ \bm{\beta}_0^TV_j\bm{\beta}_0) (M_j+V_j)+V_j\bm{\beta}_0\bm{\beta}_0^TV_j\big\}$. Hence, the asymptotic variance matrix of $\hat{\bbeta}_{\rm SRIVW, \phi} $ is 
$
        {\mathcal{V}}_{\rm SRIVW}  =  \big\{ \sum_{j=1}^{p}M_j \big\}^{-1}  \mathbb{V}\  \big\{ \sum_{j=1}^{p}M_j \big\}^{-1} .$
\end{theorem}
The proof can be found in Section S3.10 of the Supplement. In addition to Assumptions 1 and 2, Theorem \ref{theo: divw} assumes three additional conditions. First, ${\mu_{n,\min}}/{\sqrt{p}} \rightarrow \infty$ is a condition on the minimum IV strength, which is much weaker than the condition for the asymptotic normality of MV-IVW. Second, 
 $\max_j \gamma_{jk}^2\sigma_{Xjk}^{-2}/(\mu_{n,\min} + p) \rightarrow 0$ for all $k$ is the same assumption as in Theorem \ref{theo: ivw}. Third, $\phi = O_P(\mu_{n,\min}+p)$ is a rate condition on the tuning parameter $\phi$, which is allowed to be data-dependent. The selection of this tuning parameter will be discussed later in Section \ref{subsec: tuning}.

 Theorem \ref{theo: divw} ensures that as long as the overall IV strength is not too weak, specifically in the sense that ${\mu_{n,\min}}/{\sqrt{p}} \to \infty$, the SRIVW estimator has the stated theoretical guarantees.  
Notably, Theorem \ref{theo: divw} is established under the general asymptotic regime proposed in Section \ref{sec: asymptotics}, which allows full flexibility in the IV strength of different exposures. For example, it accommodates weak IVs arising either from weak marginal associations with each exposure or from strong marginal but weak conditional associations with each exposure.

As shown in Section S3.10 of the Supplement, 
the result in \eqref{eq: dmvmr normal} still holds if we replace $\mathbb{V}$ and $
\sum_j M_j$ with $\hat{\mathbb{V}}_\phi = \sum_{j=1}^p\{(1+\hat{\bm\beta}_{{\rm SRIVW}, \phi}^TV_j\hat{\bm\beta}_{{\rm SRIVW}, \phi})\hat{M_j}+V_j\hat{\bm\beta}_{{\rm SRIVW}, \phi}\hat{\bm\beta}_{{\rm SRIVW}, \phi}^TV_j\}$ and $R_{\phi}(\sum_{j=1}^{p}\hat M_j - V_j)$, respectively. This leads to the following consistent variance estimator for the SRIVW estimator: 
\begin{align}\label{SE: adivw}
    \hat{\mathcal{V}}_{{\rm SRIVW}, \phi} = & \left\{R_{\phi}\bigg(\sum_{j=1}^{p}\hat M_j - V_j\bigg)\right\}^{-1} \hat{\mathbb{V}}_\phi \left\{R_{\phi}\bigg(\sum_{j=1}^{p}\hat M_j - V_j\bigg)\right\}^{-1} .
\end{align}

\subsection{Selection of tuning parameter $\phi$}
\label{subsec: tuning}
Our asymptotic analysis in  Theorem \ref{theo: divw} requires that the tuning parameter satisfies $\phi = O_P(\mu_{n,\min}+p)$. In practice, 
to choose a $\phi$ satisfying this rate condition, we propose minimizing the following objective function (also known as \textit{Q}-statistic):
\begin{align}\label{J obj: no pleiotropy}
    Q(\phi) = \sum_{j=1}^{p} \frac{(\hat \Gamma_j - \hat \bgamma_j^T\hat\bbeta_{{\rm SRIVW}, \phi})^2}{\sigma_{Yj}^2 + \hat\bbeta_{{\rm SRIVW}, \phi}^T \Sigma_{Xj} \hat\bbeta_{{\rm SRIVW}, \phi}},\ \text{ subject to $\phi \in B$},
\end{align}
where $\hat\bbeta_{{\rm SRIVW}, \phi}$ is defined in \eqref{eq: SRIVW} and $B$ represents a pre-specified interval in which every $\phi\in B$ satisfies $\phi = O_P(\mu_{n,\min}+p)$. The choice of $B$ can influence the performance of the estimator; a suitable choice of $B$ will be discussed below. In \cite{sanderson2021testing} and GRAPPLE \citep{Wang2021grapple}, the objective function $\sum_{j=1}^{p} \{(\hat \Gamma_j - \hat \bgamma_j^T\bbeta)^2/(\sigma_{Yj}^2 + \bbeta^T \Sigma_{Xj} \bbeta)\}$ is minimized with respect to the entire vector $\bbeta$, but the solution has no closed-form expression. Moreover, this minimization process can be computationally intensive and numerically unstable, especially when dealing with more than a few exposures and weak IVs. In contrast, minimizing $Q(\phi)$ with respect to a single parameter $\phi$ is computationally more tractable. The minimizer of $Q(\phi)$ falls within the class of SRIVW estimators, which offers valid statistical inference under the conditions of Theorem \ref{theo: divw}. As shown in our simulation studies, the proposed estimator performs comparably to or even better than GRAPPLE and other state-of-the-art methods.

Now we discuss our choice of the interval $B$. As previously mentioned in Section \ref{sec: notation and setup}, 
we can use $\hat \lambda_{\min}/\sqrt{p}$ to approximate $\mu_{n, \min}/\sqrt{p}$, where $\hat \lambda_{\min}$ is
the smallest eigenvalue of the sample IV strength matrix $\sum_{j=1}^{p} \Omega_j^{-1} \hat \bgamma_j \hat \bgamma_j^T \Omega_j^{-T} - p I_K$.  Our choice of $B$ is $[0, e^{c - \hat \lambda_{\min}/\sqrt{p}}]$, where $c$ is a positive constant. The rationale behind this upper bound is as follows. First, it is clear that this upper bound approaches 0 in probability as $\mu_{n,\min}/\sqrt{p} \rightarrow \infty$ and thus satisfies the rate condition $O_P(\mu_{n,\min}+p)$. Second, when $\hat \lambda_{\min}/\sqrt{p} < c$, the upper bound of $B$ increases exponentially. This ensures that when IV strength is low, a wide range of $\phi$ values is considered, enabling a data-driven selection. Based on our observation that the SRIVW estimator is nearly unbiased when $\hat\lambda_{\min}/\sqrt{p} > 17$, even without any regularization, we set $c=17$ in our implementation  and minimize $Q(\phi)$  over a grid of values from $B$, specifically $0 \cup \{e^{i - \hat \lambda_{\min}/\sqrt{p}} : i=0,0.5, \dots, 17\}$.

\subsection{Extensions}
\label{subsec: bhp}

Although MVMR can adjust for known pleiotropic pathways by including additional exposure variables, it does not eliminate the possibility of residual pleiotropic effects. To address this, we extend the SRIVW estimator to accommodate a specific form of pleiotropy in MR, known as balanced horizontal pleiotropy \citep{Bowden:2017aa, Hemani:2018aa, zhao2018statistical}. Specifically, we generalize the outcome model in \eqref{eq: out} to 
 $ Y = \ \bm X^T \bm\beta_0 + \bmZ^T \bm\alpha + g(\bm U,E_Y)$, where $\bm\alpha= (\alpha_1,\dots, \alpha_p)^T$ denotes the vector of unknown direct effects of the SNPs on the outcome, and each $\alpha_j$ is assumed to follow a mean-zero normal distribution with unknown variance. 
  In Section S2, we show that the SRIVW estimator is still consistent and asymptotically normal under balanced horizontal pleiotropy. However, it generally has larger variance due to the additional uncertainty introduced by the random effects $\alpha_j$. We refer to the SRIVW estimator, along with its adjusted variance estimator under this setting, as the SRIVW-pleio method. In practice, SRIVW-pleio and the standard SRIVW produce nearly identical results when there is no pleiotropy. However, in the presence of pleiotropy, SRIVW-pleio may select a different tuning parameter by minimizing a modified \textit{Q}-statistic that accounts for pleiotropic effects, potentially leading to different point estimates and standard errors. See Section S2 in the Supplement for more details.

Finally, the SRIVW estimator can also be extended to accommodate overlapping exposure and outcome datasets; see Section S6 of the Supplement for details.

\section{Simulations}
\label{sec: simu}

To evaluate the finite-sample performance of our proposed estimator in the presence of many weak IVs, we conduct a simulation study using summary statistics from the \textsf{rawdat\_mvmr} dataset, available in the \textsf{MVMR} package \citep{sanderson2021testing} in \textsf{R}. The dataset includes  $p=145$ uncorrelated SNPs. The three exposures---low-density lipoprotein cholesterol, high-density lipoprotein cholesterol, and triglycerides---are obtained from the Global Lipids Genetics Consortium\nocite{global2013discovery}. The observed effects and estimated variances from these data are used as the true SNP-exposure associations $\bgamma_j$ and variances $\sigma_{Xjk}^2$ for $j=1,\dots, p$ and $k=1,2,3$. The outcome is systolic blood pressure from the UK Biobank, with its estimated variances used as $\sigma_{Yj}^2$ for $j=1,\dots, p$. To generate the true $\Gamma_j$'s, we let $\Gamma_j = \bgamma_j^T\bbeta_0$ for $j = 1,\dots,p$, where $\bbeta_0 = (\beta_{01},\beta_{02},\beta_{03})$ denotes the vector of true causal effects. The shared correlation matrix $\Sigma$ for the estimated SNP-exposure associations also comes from the \textsf{MVMR} package and is used to compute $\Sigma_{Xj}$ for all $j$. Further details are provided in Section S4.1 of the Supplement.

Our simulation study resembles a factorial design, with the following factors:
\begin{enumerate}
    \item {Factor 1:} True causal effects. (i) $\bbeta_0 = (0.8, 0.4, 0)$; (ii) $\bbeta_0 = (0.1, -0.5, -0.9)$.
    \item  {Factor 2:}  Relative IV strength across exposures. (i) the first exposure has weaker IV strength; (ii) all exposures have similar IV strength.
    \item {Factor 3:} IV strength parameter  ${\lambda_{\min}}/{\sqrt{p}}$: different values of  ${\lambda_{\min}}/{\sqrt{p}}$ are achieved by
 dividing the true SNP-exposure associations by a different value $D$. For Factor 2(i), we divide the $\gamma_{j1}$'s by $D$, where  $D = (2.5, 5.5, 9.25)$.  For Factor 2(ii),  all the true SNP-exposure associations are divided by $D$,  where $D = (2, 4.5, 7.5)$. As such, the values of $D$ depend on Factor 2 and are chosen to ensure comparable ${\lambda_{\min}}/{\sqrt{p}}$ across all settings. 
\end{enumerate}

For each setting, we generate $\hat \Gamma_j$'s and $\hat \bgamma_j$'s according to Assumption 1. We compare SRIVW in \eqref{eq: SRIVW} with five other MVMR methods: MV-IVW \citep{Burgess2015mvmr}, MV-Median   \citep{grant2021pleiotropy}, MV-Egger \citep{rees2017extending}, MRBEE \citep{lorincz2024mrbee}, and GRAPPLE \citep{Wang2021grapple}.  Throughout the simulations, MV-Median and MV-Egger are implemented using the \textsf{mr\_mvmedian} and \textsf{mr\_mvegger} functions from the \textsf{MendelianRandomization} package; MRBEE using the \textsf{MRBEE.IMRP} function from the \textsf{MRBEE} package;
GRAPPLE using the \textsf{grappleRobustEst} function from the \textsf{GRAPPLE} package---all with their default settings. 
We compare the simulation mean, standard deviation, average estimated standard errors, and the coverage probability of 95\% confidence intervals from normal approximation. Table 1 shows the simulation results with 10,000 repetitions when $\bbeta_0 = (0.8, 0.4, 0)$,  and the first exposure has weaker IV strength than the other two. When $D = 2.5, 5.5, 9.25$, the true ${ \lambda_{\min}}/{\sqrt{p}} \approx 103.6, 21.9, 7.7$, and the simulation mean ${\hat \lambda_{\min}}/{\sqrt{p}} \approx 103.2, 21.7, 7.6$, respectively. Below is a summary of the results.

    Across all scenarios, MV-IVW is biased and has inadequate coverage probabilities. In particular, even if the third exposure's true causal effect is 0, and the conditional $F$-statistic for that exposure is above 10, the coverage probabilities based on MV-IVW can be as low as 30\%. The poor coverage probabilities are likely due to both estimation bias and overly narrow confidence intervals. This is consistent with the theoretical analysis of MV-IVW, because $\lambda_{\rm min}/p^2 < 1$ and $\lambda_{\rm min}/p < 10$ in all settings, indicating that the estimation bias is non-negligible compared to the standard error.  MV-Egger and MV-Median have similar performance compared to MV-IVW, with slightly larger standard deviations.  Unlike these three traditional MVMR methods, GRAPPLE and MRBEE perform well in estimating $\beta_{02}$ and $\beta_{03}$. However, when estimating $\beta_{01}$ under very weak IV strength (the mean ${\hat \lambda_{\min}}/{\sqrt{p}}$ is 7.6), MRBEE is noticeably biased, and GRAPPLE has slight under-coverage. MRBEE also has the largest standard deviations among all the methods and exhibits over-coverage. In comparison, SRIVW has the best overall performance in these scenarios. It has the smallest bias and the coverage probabilities are close to the nominal level. When the mean ${\hat \lambda_{\min}}/{\sqrt{p}}$ is 7.6, the slight difference between the empirical standard deviation and the average estimated standard error for SRIVW in estimating $\beta_{01}$ is due to outliers in some simulation runs.

Section S4.1 of the Supplement contains simulation results for the remaining settings. Generally, the patterns we observe in Table 1 remain consistent across different settings. We highlight two additional findings. First, unlike univariable MR, the bias in MV-IVW can be either toward or away from zero. For instance, in the Supplemental Table S3, MV-IVW shows bias away from 0 for $\beta_{01}$ and bias toward 0 for $\beta_{02}$ and $\beta_{03}$. Second, when all exposures are weakly predicted, SRIVW remains valid, consistent with the theoretical analysis. Overall, GRAPPLE and SRIVW show comparable performance and consistently outperform the other methods.

In Section S4.3 of the Supplement, we conduct another simulation study in which we first generate individual-level data, then perform IV selection and estimate all summary statistics including \{$\Sigma_{Xj}$, $\sigma_{Yj}^2, j = 1,\dots, p$\}. The results are in the Supplemental Table S5. Again, SRIVW performs the best in terms of all metrics. Notably, its standard errors adequately capture the variability of the estimates, confirming that ignoring the uncertainty in estimating $\Sigma_{Xj}$ and $\sigma_{Yj}^2$ has negligible impact on inference.

\section{Real data applications}
\label{sec: real}

To evaluate the performance of our proposed method and compare it against existing approaches---including MV-IVW, MV-Median, MV-Egger, GRAPPLE, and MRBEE---we apply these methods to 
five real data examples that represent a range of realistic scenarios. Example 1 investigates the causal effects of lipoprotein(a) and low-density lipoprotein cholesterol on coronary artery disease. Example 2 examines the effects of body fat mass and fat-free mass on atrial fibrillation. Example 3 evaluates the effects of adult and childhood body mass index (BMI) on breast cancer. Example 4 explores the effects of visceral adipose tissue and gluteofemoral adipose tissue on coronary artery disease, with each analysis controlling for BMI. Finally, Example 5 estimates the effects of 19 lipoprotein subfraction traits on coronary artery disease, with each analysis controlling for traditional lipid profiles. These examples are ordered by overall IV strength, with Example 1 having the strongest IV strength. 

In the first four examples, we adopt a two-sample design, while Example 5 follows a three-sample design as described in \cite{zhao2021Mendelian}. We identify SNPs associated with at least one of the exposures using a p-value threshold of 5e-8 (except for Example 5, where we use a threshold of 1e-4 as in \cite{zhao2021Mendelian}), and apply LD clumping (minimum distance of 10 megabase pairs and LD $r^2 \le 0.001$) to select a set of uncorrelated SNPs. Finally, the exposure and outcome datasets are harmonized to include only the selected instruments. In all analyses, we assume independence between the exposure and outcome datasets. Detailed information on the data sources can be found in Section S5 of the Supplement.

Figure 2 presents the main results from the five examples, with the remaining results presented in Section S5 of the Supplement. We report both $\hat \lambda_{\min}/\sqrt{p}$ and conditional $F$-statistics to assess IV strength, where $\hat \lambda_{\min}$ is the minimum eigenvalue of the sample IV strength matrix, as defined in Section \ref{sec: asymptotics}. We highlight several key observations below. 

First, the proposed IV strength metric $\hat \lambda_{\min}/\sqrt{p}$ can provide a good assessment of overall IV strength and serve as a useful guide for applying MVMR methods.

Second, when IV strength is moderate to large (e.g., Examples 1-2), estimates from all methods are similar and generally point in the same direction. 

Third, when IV strength weakens (as observed progressively from Example 3 to Example 5), the point estimates of MV-IVW, MV-Egger, and MV-Median diverge notably from those obtained by
GRAPPLE, MRBEE, and SRIVW. While the true effects are unknown, this divergence likely indicates bias due to weak IVs, 
as suggested by our simulation results with similar IV strength. Notably, this bias does not necessarily trend toward zero and can sometimes lead to estimates in the opposite direction, for instance, in the estimated effects of adult BMI in Examples 3 and 4. 

Fourth, confidence intervals from MV-IVW are significantly narrower than those from GRAPPLE, MRBEE, and SRIVW, supporting our theoretical results that MV-IVW confidence intervals can be overly tight. In Example 3, the IVW method's significant protective effect of adult BMI on breast cancer is likely a spurious result. Similarly, confidence intervals from MV-Egger and MV-Median exhibit the same issue as MV-IVW.

Finally, across different applications, the point estimates of GRAPPLE, MRBEE, and SRIVW are generally similar, but their significance levels may differ. For example, in Example 4, GRAPPLE suggests a significant causal effect of BMI, whereas MRBEE and SRIVW do not. Additionally, MRBEE typically produces wider confidence intervals than the other two methods. Furthermore, SRIVW-pleio generally produces wider confidence intervals than SRIVW, better reflecting the additional variability introduced by pleiotropy.

The conclusions obtained from SRIVW in these applications align with previous findings. For instance, the results for Example 1 are consistent with \cite{burgess2018association}, Example 2 with \cite{tikkanen2019body}, Example 3 with \cite{richardson2020use, hao2023reassessing, wu2024causal}, Example 4 with \cite{chen2022effect}, and Example 5 with \cite{zhao2021Mendelian}.

\section{Discussion}
\label{sec: discussion}

The implications of our developments extend beyond MVMR and apply more broadly to IV approaches with multiple exposures, such as the two-stage least squares method. Furthermore, our proofs address several technical challenges arising from the generality of our asymptotic framework, including the handling of varying and unspecified growth rates of the $\mu_{nk}$'s, providing insights that could inform the theoretical analysis of future methods.  Our methods are implemented in the \textsf{mvmr.srivw} function in the \textsf{R} package \textsf{mr.divw}, available at https://github.com/tye27/mr.divw.

Motivated by our theoretical analysis, we introduce $\sum_{j=1}^{p} \Omega_j^{-1} \bgamma_j \bgamma_j^T \Omega_j^{-T}$  as an IV strength matrix for two-sample summary-data MVMR. In practice, we recommend applying the proposed SRIVW method only when the minimum eigenvalue of its sample analog, scaled by the square root of the number of SNPs, denoted by $\hat\lambda_{\min}/\sqrt{p}$, is at least 7, to ensure that the promised theoretical guarantees hold. However, this recommendation is based on empirical evidence from our simulation study and real data applications, and further research is needed to assess its robustness.
Moreover, both our theoretical analysis in Section 3 and the simulation results indicate that a conditional $F$-statistic above 10 for a single exposure is insufficient to ensure reliable inference of its causal effect in MVMR, even when the true causal effect is zero. Therefore, assessing overall IV strength is crucial in practice, and our metric, $\hat \lambda_{\min}/\sqrt{p}$, provides a simple and effective measure.

\spacingset{1.2}
\bibliographystyle{apalike}
\bibliography{reference}

\begin{thebibliography}{}

\bibitem[Agrawal et~al., 2022]{agrawal2022inherited}
Agrawal, S., Wang, M., Klarqvist, M.~D., Smith, K., Shin, J., Dashti, H., Diamant, N., Choi, S.~H., Jurgens, S.~J., Ellinor, P.~T., et~al. (2022).
\newblock Inherited basis of visceral, abdominal subcutaneous and gluteofemoral fat depots.
\newblock {\em Nature Communications}, 13(1):3771.

\bibitem[Bulik-Sullivan et~al., 2015]{bulik2015atlas}
Bulik-Sullivan, B., Finucane, H.~K., Anttila, V., Gusev, A., Day, F.~R., Loh, P.-R., Consortium, R., Consortium, P.~G., for Anorexia Nervosa of~the Wellcome Trust Case Control Consortium~3, G.~C., Duncan, L., et~al. (2015).
\newblock An atlas of genetic correlations across human diseases and traits.
\newblock {\em Nature Genetics}, 47(11):1236--1241.

\bibitem[Burgess et~al., 2018]{burgess2018association}
Burgess, S., Ference, B.~A., Staley, J.~R., Freitag, D.~F., Mason, A.~M., Nielsen, S.~F., Willeit, P., Young, R., Surendran, P., Karthikeyan, S., et~al. (2018).
\newblock Association of {LPA} variants with risk of coronary disease and the implications for lipoprotein(a)-lowering therapies: a {Mendelian} randomization analysis.
\newblock {\em JAMA Cardiology}, 3(7):619--627.

\bibitem[Chao and Swanson, 2005]{Chao:2005aa}
Chao, J.~C. and Swanson, N.~R. (2005).
\newblock Consistent estimation with a large number of weak instruments.
\newblock {\em Econometrica}, 73(5):1673--1692.

\bibitem[Chen et~al., 2022]{chen2022effect}
Chen, Q., Wu, Y., Gao, Y., Zhang, Z., Shi, T., and Yan, B. (2022).
\newblock Effect of visceral adipose tissue mass on coronary artery disease and heart failure: a {Mendelian} randomization study.
\newblock {\em International Journal of Obesity}, 46(12):2102--2106.

\bibitem[Davis et~al., 2017]{davis2017common}
Davis, J.~P., Huyghe, J.~R., Locke, A.~E., Jackson, A.~U., Sim, X., Stringham, H.~M., Teslovich, T.~M., Welch, R.~P., Fuchsberger, C., Narisu, N., et~al. (2017).
\newblock Common, low-frequency, and rare genetic variants associated with lipoprotein subclasses and triglyceride measures in {Finnish} men from the {METSIM} study.
\newblock {\em PLOS genetics}, 13(10):e1007079.

\bibitem[{Global Lipids Genetics Consortium}, 2013]{global2013discovery}
{Global Lipids Genetics Consortium} (2013).
\newblock Discovery and refinement of loci associated with lipid levels.
\newblock {\em Nature Genetics}, 45(11):1274--1283.

\bibitem[Hao et~al., 2023]{hao2023reassessing}
Hao, Y., Xiao, J., Liang, Y., Wu, X., Zhang, H., Xiao, C., Zhang, L., Burgess, S., Wang, N., Zhao, X., et~al. (2023).
\newblock Reassessing the causal role of obesity in breast cancer susceptibility: a comprehensive multivariable {Mendelian} randomization investigating the distribution and timing of exposure.
\newblock {\em International Journal of Epidemiology}, 52(1):58--70.

\bibitem[Harris et~al., 2024]{harris2024new}
Harris, B.~H., Di~Giovannantonio, M., Zhang, P., Harris, D.~A., Lord, S.~R., Allen, N.~E., Maughan, T.~S., Bryant, R.~J., Harris, A.~L., Bond, G.~L., et~al. (2024).
\newblock New role of fat-free mass in cancer risk linked with genetic predisposition.
\newblock {\em Scientific Reports}, 14(1):7270.

\bibitem[Helgeland et~al., 2022]{helgeland2022characterization}
Helgeland, {\O}., Vaudel, M., Sole-Navais, P., Flatley, C., Juodakis, J., Bacelis, J., Kol{\o}en, I.~L., Knudsen, G.~P., Johansson, B.~B., Magnus, P., et~al. (2022).
\newblock Characterization of the genetic architecture of infant and early childhood body mass index.
\newblock {\em Nature Metabolism}, 4(3):344--358.

\bibitem[Hoffmann et~al., 2018]{hoffmann2018large}
Hoffmann, T.~J., Theusch, E., Haldar, T., Ranatunga, D.~K., Jorgenson, E., Medina, M.~W., Kvale, M.~N., Kwok, P.-Y., Schaefer, C., Krauss, R.~M., et~al. (2018).
\newblock A large electronic-health-record-based genome-wide study of serum lipids.
\newblock {\em Nature Genetics}, 50(3):401--413.

\bibitem[Kettunen et~al., 2016]{kettunen2016genome}
Kettunen, J., Demirkan, A., W{\"u}rtz, P., Draisma, H.~H., Haller, T., Rawal, R., Vaarhorst, A., Kangas, A.~J., Lyytik{\"a}inen, L.-P., Pirinen, M., et~al. (2016).
\newblock Genome-wide study for circulating metabolites identifies 62 loci and reveals novel systemic effects of {LPA}.
\newblock {\em Nature Communications}, 7(1):11122.

\bibitem[Michailidou et~al., 2017]{michailidou2017association}
Michailidou, K., Lindstr{\"o}m, S., Dennis, J., Beesley, J., Hui, S., Kar, S., Lema{\c{c}}on, A., Soucy, P., Glubb, D., Rostamianfar, A., et~al. (2017).
\newblock Association analysis identifies 65 new breast cancer risk loci.
\newblock {\em Nature}, 551(7678):92--94.

\bibitem[Nelson et~al., 2017]{nelson2017association}
Nelson, C.~P., Goel, A., Butterworth, A.~S., Kanoni, S., Webb, T.~R., Marouli, E., Zeng, L., Ntalla, I., Lai, F.~Y., Hopewell, J.~C., et~al. (2017).
\newblock Association analyses based on false discovery rate implicate new loci for coronary artery disease.
\newblock {\em Nature Genetics}, 49(9):1385--1391.

\bibitem[Nielsen et~al., 2018]{nielsen2018biobank}
Nielsen, J.~B., Thorolfsdottir, R.~B., Fritsche, L.~G., Zhou, W., Skov, M.~W., Graham, S.~E., Herron, T.~J., McCarthy, S., Schmidt, E.~M., Sveinbjornsson, G., et~al. (2018).
\newblock Biobank-driven genomic discovery yields new insight into atrial fibrillation biology.
\newblock {\em Nature Genetics}, 50(9):1234--1239.

\bibitem[Pulit et~al., 2019]{pulit2019meta}
Pulit, S.~L., Stoneman, C., Morris, A.~P., Wood, A.~R., Glastonbury, C.~A., Tyrrell, J., Yengo, L., Ferreira, T., Marouli, E., Ji, Y., et~al. (2019).
\newblock Meta-analysis of genome-wide association studies for body fat distribution in 694649 individuals of european ancestry.
\newblock {\em Human Molecular Genetics}, 28(1):166--174.

\bibitem[Richardson et~al., 2020]{richardson2020use}
Richardson, T.~G., Sanderson, E., Elsworth, B., Tilling, K., and Smith, G.~D. (2020).
\newblock Use of genetic variation to separate the effects of early and later life adiposity on disease risk: {Mendelian} randomisation study.
\newblock {\em BMJ}, 369.

\bibitem[Sanderson et~al., 2021]{sanderson2021testing}
Sanderson, E., Spiller, W., and Bowden, J. (2021).
\newblock Testing and correcting for weak and pleiotropic instruments in two-sample multivariable {Mendelian} randomization.
\newblock {\em Statistics in Medicine}, 40(25):5434--5452.

\bibitem[Sinnott-Armstrong et~al., 2021]{sinnott2021genetics}
Sinnott-Armstrong, N., Tanigawa, Y., Amar, D., Mars, N., Benner, C., Aguirre, M., Venkataraman, G.~R., Wainberg, M., Ollila, H.~M., Kiiskinen, T., et~al. (2021).
\newblock Genetics of 35 blood and urine biomarkers in the {UK} biobank.
\newblock {\em Nature Genetics}, 53(2):185--194.

\bibitem[Tikkanen et~al., 2019]{tikkanen2019body}
Tikkanen, E., Gustafsson, S., Knowles, J.~W., Perez, M., Burgess, S., and Ingelsson, E. (2019).
\newblock Body composition and atrial fibrillation: a {Mendelian} randomization study.
\newblock {\em European Heart Journal}, 40(16):1277--1282.

\bibitem[Van~Wieringen and Peeters, 2016]{van2016ridge}
Van~Wieringen, W.~N. and Peeters, C.~F. (2016).
\newblock Ridge estimation of inverse covariance matrices from high-dimensional data.
\newblock {\em Computational Statistics \& Data Analysis}, 103:284--303.

\bibitem[Vogelezang et~al., 2020]{vogelezang2020novel}
Vogelezang, S., Bradfield, J.~P., Ahluwalia, T.~S., Curtin, J.~A., Lakka, T.~A., Grarup, N., Scholz, M., Van~der Most, P.~J., Monnereau, C., Stergiakouli, E., et~al. (2020).
\newblock Novel loci for childhood body mass index and shared heritability with adult cardiometabolic traits.
\newblock {\em PLOS Genetics}, 16(10):e1008718.

\bibitem[Wang et~al., 2021]{Wang2021grapple}
Wang, J., Zhao, Q., Bowden, J., Hemani, G., Davey~Smith, G., Small, D.~S., and Zhang, N.~R. (2021).
\newblock Causal inference for heritable phenotypic risk factors using heterogeneous genetic instruments.
\newblock {\em PLOS Genetics}, 17(6):1--24.

\bibitem[Wu et~al., 2024]{wu2024causal}
Wu, Z., Lewis, E., Zhao, Q., and Wang, J. (2024).
\newblock Causal mediation analysis for time-varying heritable risk factors with {Mendelian} randomization.
\newblock \textit{bioRxiv}. https://doi.org/10.1101/2024.02.10.579129.

\bibitem[Zhao et~al., 2021]{zhao2021Mendelian}
Zhao, Q., Wang, J., Miao, Z., Zhang, N.~R., Hennessy, S., Small, D.~S., and Rader, D.~J. (2021).
\newblock A {Mendelian} randomization study of the role of lipoprotein subfractions in coronary artery disease.
\newblock {\em eLife}, 10:e58361.

\end{thebibliography}


\begin{thebibliography}{}

\bibitem[Andrews and Stock, 2005]{andrews2005inference}
Andrews, D. and Stock, J.~H. (2005).
\newblock Inference with weak instruments.
\newblock National Bureau of Economic Research, Technical Paper 0313.

\bibitem[Andrews et~al., 2019]{andrews2019weak}
Andrews, I., Stock, J.~H., and Sun, L. (2019).
\newblock Weak instruments in instrumental variables regression: theory and practice.
\newblock {\em Annual Review of Economics}, 11:727--753.

\bibitem[Bowden et~al., 2017]{Bowden:2017aa}
Bowden, J., Del Greco~M, F., Minelli, C., Davey~Smith, G., Sheehan, N., and Thompson, J. (2017).
\newblock A framework for the investigation of pleiotropy in two-sample summary data {Mendelian} randomization.
\newblock {\em Statistics in Medicine}, 36(11):1783--1802.

\bibitem[Burgess et~al., 2018]{burgess2018association}
Burgess, S., Ference, B.~A., Staley, J.~R., Freitag, D.~F., Mason, A.~M., Nielsen, S.~F., Willeit, P., Young, R., Surendran, P., Karthikeyan, S., et~al. (2018).
\newblock Association of {LPA} variants with risk of coronary disease and the implications for lipoprotein(a)-lowering therapies: a {Mendelian} randomization analysis.
\newblock {\em JAMA Cardiology}, 3(7):619--627.

\bibitem[Burgess and Thompson, 2015]{Burgess2015mvmr}
Burgess, S. and Thompson, S.~G. (2015).
\newblock {Multivariable {Mendelian} randomization: the use of pleiotropic genetic variants to estimate causal effects}.
\newblock {\em American {Journal} of {Epidemiology}}, 181(4):251--260.

\bibitem[Chao and Swanson, 2005]{Chao:2005aa}
Chao, J.~C. and Swanson, N.~R. (2005).
\newblock Consistent estimation with a large number of weak instruments.
\newblock {\em Econometrica}, 73(5):1673--1692.

\bibitem[Chen et~al., 2022]{chen2022effect}
Chen, Q., Wu, Y., Gao, Y., Zhang, Z., Shi, T., and Yan, B. (2022).
\newblock Effect of visceral adipose tissue mass on coronary artery disease and heart failure: a {Mendelian} randomization study.
\newblock {\em International Journal of Obesity}, 46(12):2102--2106.

\bibitem[Cirulli and Goldstein, 2010]{Cirulli:2010aa}
Cirulli, E.~T. and Goldstein, D.~B. (2010).
\newblock Uncovering the roles of rare variants in common disease through whole-genome sequencing.
\newblock {\em Nature Reviews Genetics}, 11(6):415--425.

\bibitem[Cragg and Donald, 1993]{cragg1993testing}
Cragg, J.~G. and Donald, S.~G. (1993).
\newblock Testing identifiability and specification in instrumental variable models.
\newblock {\em Econometric Theory}, 9(2):222--240.

\bibitem[Davey~Smith and Ebrahim, 2003]{Davey-Smith:2003aa}
Davey~Smith, G. and Ebrahim, S. (2003).
\newblock `{{Mendelian}} randomization': can genetic epidemiology contribute to understanding environmental determinants of disease?
\newblock {\em International Journal of Epidemiology}, 32(1):1--22.

\bibitem[Gibson, 2012]{Gibson:2012aa}
Gibson, G. (2012).
\newblock Rare and common variants: twenty arguments.
\newblock {\em Nature Reviews Genetics}, 13(2):135--145.

\bibitem[{Global Lipids Genetics Consortium}, 2013]{global2013discovery}
{Global Lipids Genetics Consortium} (2013).
\newblock Discovery and refinement of loci associated with lipid levels.
\newblock {\em Nature Genetics}, 45(11):1274--1283.

\bibitem[Grant and Burgess, 2021]{grant2021pleiotropy}
Grant, A.~J. and Burgess, S. (2021).
\newblock Pleiotropy robust methods for multivariable {Mendelian} randomization.
\newblock {\em Statistics in Medicine}, 40(26):5813--5830.

\bibitem[Grant and Burgess, 2022]{grant2022efficient}
Grant, A.~J. and Burgess, S. (2022).
\newblock An efficient and robust approach to {{Mendelian}} randomization with measured pleiotropic effects in a high-dimensional setting.
\newblock {\em Biostatistics}, 23(2):609--625.

\bibitem[Hao et~al., 2023]{hao2023reassessing}
Hao, Y., Xiao, J., Liang, Y., Wu, X., Zhang, H., Xiao, C., Zhang, L., Burgess, S., Wang, N., Zhao, X., et~al. (2023).
\newblock Reassessing the causal role of obesity in breast cancer susceptibility: a comprehensive multivariable {Mendelian} randomization investigating the distribution and timing of exposure.
\newblock {\em International Journal of Epidemiology}, 52(1):58--70.

\bibitem[Hemani et~al., 2018a]{Hemani:2018aa}
Hemani, G., Bowden, J., and Davey~Smith, G. (2018a).
\newblock Evaluating the potential role of pleiotropy in {Mendelian} randomization studies.
\newblock {\em Human Molecular Genetics}, 27(R2):R195--R208.

\bibitem[Hemani et~al., 2018b]{Hemani:2018ab}
Hemani, G., Zheng, J., Elsworth, B., Wade, K.~H., Haberland, V., Baird, D., Laurin, C., Burgess, S., Bowden, J., Langdon, R., Tan, V.~Y., Yarmolinsky, J., Shihab, H.~A., Timpson, N.~J., Evans, D.~M., Relton, C., Martin, R.~M., Davey~Smith, G., Gaunt, T.~R., Haycock, P.~C., and Loos, R. (2018b).
\newblock The {MR-Base} platform supports systematic causal inference across the human phenome.
\newblock {\em eLife}, 7:e34408.

\bibitem[Kress et~al., 1989]{kress1989linear}
Kress, R., Maz'ya, V., and Kozlov, V. (1989).
\newblock {\em Linear integral equations}, volume~82.
\newblock Springer.

\bibitem[Lin et~al., 2023]{lin2023robust}
Lin, Z., Xue, H., and Pan, W. (2023).
\newblock Robust multivariable {Mendelian} randomization based on constrained maximum likelihood.
\newblock {\em The American Journal of Human Genetics}, 110(4):592--605.

\bibitem[Lorincz-Comi et~al., 2024]{lorincz2024mrbee}
Lorincz-Comi, N., Yang, Y., Li, G., and Zhu, X. (2024).
\newblock {MRBEE}: A bias-corrected multivariable {Mendelian} randomization method.
\newblock {\em Human Genetics and Genomics Advances}, 5(3).

\bibitem[Newey and Windmeijer, 2009]{Newey:2009aa}
Newey, W.~K. and Windmeijer, F. (2009).
\newblock Generalized method of moments with many weak moment conditions.
\newblock {\em Econometrica}, 77(3):687--719.

\bibitem[Patel et~al., 2024]{patel2024weak}
Patel, A., Lane, J., and Burgess, S. (2024).
\newblock Weak instruments in multivariable mendelian randomization: methods and practice.
\newblock {\em arXiv preprint arXiv:2408.09868}.

\bibitem[Rees et~al., 2017]{rees2017extending}
Rees, J.~M., Wood, A.~M., and Burgess, S. (2017).
\newblock Extending the {MR-Egger} method for multivariable {Mendelian} randomization to correct for both measured and unmeasured pleiotropy.
\newblock {\em Statistics in Medicine}, 36(29):4705--4718.

\bibitem[Richardson et~al., 2020]{richardson2020use}
Richardson, T.~G., Sanderson, E., Elsworth, B., Tilling, K., and Smith, G.~D. (2020).
\newblock Use of genetic variation to separate the effects of early and later life adiposity on disease risk: {Mendelian} randomisation study.
\newblock {\em BMJ}, 369.

\bibitem[Sanderson et~al., 2019]{sanderson2019examination}
Sanderson, E., Davey~Smith, G., Windmeijer, F., and Bowden, J. (2019).
\newblock An examination of multivariable {Mendelian} randomization in the single-sample and two-sample summary data settings.
\newblock {\em International Journal of Epidemiology}, 48(3):713--727.

\bibitem[Sanderson et~al., 2022a]{sanderson2022Mendelian}
Sanderson, E., Glymour, M.~M., Holmes, M.~V., Kang, H., Morrison, J., Munaf{\`o}, M.~R., Palmer, T., Schooling, C.~M., Wallace, C., Zhao, Q., et~al. (2022a).
\newblock {Mendelian} randomization.
\newblock {\em Nature Reviews Methods Primers}, 2(1):6.

\bibitem[Sanderson et~al., 2022b]{sanderson2022estimation}
Sanderson, E., Richardson, T.~G., Morris, T.~T., Tilling, K., and Smith, G.~D. (2022b).
\newblock Estimation of causal effects of a time-varying exposure at multiple time points through multivariable {Mendelian} randomization.
\newblock {\em PLOS Genetics}, 18(7):1--19.

\bibitem[Sanderson et~al., 2021]{sanderson2021testing}
Sanderson, E., Spiller, W., and Bowden, J. (2021).
\newblock Testing and correcting for weak and pleiotropic instruments in two-sample multivariable {Mendelian} randomization.
\newblock {\em Statistics in Medicine}, 40(25):5434--5452.

\bibitem[Sanderson and Windmeijer, 2016]{sanderson2016weak}
Sanderson, E. and Windmeijer, F. (2016).
\newblock A weak instrument {$F$}-test in linear {IV} models with multiple endogenous variables.
\newblock {\em Journal of Econometrics}, 190(2):212--221.

\bibitem[Staiger and Stock, 1997]{Staiger:1997aa}
Staiger, D. and Stock, J.~H. (1997).
\newblock Instrumental variables regression with weak instruments.
\newblock {\em Econometrica}, 65(3):557--586.

\bibitem[Stock and Yogo, 2001]{Stock2001:ivtest}
Stock, J. and Yogo, M. (2001).
\newblock Testing for weak instruments in linear {IV} regression.
\newblock {\em Identification and Inference for Econometric Models: Essays in Honor of Thomas Rothenberg}, pages 80--108.

\bibitem[Stock et~al., 2002]{Stock:2002aa}
Stock, J.~H., Wright, J.~H., and Yogo, M. (2002).
\newblock A survey of weak instruments and weak identification in generalized method of moments.
\newblock {\em Journal of Business \& Economic Statistics}, 20(4):518--529.

\bibitem[Tikkanen et~al., 2019]{tikkanen2019body}
Tikkanen, E., Gustafsson, S., Knowles, J.~W., Perez, M., Burgess, S., and Ingelsson, E. (2019).
\newblock Body composition and atrial fibrillation: a {Mendelian} randomization study.
\newblock {\em European Heart Journal}, 40(16):1277--1282.

\bibitem[Van~Wieringen and Peeters, 2016]{van2016ridge}
Van~Wieringen, W.~N. and Peeters, C.~F. (2016).
\newblock Ridge estimation of inverse covariance matrices from high-dimensional data.
\newblock {\em Computational Statistics \& Data Analysis}, 103:284--303.

\bibitem[Wang et~al., 2021]{Wang2021grapple}
Wang, J., Zhao, Q., Bowden, J., Hemani, G., Davey~Smith, G., Small, D.~S., and Zhang, N.~R. (2021).
\newblock Causal inference for heritable phenotypic risk factors using heterogeneous genetic instruments.
\newblock {\em PLOS Genetics}, 17(6):1--24.

\bibitem[Wu et~al., 2024]{wu2024causal}
Wu, Z., Lewis, E., Zhao, Q., and Wang, J. (2024).
\newblock Causal mediation analysis for time-varying heritable risk factors with {Mendelian} randomization.
\newblock \textit{bioRxiv}. https://doi.org/10.1101/2024.02.10.579129.

\bibitem[Ye et~al., 2021]{Ye:2019}
Ye, T., Shao, J., and Kang, H. (2021).
\newblock Debiased inverse-variance weighted estimator in two-sample summary-data {Mendelian} randomization.
\newblock {\em The Annals of Statistics}, 49(4):2079--2100.

\bibitem[Zhao et~al., 2020]{zhao2018statistical}
Zhao, Q., Wang, J., Hemani, G., Bowden, J., and Small, D.~S. (2020).
\newblock Statistical inference in two-sample summary-data {Mendelian} randomization using robust adjusted profile score.
\newblock {\em Annals of Statistics}, 48(3):1742--1769.

\bibitem[Zhao et~al., 2021]{zhao2021Mendelian}
Zhao, Q., Wang, J., Miao, Z., Zhang, N.~R., Hennessy, S., Small, D.~S., and Rader, D.~J. (2021).
\newblock A {Mendelian} randomization study of the role of lipoprotein subfractions in coronary artery disease.
\newblock {\em eLife}, 10:e58361.

\bibitem[Zuber et~al., 2020]{zuber2020selecting}
Zuber, V., Colijn, J.~M., Klaver, C., and Burgess, S. (2020).
\newblock Selecting likely causal risk factors from high-throughput experiments using multivariable {Mendelian} randomization.
\newblock {\em Nature Communications}, 11(1):29.

\end{thebibliography}

\begin{table}[h]

\caption{Simulation results for six MVMR estimators in the main simulation study with 10,000 repetitions, when $\bbeta_0 = (\beta_{01}, \beta_{02}, \beta_{03}) = (0.8, 0.4, 0)$,  and the first exposure has weaker IV strength than the other two. The number of SNPs is $p = 145$. $\hat \lambda_{\min}$ represents the minimum eigenvalue of the sample IV strength matrix $\sum_{j=1}^{p} \Omega_j^{-1} \hat\bgamma_j \hat\bgamma_j^T \Omega_j^{-T} - pI_K$. Abbreviations: Est = estimated causal effect; SD = standard deviation; SE = standard error; CP = coverage probability.}
\centering
\resizebox{\linewidth}{!}{
\begin{tabular}[t]{c|cccc|cccc|cccc}
\toprule
\multicolumn{1}{c}{ } & \multicolumn{4}{c}{$\beta_{01}$ = 0.8} & \multicolumn{4}{c}{$\beta_{02}$ = 0.4} & \multicolumn{4}{c}{$\beta_{03}$ = 0} \\
\cmidrule(l{3pt}r{3pt}){2-5} \cmidrule(l{3pt}r{3pt}){6-9} \cmidrule(l{3pt}r{3pt}){10-13}
Estimator & Est & SD & SE & CP & Est & SD & SE & CP & Est & SD & SE & CP\\
\midrule
\addlinespace[0.3em]
\multicolumn{13}{l}{\textit{Mean $\hat \lambda_{\min}/\sqrt{p}$ = 103.4, mean conditional $F$-statistics = 9.8, 38.4, 23.5}}\\
\hspace{1em}\cellcolor{gray!6}{MV-IVW} & \cellcolor{gray!6}{0.720} & \cellcolor{gray!6}{0.026} & \cellcolor{gray!6}{0.028} & \cellcolor{gray!6}{0.179} & \cellcolor{gray!6}{0.395} & \cellcolor{gray!6}{0.009} & \cellcolor{gray!6}{0.009} & \cellcolor{gray!6}{0.909} & \cellcolor{gray!6}{-0.015} & \cellcolor{gray!6}{0.012} & \cellcolor{gray!6}{0.013} & \cellcolor{gray!6}{0.781}\\
\hspace{1em}MV-Egger & 0.739 & 0.035 & 0.037 & 0.620 & 0.395 & 0.009 & 0.008 & 0.885 & -0.016 & 0.013 & 0.013 & 0.773\\
\hspace{1em}\cellcolor{gray!6}{MV-Median} & \cellcolor{gray!6}{0.723} & \cellcolor{gray!6}{0.036} & \cellcolor{gray!6}{0.042} & \cellcolor{gray!6}{0.556} & \cellcolor{gray!6}{0.395} & \cellcolor{gray!6}{0.011} & \cellcolor{gray!6}{0.013} & \cellcolor{gray!6}{0.956} & \cellcolor{gray!6}{-0.014} & \cellcolor{gray!6}{0.017} & \cellcolor{gray!6}{0.019} & \cellcolor{gray!6}{0.914}\\
\hspace{1em}GRAPPLE & 0.798 & 0.031 & 0.032 & 0.953 & 0.400 & 0.009 & 0.009 & 0.956 & -0.001 & 0.014 & 0.014 & 0.955\\
\hspace{1em}\cellcolor{gray!6}{MRBEE} & \cellcolor{gray!6}{0.805} & \cellcolor{gray!6}{0.033} & \cellcolor{gray!6}{0.034} & \cellcolor{gray!6}{0.946} & \cellcolor{gray!6}{0.400} & \cellcolor{gray!6}{0.010} & \cellcolor{gray!6}{0.010} & \cellcolor{gray!6}{0.938} & \cellcolor{gray!6}{0.001} & \cellcolor{gray!6}{0.014} & \cellcolor{gray!6}{0.015} & \cellcolor{gray!6}{0.950}\\
\hspace{1em}SRIVW & 0.803 & 0.033 & 0.033 & 0.954 & 0.400 & 0.009 & 0.010 & 0.951 & 0.000 & 0.014 & 0.014 & 0.953\\
\addlinespace[0.3em]
\multicolumn{13}{l}{\textit{Mean $\hat \lambda_{\min}/\sqrt{p}$ = 21.7, mean conditional $F$-statistics = 2.9, 28.3, 14.7}}\\
\hspace{1em}\cellcolor{gray!6}{MV-IVW} & \cellcolor{gray!6}{0.524} & \cellcolor{gray!6}{0.044} & \cellcolor{gray!6}{0.047} & \cellcolor{gray!6}{0.000} & \cellcolor{gray!6}{0.393} & \cellcolor{gray!6}{0.008} & \cellcolor{gray!6}{0.008} & \cellcolor{gray!6}{0.858} & \cellcolor{gray!6}{-0.023} & \cellcolor{gray!6}{0.011} & \cellcolor{gray!6}{0.011} & \cellcolor{gray!6}{0.440}\\
\hspace{1em}MV-Egger & 0.563 & 0.064 & 0.066 & 0.056 & 0.393 & 0.008 & 0.008 & 0.864 & -0.023 & 0.011 & 0.012 & 0.480\\
\hspace{1em}\cellcolor{gray!6}{MV-Median} & \cellcolor{gray!6}{0.510} & \cellcolor{gray!6}{0.064} & \cellcolor{gray!6}{0.064} & \cellcolor{gray!6}{0.008} & \cellcolor{gray!6}{0.395} & \cellcolor{gray!6}{0.010} & \cellcolor{gray!6}{0.011} & \cellcolor{gray!6}{0.942} & \cellcolor{gray!6}{-0.025} & \cellcolor{gray!6}{0.015} & \cellcolor{gray!6}{0.016} & \cellcolor{gray!6}{0.659}\\
\hspace{1em}GRAPPLE & 0.787 & 0.077 & 0.076 & 0.939 & 0.400 & 0.009 & 0.009 & 0.957 & -0.001 & 0.014 & 0.014 & 0.954\\
\hspace{1em}\cellcolor{gray!6}{MRBEE} & \cellcolor{gray!6}{0.829} & \cellcolor{gray!6}{0.104} & \cellcolor{gray!6}{0.105} & \cellcolor{gray!6}{0.970} & \cellcolor{gray!6}{0.400} & \cellcolor{gray!6}{0.010} & \cellcolor{gray!6}{0.010} & \cellcolor{gray!6}{0.946} & \cellcolor{gray!6}{0.003} & \cellcolor{gray!6}{0.015} & \cellcolor{gray!6}{0.016} & \cellcolor{gray!6}{0.960}\\
\hspace{1em}SRIVW & 0.818 & 0.099 & 0.098 & 0.960 & 0.400 & 0.010 & 0.010 & 0.955 & 0.002 & 0.015 & 0.015 & 0.957\\
\addlinespace[0.3em]
\multicolumn{13}{l}{\textit{Mean $\hat \lambda_{\min}/\sqrt{p}$ = 7.6, mean conditional $F$-statistics = 1.7, 29.3, 17.2}}\\
\hspace{1em}\cellcolor{gray!6}{MV-IVW} & \cellcolor{gray!6}{0.330} & \cellcolor{gray!6}{0.053} & \cellcolor{gray!6}{0.054} & \cellcolor{gray!6}{0.000} & \cellcolor{gray!6}{0.393} & \cellcolor{gray!6}{0.007} & \cellcolor{gray!6}{0.007} & \cellcolor{gray!6}{0.827} & \cellcolor{gray!6}{-0.023} & \cellcolor{gray!6}{0.009} & \cellcolor{gray!6}{0.009} & \cellcolor{gray!6}{0.303}\\
\hspace{1em}MV-Egger & 0.363 & 0.083 & 0.082 & 0.001 & 0.393 & 0.007 & 0.007 & 0.855 & -0.023 & 0.009 & 0.010 & 0.371\\
\hspace{1em}\cellcolor{gray!6}{MV-Median} & \cellcolor{gray!6}{0.310} & \cellcolor{gray!6}{0.075} & \cellcolor{gray!6}{0.069} & \cellcolor{gray!6}{0.000} & \cellcolor{gray!6}{0.395} & \cellcolor{gray!6}{0.009} & \cellcolor{gray!6}{0.010} & \cellcolor{gray!6}{0.944} & \cellcolor{gray!6}{-0.026} & \cellcolor{gray!6}{0.013} & \cellcolor{gray!6}{0.014} & \cellcolor{gray!6}{0.557}\\
\hspace{1em}GRAPPLE & 0.771 & 0.158 & 0.154 & 0.917 & 0.400 & 0.009 & 0.009 & 0.961 & -0.002 & 0.015 & 0.015 & 0.946\\
\hspace{1em}\cellcolor{gray!6}{MRBEE} & \cellcolor{gray!6}{0.912} & \cellcolor{gray!6}{0.280} & \cellcolor{gray!6}{0.419} & \cellcolor{gray!6}{0.975} & \cellcolor{gray!6}{0.396} & \cellcolor{gray!6}{0.031} & \cellcolor{gray!6}{0.013} & \cellcolor{gray!6}{0.939} & \cellcolor{gray!6}{0.006} & \cellcolor{gray!6}{0.021} & \cellcolor{gray!6}{0.028} & \cellcolor{gray!6}{0.976}\\
\hspace{1em}SRIVW & 0.785 & 0.159 & 0.221 & 0.953 & 0.400 & 0.009 & 0.010 & 0.965 & -0.001 & 0.015 & 0.017 & 0.957\\
\bottomrule
\end{tabular}}
\end{table}

\begin{figure}[h]
    \centering\includegraphics[width=1\textwidth]{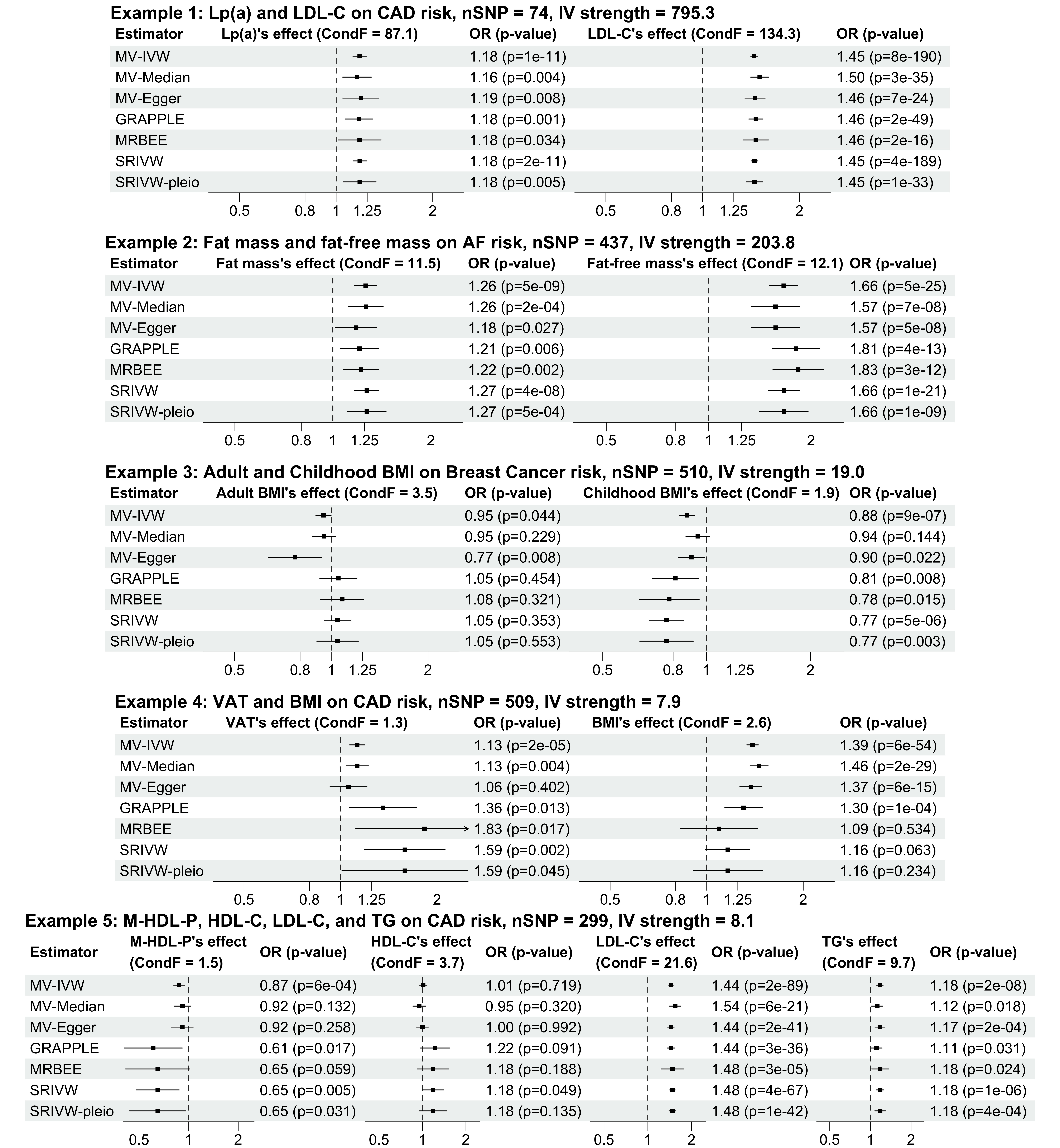}
    \caption{{Estimated effects of exposures on outcomes from Examples 1-5 across different estimators. The odds ratios (OR) and corresponding p-values are reported, with horizontal bars representing 95\% confidence intervals. IV strength is calculated as $\hat \lambda_{\mathrm{min}}/\sqrt{p}$, where $\hat \lambda_{\mathrm{min}}$ represents the minimum eigenvalue of the sample IV strength matrix. Abbreviations: Lp(a) (lipoprotein(a)), LDL-C (low-density lipoprotein cholesterol), HDL-C (high-density lipoprotein cholesterol), TG (Triglycerides), CAD (coronary artery disease), AF (atrial fibrillation),  VAT (visceral adipose tissue),  M-HDL-P (medium HDL particle), nSNP (number of SNPs as IVs), CondF (conditional $F$-statistic).}}
\end{figure}

\clearpage

\pagenumbering{arabic}
\setcounter{equation}{0}
\setcounter{table}{0}
\setcounter{section}{0}
\setcounter{lemma}{0}
\setcounter{assumption}{2}
\setcounter{theorem}{2}
\renewcommand{\theequation}{S\arabic{equation}}
\renewcommand{\thetable}{S\arabic{table}}
\renewcommand{\thelemma}{S\arabic{lemma}}
\renewcommand{\thesection}{S\arabic{section}}

\def\spacingset#1{\renewcommand{\baselinestretch}%
{#1}\small\normalsize} \spacingset{1}

\if1\blind
{
	 \begin{center} 
	\spacingset{1.5} 	{\LARGE\bf  Supplement to ``A More Robust Approach to Multivariable Mendelian Randomization''} \\ \bigskip \bigskip
	\spacingset{1} 
	{\large Yinxiang Wu$ ^1 $, Hyunseung Kang$ ^2 $, and Ting Ye$ ^1 $ } \\ \bigskip
	 {$ ^1 $Department of Biostatistics, University of Washington, Seattle, Washington, U.S.A. \\ 
    $ ^2 $Department of Statistics, University of Wisconsin-Madison, Madison, Wisconsin, U.S.A.\\}
\end{center}
} \fi

\if0\blind
{
  \bigskip \spacingset{1.5} 
  \begin{center}
     {\LARGE \bf Supplement to ``A More Robust Approach to Multivariable Mendelian Randomization''}
\end{center}
  \medskip
} \fi

\section{Notations, useful results and lemmas}\label{sec: notation}

Throughout the Supplement, $\bm0$ denotes a vector or a matrix of all zeros, $I_{K}$ denotes an $K \times K$ identity matrix, a matrix $A \ge 0$ mean that $A$ is a positive semi-definite matrix, with strict inequality if $A$ is positive definite. We refer to the entry located at $t$-th row and $s$-th column of a matrix $A$ as the $(t,s)$ element of that matrix, denoted as $A_{[t,s]}$. We let $\lambda_{\min}$ and $\lambda_{\max}$ denote the smallest and largest eigenvalue of a square matrix. We use $c$ and $C$ to denote generic positive constants. For two sequences of real numbers $a_n$ and $b_n$ indexed by $n$, we write $a_n = O(b_n)$ if $|a_n| \le c|b_n|$ for all $n$ and a constant $c$, $a_n = o(b_n)$ if $a_n/b_n \rightarrow 0$ as $n \rightarrow \infty$, $a_n = \Theta (b_n)$ if $c^{-1}b_n \le |a_n| \le c b_n$ for all $n$ and a constant $c$. We say a sequence of matrices is $O(1)$ or $\Theta(1)$ if every entry is so. We use $\xrightarrow{P}$ to denote convergence in probability and $\xrightarrow{D}$ to denote convergence in distribution. For random variables $X$ and $Y$, we denote $X = o_p(Y)$ if $X/Y \xrightarrow[]{P} 0$, $X = O_p(Y)$ if $X/Y$ is bounded in probability. To ease the notation, unless otherwise specified, we use $\sum_j$ to denote $\sum_{j = 1}^{p}$.

Also recall the following definitions: $M_j = \bgamma_j\bgamma_j^T\sigma_{Yj}^{-2}$, $\hat M_j = \hat \bgamma_j\hat \bgamma_j^T\sigma_{Yj}^{-2}$, $V_j = \Sigma_{Xj}\sigma_{Yj}^{-2}$, $\mathbb{V} = \sum_j \{(1 + \bm{\beta}_0^TV_j\bm{\beta}_0)(M_j + V_j) + V_j\bm{\beta}_0\bm{\beta}_0^TV_j\}$.

In the following, we present some useful results and establish lemmas that will be repeatedly used in the proofs. We start with some useful results that can be directly derived using probability theory and algebras:
\begin{align}\label{eq: Vs}
    E\left \{\sum_j \hat M_j \right \} = \sum_j (M_j + V_j) 
\end{align}
\begin{align} \label{eq: bias}
    E\left \{\sum_{j}(\hat \bgamma_j\hat \Gamma_j\sigma_{Yj}^{-2} - \hat M_j\bbeta_0)\right \} = -\sum_j V_j\bbeta_0
\end{align}
and 
\begin{align}
    Cov\left \{\sum_{j}(\hat \bgamma_j\hat \Gamma_j\sigma_{Yj}^{-2} - \hat M_j\bbeta_0)\right \} = \mathbb{V} \label{eq: Vt}.
\end{align}

Next, we establish several useful lemmas.

\begin{lemma}\label{lemma0}
 Let $A \in \mathbb{R}^{p\times q}$ be a real-valued matrix and $B \in \mathbb{R}^{q\times q}$ be a positive definite matrix. Then, we have $\lambda_{\min}(ABA^T) \ge \lambda_{\min}(B) \lambda_{\min}(AA^T)$, and $\lambda_{\max}(ABA^T) \le \lambda_{\max}(B) \lambda_{\max}(AA^T)$.
\end{lemma}
\begin{proof}
    By definition, the minimum eigenvalue of a symmetric $p\times p$ matrix $M$ satisfies
    \begin{align*}
        \lambda_{\min}(M) = {\rm min}_{x\in \mathbb{R}^p}(\frac{x^T M x}{x^Tx}), 
    \end{align*}
where $x \in \mathbb{R}^p$. Take any $x \in \mathbb{R}^p$, then
    \begin{align*}
        x^T A B A^T x \ge \lambda_{\min}(B) x^T A A^T x \ge \lambda_{\min}(B)\lambda_{\min}(AA^T)x^Tx
    \end{align*}
Hence, $\lambda_{\min}(A B A^T) = {\rm min}_{x\in \mathbb{R}^p}(x^T A B A^T x/x^Tx) \ge \lambda_{\min}(B)\lambda_{\min}(AA^T)$

On the other hand, by definition, the maximum eigenvalue of a symmetric $p\times p$ matrix $M$ satisfies
    \begin{align*}
        \lambda_{\max}(M) = {\rm max}_{x\in \mathbb{R}^p}(\frac{x^T M x}{x^Tx}), 
    \end{align*}
where $x \in \mathbb{R}^p$. Take any $x \in \mathbb{R}^p$,
    \begin{align*}
        x^T A B A^T x \le \lambda_{\max}(B) x^T A A^T x \le \lambda_{\max}(B)\lambda_{\max}(AA^T)x^Tx
    \end{align*}
Hence, $\lambda_{\max}(ABA^T) = {\rm max}_{x\in \mathbb{R}^p}(x^T A B A^T x/x^Tx) \le \lambda_{\max}(B)\lambda_{\max}(AA^T)$
\end{proof}

\begin{lemma}\label{lemma1}
    Let $T_n = \tilde S_n {\rm diag}(\sqrt{\mu_{n1} + p},...,\sqrt{\mu_{nK} + p})$. Under Assumptions 1-2, $T_n^{-1}(\sum_j M_j + V_j)T_n^{-T}$ and $T_n^{-1} \mathbb{V} T_n^{-T}$ are bounded element-wise with minimum eigenvalues bounded away from 0.
\end{lemma}

\begin{proof}
Define \(\mu_{n,\text{min}} = \min\{\mu_{n1},\ldots,\mu_{nK}\}\) and \(\mu_{n,\text{max}} = \max\{\mu_{n1},\ldots,\mu_{nK}\}\). We first show that \(T_n^{-1}(\sum_j M_j + V_j)T_n^{-T}\) is element-wise bounded, and its minimum eigenvalue is bounded away from 0. To establish this, it suffices to consider three scenarios: (1) \(\frac{\mu_{n,\text{min}}}{p} \rightarrow \tau\) for some \(\tau \in (0, \infty]\), (2) \(\frac{\mu_{n,\text{max}}}{p} \rightarrow 0\), and (3) assuming, without loss of generality, there exists an index \(l\) (where \(1 \leq l \leq K-1\)) such that for all \(k \leq l\), \(\frac{\mu_{nk}}{p} \rightarrow \tau_k\) with \(\tau_k \in (0,\infty]\), and for all \(k > l\), \(\frac{\mu_{nk}}{p} \rightarrow 0\). In the third scenario, we categorize the indices based on the convergence rates of their respective \(\mu_k\)'s.

   Let $\tilde \bgamma_j = \tilde S_n^{-1} \bgamma_j$ and $\tilde M_j = \tilde \bgamma_j \tilde \bgamma_j^T \sigma_{Y_j}^{-2}$. Then, $S_n^{-1}(\sum_j M_j) S_n^{-T}$ being bounded (Assumption 2(iii)) is equivalent to that 
\begin{equation*}
 {\rm {\rm diag}}(\frac{1}{\sqrt{\mu_{n1}}},...,\frac{1}{\sqrt{\mu_{nK}}})\ (\sum_j \tilde M_j)  {\rm {\rm diag}}(\frac{1}{\sqrt{\mu_{n1}}},...,\frac{1}{\sqrt{\mu_{nK}}})
\end{equation*} 
is bounded. Thus, the $(t,s)$ entry of the matrix $\sum_j \tilde M_j$ satisfies $\sum_j \tilde \gamma_{jt}\tilde \gamma_{js}\sigma_{Yj}^{-2} = O(\sqrt{\mu_{nt}\mu_{ns}})$.
    
    For the scenario (1), note that
    \begin{align*}
        T_n^{-1} (\sum_j M_j) T_n^{-T} =  {\rm {\rm diag}}(\frac{1}{\sqrt{\mu_{n1} + p}},...,\frac{1}{\sqrt{\mu_{nk} + p}}) (\sum_j \tilde M_j) {\rm {\rm diag}}(\frac{1}{\sqrt{\mu_{n1} + p}},...,\frac{1}{\sqrt{\mu_{nk} + p}})    
    \end{align*}
    Then, using Lemma \ref{lemma0}, we have
    \begin{align*}
        & \lambda_{\min}(T_n^{-1} (\sum_j M_j) T_n^{-T}) \\
        & = \lambda_{\min}({\rm diag}(\sqrt{\frac{\mu_{n1}}{\mu_{n1} + p}},..., \sqrt{\frac{\mu_{n1}}{\mu_{n1} + p}}) \{S_n^{-1} (\sum_j M_j) S_n^{-T}\} {\rm diag}(\sqrt{\frac{\mu_{n1}}{\mu_{n1} + p}},..., \sqrt{\frac{\mu_{n1}}{\mu_{n1} + p}})) \\
        & \ge \Theta(\lambda_{\min}(S_n^{-1} (\sum_j M_j) S_n^{-T})) > 0
    \end{align*}
    Similarly, by applying Lemma \ref{lemma0}, we have
    \begin{align*}
        & \lambda_{\max}(T_n^{-1} (\sum_j M_j) T_n^{-T}) \le \Theta(\lambda_{\max}(S_n^{-1} (\sum_j M_j) S_n^{-T})) < \infty
    \end{align*}
    To analyze the term $T_n^{-1} (\sum_j V_j) T_n^{-T}$, we let $\tilde V_j = \tilde S_n^{-1} V_j \tilde S_n^{-T}$. Based on Assumptions 1-2, we know $\tilde V_j$ is also positive definite and bounded. {In particular, the minimum and maximum eigenvalue of $\tilde V_j$ are $\Theta(1)$. Thus, $\sum_j \tilde V_j$ is always positive definite with every entry of rate $O(p)$. Furthermore, we know, by Weyl's inequality, the maximum and minimum eigenvalue of $\sum_j \tilde V_j$ are $\Theta(p)$}. Then, note that
    \begin{align*}
        & T_n^{-1} (\sum_j V_j) T_n^{-T} = {\rm {\rm diag}}(\frac{1}{\sqrt{\mu_{n1} + p}},...,\frac{1}{\sqrt{\mu_{nk} + p}}) (\sum_j \tilde V_j) {\rm {\rm diag}}(\frac{1}{\sqrt{\mu_{n1} + p}},...,\frac{1}{\sqrt{\mu_{nk} + p}})    .
    \end{align*}
    Under the current scenario that $p = O(\mu_{nk})$ for every $k$, $T_n^{-1} (\sum_j  V_j) T_n^{-T}$ is bounded and positive semi-definite.
    Thus, $T_n^{-1}(\sum_j M_j + V_j)T_n^{-T}$ is bounded with eigenvalues bounded away from 0 as $n \rightarrow \infty$.

In scenario (2), where \(\mu_{n,\text{max}} = o(p)\), it is observed that \(T_n^{-1} (\sum_j M_j) T_n^{-T}\) converges to a zero matrix. However, \(T_n^{-1} (\sum_j V_j) T_n^{-T}\) becomes the dominant factor and remains bounded, with its eigenvalues bounded away from zero. Consequently, as \(n \rightarrow \infty\), \(T_n^{-1}(\sum_j M_j + V_j)T_n^{-T}\) is also bounded, with eigenvalues bounded away from zero.

    Lastly, we study scenario (3). For the term $T_n^{-1} (\sum_j M_j) T_n^{-T}$, based on similar reasoning as above,  we know the upper left block with row and column indices from 1 to $l$ is bounded and its smallest and maximum eigenvalues are bounded between 0 and $\infty$. For the entries not appearing in the upper left block, we know they will converge to 0, as $\sum_j \tilde \gamma_{jt}\tilde \gamma_{js}\sigma_{Yj}^{-2}/{\sqrt{(\mu_{nt} + p)(\mu_{ns} + p)}} = O(\sqrt{\frac{\mu_{nt}\mu_{ns}}{(\mu_{nt} + p)(\mu_{ns} + p)}}) = o(1)$ for at least one of $t$ and $s$ belonging to the index set $\{l+1,...,K\}$. {For the term $T_n^{-1} (\sum_j V_j) T_n^{-T}$, we know from scenarios (1) and (2) that it is always bounded and positive semi-definite.} Therefore, $T_n^{-1} (\sum_j M_j + V_j) T_n^{-T}$ is a block {\rm diag}onal matrix that is bounded and has eigenvalues bounded between 0 and $\infty$ as $n \rightarrow \infty$.

    To show that $T_n^{-1} \mathbb{V} T_n^{-T}$ is also bounded and its mimimum eigenvalue is bounded away from 0, note that 
    \begin{align*}
        & T_n^{-1} \mathbb{V} T_n^{-T} \\
        &= T_n^{-1}(\sum_j M_j + V_j) T_n^{-T} + T_n^{-1}(\sum_j (\bbeta_0^T V_j \bbeta_0)(M_j + V_j)) T_n^{-T} + T_n^{-1} (\sum_j V_j \bbeta_0 \bbeta_0^T V_j ) T_n^{-T}
    \end{align*}
    We have shown that $T_n^{-1}(\sum_j M_j + V_j)T_n^{-T}$ is bounded with minimum eigenvalue bounded away from 0.
    Following similar arguments as above, we can show both\\ $T_n^{-1} \left\{\sum_j (\bbeta_0^T V_j \bbeta_0)(M_j + V_j) \right\} T_n^{-T}$ and $T_n^{-1} (\sum_j V_j \bbeta_0 \bbeta_0^T V_j ) T_n^{-T}$ are bounded and positive semi-definite. Because the sum of positive definite matrices and positive semi-definite matrices is positive definite, we conclude that $T_n^{-1} \mathbb{V} T_n^{-T}$ is also bounded with minimum eigenvalue bounded away from 0.
\end{proof}

\begin{lemma}
    \label{lemma2} Under Assumptions 1-2, $T_n^{-1}(\sum_j \hat M_j - (M_j + V_j)) T_n^{-T} \xrightarrow[]{P} \bm 0$ regardless of the relative  rates between $\mu_{nk}$'s and $p$, where $T_n$ is defined in Lemma \ref{lemma1}.
\end{lemma}

\begin{proof}

Recall the definitions $\tilde \bgamma_j = \tilde S_n^{-1} \bgamma_j$ and $\tilde M_j = \tilde \bgamma_j \tilde \bgamma_j^T \sigma_{Y_j}^{-2}$. Let  $\hat{\tilde {\bgamma_j}} = \tilde S_n^{-1} \hat\bgamma_j$, which satisfies  $\hat{\tilde {\bgamma_j}} \sim N(\tilde \bgamma_j, \tilde \Sigma_{Xj})$ and
\begin{align*}
    \tilde \Sigma_{Xj} & = \tilde S_n^{-1} \Sigma_{Xj} \tilde S_n^{-T} = \begin{pmatrix}
\tilde \sigma_{j11}^2 & \dots & \tilde \sigma_{j1K}^2\\
\vdots & \ddots & \vdots \\
\tilde \sigma_{jK1}^2 & \dots & \tilde \sigma_{jKK}^2 
\end{pmatrix}.
\end{align*}
Note that
\begin{flalign*}
    & T_n^{-1} \{\sum_j \hat M_j - (M_j + V_j)\} T_n^{-T} \\
    & = 
    {\rm diag}(\sqrt{\frac{1}{\mu_{n1}+p}},...,\sqrt{\frac{1}{\mu_{nK}+p}}) \tilde S_n^{-1} \{\sum_j \hat M_j - (M_j + V_j)\} \tilde S_n^{-T} {\rm diag}(\sqrt{\frac{1}{\mu_{n1}+p}},...,\sqrt{\frac{1}{\mu_{nK}+p}}) \\
    & = {\rm diag}(\sqrt{\frac{1}{\mu_{n1}+p}},...,\sqrt{\frac{1}{\mu_{nK}+p}}) \{\sum_j \hat{ \tilde{ \bgamma}}_j \hat{\tilde{ \bgamma}}_j^T \sigma_{Yj}^{-2} -  \sum_j (\tilde \bgamma_j \tilde \bgamma_j^T + \tilde \Sigma_{Xj})\sigma_{Yj}^{-2} \}
    {\rm diag}(\sqrt{\frac{1}{\mu_{n1}+p}},...,\sqrt{\frac{1}{\mu_{nK}+p}}) .
\end{flalign*}

Consider the $(t,s)$ element of the above matrix. For any $\epsilon > 0$, applying the Chebyshev's inequality, we have
\begin{align*}
    & P(|\frac{1}{\sqrt{(\mu_{nt}+p)(\mu_{ns}+p)}}\sum_{j}\{(\hat{\tilde{\gamma}}_{jt}\hat{\tilde{\gamma}}_{js} -  \tilde \gamma_{jt} \tilde \gamma_{js} + \tilde \sigma_{jts}^2)\sigma_{Yj}^{-2}\}| \ge \epsilon) \\
& \le \frac{1}{\epsilon^2}Var(\frac{1}{\sqrt{(\mu_{nt}+p)(\mu_{ns}+p)}}\sum_{j}\hat{\tilde{\gamma}}_{jt}\hat{\tilde{\gamma}}_{js}\sigma_{Yj}^{-2}) 
\end{align*}
Using the moment properties of normal distributions, we have
\begin{align*}
    {Var(\hat{\tilde{\gamma}}_{jt}\hat{\tilde{\gamma}}_{js}\sigma_{Yj}^{-2}) }= \left\{\tilde {\gamma}_{jt}^2\tilde \sigma_{jss}^2 + \tilde {\gamma}_{js}^2\tilde \sigma_{jtt}^2 + \tilde \sigma_{jss}^2\tilde \sigma_{jtt}^2 + 2\tilde {\gamma}_{jt}\tilde {\gamma}_{js}\tilde \sigma_{jst}^2 +  \tilde \sigma_{jst}^4\right \}\sigma_{Yj}^{-4}.
\end{align*} 
By the boundedness of $\tilde S_n$ and the boundedness of ratios between $\sigma_{Xjk}^2$ and $\sigma_{Yj}^2$ for every $j$ and $k$, we have $\tilde \sigma_{jst}^2\sigma_{Yj}^{-2} = O(1)$. Then, applying Cauchy-Schwarz inequality to the term $\tilde \gamma_{jt}\tilde\gamma_{js}\tilde \sigma_{jst}^2$, we can show that
\begin{align*}
  & Var(\frac{1}{\sqrt{(\mu_{nt}+p)(\mu_{ns}+p)}}\sum_{j}\hat{\tilde{\gamma}}_{jt}\hat{\tilde{\gamma}}_{js}\sigma_{Yj}^{-2}) \le  O(\frac{p(\tilde \kappa_t + \tilde \kappa_s)+p}{(\mu_{nt}+p)(\mu_{ns}+p)})
\end{align*}
where $p \tilde \kappa_k = \sum_j \tilde \gamma_{jk}^2\sigma_{Yj}^{-2}$.
Because $S_n^{-1}(\sum_j M_j) S_n^{-T}$ is bounded with the minimum eigenvalue bounded away from 0, we have $\sum_j \tilde \gamma_{jk}^2\sigma_{Yj}^{-2} = \Theta(\mu_{nk})$ for every $k$. Hence, $O(\frac{p(\tilde \kappa_t + \tilde \kappa_s)+p}{(\mu_{nt}+p)(\mu_{ns}+p)})$ can be further simplified as 
$O(\frac{\mu_{nt}+\mu_{ns}+p}{(\mu_{nt}+p)(\mu_{ns}+p)})$, {which converges to 0 as $p \rightarrow \infty$}. This implies that  $Var(\sum_{j}\hat{\tilde{\gamma}}_{jt}\hat{\tilde{\gamma}}_{js}\sigma_{Yj}^{-2}/\sqrt{(\mu_{nt}+p)(\mu_{ns}+p)})\rightarrow 0$ for every $(t,s)$ entry, completing the proof.
\end{proof}

\begin{lemma} \label{lemma 3}
    Under Assumptions 1-2, $(\sum_j V_j) \mathbb{V}^{-1} (\sum_j V_j) \rightarrow \bm 0$ if $\mu_{n,min}/p^2\rightarrow \infty$.
\end{lemma}

\begin{proof}
    Recall that $T_n = \tilde S_n {\rm diag}(\sqrt{\mu_{n1} + p},...,\sqrt{\mu_{nk} + p})$. Note that
    \begin{align*}
        (\sum_j V_j) \mathbb{V}^{-1} (\sum_j V_j) & = (\sum_j V_j) T_n^{-T} (T_n^{-1} \mathbb{V} T_n^{-T})^{-1} T_n^{-1} (\sum_j V_j)
    \end{align*}
    By Lemma \ref{lemma1}, $T_n^{-1} \mathbb{V} T_n^{-T}$ is bounded with eigenvalues bounded away from 0. Because $\tilde S_n$ is bounded, $\mu_{n,min}/p^2\rightarrow \infty$, and every entry of $\sum_j V_j$ is of rate $O(p)$, we have $T_n^{-1}(\sum_j V_j) \rightarrow \bm 0$. Hence, we have the desired result. 
\end{proof}

\begin{lemma} \label{lemma 4}
    Under Assumptions 1-2, $\mathcal{V}_{\rm IVW} = (\sum_j M_j + V_j)^{-1}\mathbb{V}(\sum_j M_j + V_j)^{-1} \rightarrow \bm0$ regardless of the relative  rates between $\mu_{nk}$'s and $p$. Furthermore, $T_n^{-1} \mathcal{V}_{\rm IVW}^{-1} T_n^{-T}$ is bounded element-wise with its minimum eigenvalues bounded away from 0.
\end{lemma}

\begin{proof}
    Recall that $T_n = \tilde S_n {\rm diag}(\sqrt{\mu_{n1} + p},...,\sqrt{\mu_{nk} + p})$. Note that
    \begin{align*}
    & (\sum_j M_j + V_j)^{-1}\mathbb{V}(\sum_j M_j + V_j)^{-1} \\
    & = T_n^{-T} \left(T_n^{-1} (\sum_j M_j + V_j) T_n^{-T}\right)^{-1} \left(T_n^{-1} \mathbb{V} T_n^{-T}\right) \left(T_n^{-1} (\sum_j M_j + V_j) T_n^{-T}\right)^{-1} T_n^{-1}
    \end{align*}
    By Lemma \ref{lemma1}, both $\left(T_n^{-1} (\sum_j M_j + V_j) T_n^{-T}\right)^{-1}$ and $T_n^{-1} \mathbb{V} T_n^{-T}$ are bounded. The desired result follows because $T_n^{-1}$ converges to a zero matrix.

    Furthermore, from above derivations and by Lemma S1 and Lemma S2, we know that $T_n^{-1} \mathcal{V}_{\rm IVW}^{-1} T_n^{-T}$ is bounded with minimum eigenvalue bounded away from 0.
\end{proof}

\begin{lemma} \label{lemma5}
    Under Assumptions 1-2, if $\max_j (\gamma_{jk}^2 \sigma_{Xjk}^{-2}) /(\mu_{n,min} + p) \rightarrow 0$ for every $k = 1, ..., K$,
    $\mathbb{V}^{-\frac{1}{2}}(\sum_j \hat \bgamma_j \hat \Gamma_j \sigma_{Yj}^{-2} - \hat M_j\bbeta_0 + V_j\bbeta_0) \xrightarrow[]{D} N(0, I_{K})$. 
\end{lemma}

\begin{proof}
    The proof uses the Lindeberg Central Limit Theorem and Cramer-Wald device.
    Pick any $\bm{c} \in \mathbb{R}^K$ such that $\bm{c}^T\bm{c} = 1$. Write $\mathbb{V} = \sum_j \mathbb{V}_j$. 
    Let 
    \begin{align*}
    s_p & = \bm{c}^T\mathbb{V}^{-\frac{1}{2}}(\sum_j \hat \bgamma_j \hat \Gamma_j \sigma_{Yj}^{-2} - \hat M_j\bbeta_0 + V_j\bbeta_0) \\
            & = \sum_{j}\bm{c}^T\mathbb{V}^{-\frac{1}{2}}\{(\hat\bgamma_j\hat\bGamma_j-\hat\bgamma_j\hat\bgamma_j^T\bm{\beta}_0)\sigma_{Yj}^{-2} + V_j\bm{\beta}_0\} \\
            & = \sum_{j}\bm{c}^T\mathbb{V}^{-\frac{1}{2}}\mathbb{V}_j^{\frac{1}{2}}\mathbb{V}_j^{-\frac{1}{2}}\{(\hat\bgamma_j\hat\bGamma_j-\hat\bgamma_j\hat\bgamma_j^T\bm{\beta}_0)\sigma_{Yj}^{-2} + V_j\bm{\beta}_0\} \\
            & = \sum_{j}\bm{d}_j^T\hat{\bm{r}}_j \in \mathbb{R}
    \end{align*}
    where for every $j$, $\bm{d}_j^T = \bm{c}^T\mathbb{V}^{-\frac{1}{2}}\mathbb{V}_j^{\frac{1}{2}}$, $\hat{\bm{r}}_j = \mathbb{V}_j^{-\frac{1}{2}}\{(\hat\bgamma_j\hat\bGamma_j-\hat\bgamma_j\hat\bgamma_j^T\bm{\beta}_0)\sigma_{Yj}^{-2} + V_j\bm{\beta}_0\}$.

    Note that $E(s_p) = \sum_{j}\bm{d}_j^TE(\hat{\bm{r}}_j) = 0$ and $Var(s_p) = \sum_{j}\bm{d}_j^TCov(\hat{\bm{r}}_j)\bm{d}_j = \sum_{j}\bm{c}^T\mathbb{V}^{-\frac{1}{2}}\mathbb{V}_j^{\frac{1}{2}}\mathbb{V}_j^{\frac{1}{2}}\mathbb{V}^{-\frac{1}{2}}\bm{c} = \bm{c}^T\mathbb{V}^{-\frac{1}{2}}(\sum_{j}\mathbb{V}_j)\mathbb{V}^{-\frac{1}{2}}\bm{c} = \bm{c}^T\bm{c} = 1$ because $Cov(\hat{\bm{r}}_j) = I_K$ for every $j$.
    To verify Lindeberg's condition, take any $\epsilon > 0$,
    \begin{align*}
        \sum_{j}E((\bm{d}_j^T\hat{\bm{r}}_j)^2\mathbbm{1}\{|\bm{d}_j^T\hat{\bm{r}}_j| \ge \epsilon\})
      & =    \sum_{j}\bm{d}_j^TE(\hat{\bm{r}}_j\hat{\bm{r}}_j^T\mathbbm{1}\{|\bm{d}_j^T\hat{\bm{r}}_j| \ge \epsilon\})\bm{d}_j \\
        & \le \sum_{j}\bm{d}_j^TE(\hat{\bm{r}}_j\hat{\bm{r}}_j^T\mathbbm{1}\{\hat{\bm{r}}_j^T\hat{\bm{r}}_j > \frac{\epsilon}{\bm{d}_j^T\bm{d}_j})\bm{d}_j \ \text{by Cauchy Schwartz inequality}\\
        & \le \max_j\lambda_{\max}E(\hat{\bm{r}}_j\hat{\bm{r}}_j^T\mathbbm{1}\{\hat{\bm{r}}_j^T\hat{\bm{r}}_j > \frac{\epsilon}{\bm{d}_j^T\bm{d}_j})\sum_{j}\bm{d}_j^T\bm{d}_j \\
        & \le \max_j\lambda_{\max}E(\hat{\bm{r}}_j\hat{\bm{r}}_j^T\mathbbm{1}\{\hat{\bm{r}}_j^T\hat{\bm{r}}_j > \frac{\epsilon}{\max_j\bm{d}_j^T\bm{d}_j})
    \end{align*}
    Let $T_n = \tilde S_n {\rm diag}(\sqrt{\mu_{n1} + p},...,\sqrt{\mu_{nk} + p})$. Note that 
    \begin{align*}
        \bm{d}_j^T\bm{d}_j & = \bm{c}^T\mathbb{V}^{-\frac{1}{2}}\mathbb{V}_j\mathbb{V}^{-\frac{1}{2}}\bm{c} \\
        & \le \lambda_{\max}(\mathbb{V}^{-\frac{1}{2}}\mathbb{V}_j\mathbb{V}^{-\frac{1}{2}}) \\
        & = \lambda_{\max}(\mathbb{V}^{-\frac{1}{2}}T_nT_n^{-1}\mathbb{V}_jT_n^{-T}T_n^T\mathbb{V}^{-\frac{1}{2}}) \\
        & \le \lambda_{\max}\{(\mathbb{V}^{\frac{1}{2}}T_n^{-T}T_n^{-1}\mathbb{V}^{\frac{1}{2}})^{-1}\}\lambda_{\max}\{T_n^{-1}\mathbb{V}_jT_n^{-T}\} \\ 
        & {\le {\rm trace}\{(\mathbb{V}^{\frac{1}{2}}T_n^{-T}T_n^{-1}\mathbb{V}^{\frac{1}{2}})^{-1}\} \lambda_{\max}\{T_n^{-1}\mathbb{V}_jT_n^{-1}\}} \\
        & {  = {\rm trace}\{\mathbb{V}^{-\frac{1}{2}}T_nT_n^{T}\mathbb{V}^{-\frac{1}{2}}\} \lambda_{\max}\{T_n^{-1}\mathbb{V}_jT_n^{-1}\}} \\
        & {  = {\rm trace}\{T_n^T \mathbb{V}^{-1} T_n\} \lambda_{\max}\{T_n^{-1}\mathbb{V}_jT_n^{-1}\}} \\
        & {  = {\rm trace}\{(T_n^{-1} \mathbb{V} T_n^{-T})^{-1}\} \lambda_{\max}\{T_n^{-1}\mathbb{V}_jT_n^{-1}\}} \\
        &  { \le {\rm trace}\{(T_n^{-1} \mathbb{V} T_n^{-T})^{-1}\} \lambda_{\max}(\mathbb{V}_j)/(\mu_{n,min} + p)}
    \end{align*}
    where the first inequality uses the definition of maximum eigenvalue, and the second and the last inequality use Lemma \ref{lemma0}, and the third inequality uses the fact that $\mathbb{V}^{\frac{1}{2}}T_n^{-T}T_n^{-1}\mathbb{V}^{\frac{1}{2}}$ is also positive definite. Because $T_n^{-1} \mathbb{V} T_n^{-T}$ is bounded by Lemma \ref{lemma1}, ${\rm trace}\{(T_n^{-1} \mathbb{V} T_n^{-T})^{-1}\}$ is bounded between 0 and $\infty$. $\lambda_{\max}(\mathbb{V}_j)/(\mu_{n,min} + p) \rightarrow 0$ by the assumption that $\max_j \frac{\gamma_{jk}^2}{\sigma_{Xjk}^2}/(\mu_{n,min} + p) \rightarrow 0$ for every $k = 1,..., K$ i.e. all diagonal entries of $\mathbb{V}_j$ grow at a rate slower than $\mu_{n,min} + p$. Thus, $\max_j\bm{d}_j^T\bm{d}_j  = o(1)$. Since $E(\hat{\bm{r}_j}\hat{\bm{r}_j}^T) = I$, we have
    \begin{align*}
    E(\hat{\bm{r}}_j\hat{\bm{r}}_j^T\mathbbm{1}\{|\hat{\bm{r}}_j^T\hat{\bm{r}}_j| > \frac{\epsilon}{\max_j\bm{d}_j^T\bm{d}_j}) \rightarrow \bm{0}
    \end{align*} for every $j$ by the dominated convergence theorem. Because eigenvalues of a matrix are continuous functions of its entries, we have 
    \begin{align*}
    { \max_j\lambda_{\max}E(\hat{\bm{r}}_j\hat{\bm{r}}_j^T\mathbbm{1}\{|\hat{\bm{r}}_j^T\hat{\bm{r}}_j| > \frac{\epsilon}{\max_j\bm{d}_j^T\bm{d}_j}) \rightarrow 0 } 
    \end{align*} 
    Therefore, by the Lindeberg-CLT theorem, we have $s_p \xrightarrow[]{D} N(0,1)$ holds as $p \rightarrow \infty$ for any $\bm{c} \in \mathbb{R}^K$ such that $\bm{c}^T\bm{c} = 1$.
    The desired result follows by applying the Cramer-Wald device theorem.
\end{proof}

\begin{lemma}\label{lemma6}
     Under Assumptions 1-2, if $\mu_{n,min}/\sqrt{p} \rightarrow \infty$,
     $S_n^{-1}(\sum_j \hat M_j - (M_j + V_j))S_n^{-T} \xrightarrow[]{P} \bm 0$.
\end{lemma}

\begin{proof}
The proof is similar to that of Lemma \ref{lemma2}. 
Recall the definitions $\tilde \bgamma_j = \tilde S_n^{-1} \bgamma_j$ and $\tilde M_j = \tilde \bgamma_j \tilde \bgamma_j^T \sigma_{Y_j}^{-2}$, and  $\hat{\tilde {\bgamma_j}} = \tilde S_n^{-1} \hat\bgamma_j$, which satisfies  $\hat{\tilde {\bgamma_j}} \sim N(\tilde \bgamma_j, \tilde \Sigma_{Xj})$ and
\begin{align*}
    \tilde \Sigma_{Xj} & = \tilde S_n^{-1} \Sigma_{Xj} \tilde S_n^{-T} = \begin{pmatrix}
\tilde \sigma_{j11}^2 & \dots & \tilde \sigma_{j1K}^2\\
\vdots & \ddots & \vdots \\
\tilde \sigma_{jK1}^2 & \dots & \tilde \sigma_{jKK}^2 
\end{pmatrix}.
\end{align*}

Note that 
\begin{flalign*}
    & S_n^{-1}(\sum_j \hat M_j - (M_j + V_j))S_n^{-T} = \\
    & {\rm diag}(\sqrt{\frac{1}{\mu_{n1}}},...,\sqrt{\frac{1}{\mu_{nK}}}) \tilde S_n^{-1} \{\sum_j \hat M_j - (M_j + V_j))\} \tilde S_n^{-T} {\rm diag}(\sqrt{\frac{1}{\mu_{n1}}},...,\sqrt{\frac{1}{\mu_{nK}}}) \\
    & = {\rm diag}(\sqrt{\frac{1}{\mu_{n1}}},...,\sqrt{\frac{1}{\mu_{nK}}}) \{\sum_j \hat{ \tilde{ \bgamma}}_j \hat{\tilde{ \bgamma}}_j^T \sigma_{Yj}^{-2} -  \sum_j (\tilde \bgamma_j \tilde \bgamma_j^T + \tilde \Sigma_{Xj}) \}{\rm diag}(\sqrt{\frac{1}{\mu_{n1}}},...,\sqrt{\frac{1}{\mu_{nK}}}) .
\end{flalign*}

Consider its $(t,s)$ element of the above matrix. For any $\epsilon > 0$, applying the Chebyshev's inequality, we have
\begin{align*}
    & P(|\frac{1}{\sqrt{\mu_{nt}\mu_{ns}}}\sum_{j}\{\hat{\tilde{\gamma}}_{jt}\hat{\tilde{\gamma}}_{js} - ( \tilde \gamma_{jt} \tilde \gamma_{js} + \tilde \sigma_{jts}^2\sigma_{Yj}^{-2}\}| \ge \epsilon)
\le \frac{1}{\epsilon^2}Var(\frac{1}{\sqrt{\mu_{nt}\mu_{ns}}}\sum_{j}\hat{\tilde{\gamma}}_{jt}\hat{\tilde{\gamma}}_{js}\sigma_{Yj}^{-2}) 
\end{align*}
Using the moment properties of normal distributions and Cauchy-Schwarz inequality (see Lemma \ref{lemma2} for more details), we can show that
\begin{align*}
  & Var(\frac{1}{\sqrt{\mu_{nt}\mu_{ns}}}\sum_{j}\hat{\tilde{\gamma}}_{jt}\hat{\tilde{\gamma}}_{js}\sigma_{Yj}^{-2}) \le  O(\frac{p(\tilde \kappa_t + \tilde \kappa_s)+p}{\mu_{nt}\mu_{ns}})
\end{align*}
where $p\tilde \kappa_t = \sum_j \tilde \gamma_{jt}^2\sigma_{Yj}^{-2}$. Because $p\tilde \kappa_k = \Theta(\mu_{nk})$ for every $k = 1,...,K$ (for the same reasoning as in the proof of Lemma \ref{lemma2}), $O(\frac{p(\tilde \kappa_t + \tilde \kappa_s)+p}{\mu_{nt}\mu_{ns}})$ can be simplified to $O(\frac{\mu_{nt} + \mu_{ns} + p}{\mu_{nt}\mu_{ns}})$. Now because $\mu_{n,min}$ grows faster than $\sqrt{p}$, we have $Var(\sum_{j}\hat{\tilde{\gamma}}_{jt}\hat{\tilde{\gamma}}_{js}\sigma_{Yj}^{-2}/\sqrt{\mu_{nt}\mu_{ns}})\rightarrow 0$ for every $(t,s)$ entry. This completes the proof.
\end{proof}

\begin{lemma} \label{lemma7}
    Under Assumptions 1-2, if $\mu_{n,min}/\sqrt{p} \rightarrow \infty$, $\mathcal{V}_{\rm SRIVW} = (\sum_j M_j)^{-1}\mathbb{V}(\sum_j M_j)^{-1} \rightarrow \bm0$ 
\end{lemma}

\begin{proof}
    Let $T_n = \tilde S_n {\rm diag}(\sqrt{\mu_{n1} + p},...,\sqrt{\mu_{nk} + p})$. Note that
    \begin{align*}
    & (\sum_j M_j)^{-1}\mathbb{V}(\sum_j M_j)^{-1} \\
    & = S_n^{-T} \left(S_n^{-1} (\sum_j M_j) S_n^{-T}\right)^{-1} \left(S_n^{-1} \mathbb{V} S_n^{-T}\right) \left(S_n^{-1} (\sum_j M_j ) S_n^{-T}\right)^{-1} S_n^{-1} \\
    & = \tilde S_n^{-T} \underbrace{{\rm{{\rm diag}}}(\frac{1}{\sqrt{\mu_{n1}}},...,\frac{1}{\sqrt{\mu_{nK}}}) \left(S_n^{-1} (\sum_j M_j) S_n^{-T}\right)^{-1} {\rm{{\rm diag}}}(\sqrt{\frac{\mu_{n1}+p}{\mu_{n1}}},...,\sqrt{\frac{\mu_{nK}+p}{\mu_{nK}}})}_{A_1} \left(T_n^{-1} \mathbb{V} T_n^{-T}\right) \\
    & \underbrace{{\rm{{\rm diag}}} (\sqrt{\frac{\mu_{n1}+p}{\mu_{n1}}},...,\sqrt{\frac{\mu_{nK}+p}{\mu_{nK}}}) \left(S_n^{-1} (\sum_j M_j) S_n^{-T}\right)^{-1} {\rm{{\rm diag}}}(\frac{1}{\sqrt{\mu_{n1}}},...,\frac{1}{\sqrt{\mu_{nK}}})}_{A_1^T} \tilde S_n^{-1} 
    \end{align*}
    We know $\tilde S_n^{-T}$, $\tilde S_n^{-1}$, and $(S_n^{-1} (\sum_j M_j) S_n^{-T})^{-1}$ are all bounded by Assumption 2. By Lemma \ref{lemma1}, $T_n^{-1}\mathbb{V}T_n^{-T}$ is also bounded. It then sufficies to show $A_1$ and $A_1^T$ are $o(1)$. Consider any $(t,s)$ entry of the matrix $A_1$, for $t,s = 1,...,K$. $A_{1[t,s]} = O(\sqrt{(\mu_{ns} + p)/(\mu_{ns}\mu_{nt})})$, which is $o(1)$ if $\mu_{n,min}/\sqrt{p}\rightarrow\infty$. Thus, $(\sum_j M_j)^{-1}\mathbb{V}(\sum_j M_j)^{-1} \rightarrow \bm 0$.
\end{proof}

\section{Extension to balanced horizontal pleiotropy}\label{sec: extension to bhp}

 In this section, we extend SRIVW to accommodate one specific form of pleiotropy in MR, known as balanced horizontal pleiotropy. The resulting estimator is referred to as SRIVW-pleio. Specifically, we generalize our outcome model to 
 $ Y = \ \bm X^T \bm\beta_0 + \bmZ^T \bm\alpha + g(\bm U,E_Y)$, where $\bm\alpha= (\alpha_1,\dots, \alpha_p)^T$ is the vector of the unknown direct effects of instruments on the outcome and each $\alpha_j$ is assumed to follow a mean-zero normal distribution with unknown variance $\tau_0^2$. This model implies that $ \Gamma_j = \bm \gamma_j^T\bm\beta_0 + \alpha_j  $ holds for every $j$. We further make the following assumption which modifies the Assumption 1 to accommodate balanced horizontal pleiotropy.
\begin{assumption}\label{assump: 3}
Suppose Assumption 1 holds, except that $\alpha_j \sim N(0, \tau_0^2)$ and conditional on $\alpha_j$, $\hat{\Gamma}_j \sim N(\alpha_j+\bgamma_j^T\bbeta_0, \sigma_{Yj}^2)$ for every $j$,  and there exists a positive constant $c_{+}$ such that $\tau_0 \le c_{+}\sigma_{Yj}$ across $j$. 
\end{assumption}

 The theorem below establishes the consistency and asymptotic normality of SRIVW-pleio under balanced horizontal pleiotropy.
\begin{theorem} \label{theo: divw}
Assume Assumptions 2-3, ${\mu_{n,\min}}/{\sqrt{p}}\rightarrow \infty$, and $\max_j \gamma_{jk}^2\sigma_{Xjk}^{-2}/(\mu_{n,\min} + p) \rightarrow 0$ for all $k$ as $n\to\infty$. If $\phi = O_P(\mu_{n,\min}+p)$, 
    then $\hat{\bbeta}_{\rm SRIVW-pleio, \phi} $ is consistent and asymptotically normal, i.e., as $n\to \infty$,
\begin{align*}
    \mathbb{U}^{-\frac{1}{2}}\bigg\{ \sum_{j=1}^{p}M_j \bigg\}( \hat{\bbeta}_{{\rm SRIVW-pleio}, \phi} - \bbeta_0) \xrightarrow[]{d} N(\bm0, I_{K}),
\end{align*}
where $\mathbb{U} = \sum_{j=1}^p\{(1 + \tau_0^2\sigma_{Yj}^{-2} + \bm{\beta}_0^TV_j\bm{\beta}_0)(M_j + V_j) + V_j\bm{\beta}_0\bm{\beta}_0^TV_j\}$.
\end{theorem}
The proof is provided in the section S3.11. This theorem suggests that the asymptotic variance of SRIVW-pleio is given by $\bigg \{\sum_{j=1}^p M_j\bigg\}^{-1} \mathbb{U} \bigg \{\sum_{j=1}^p M_j\bigg\}^{-1}$. We show in the section S3.11 that a consistent variance estimator is given by
\begin{align*}
    \hat {\mathcal{V}}_{\rm SRIVW-pleio} = \bigg\{R_{\phi}(\sum_j \hat M_j - V_j) \bigg\}^{-1}\hat{\mathbb{U}}_{\phi} \bigg\{R_{\phi}(\sum_j \hat M_j - V_j)\bigg\}^{-1}
\end{align*}
where 
\begin{align*}
     &\hat{\mathbb{U}}_{\phi} = \sum_{j=1}^p\left\{\left(1+\hat{\tau}_{\phi}^2\sigma_{Yj}^{-2}+\hat{\bm\beta}_{{\rm SRIVW-pleio},\phi}^TV_j\hat{\bm\beta}_{{\rm SRIVW-pleio},\phi}\right)\hat{M_j}+V_j\hat{\bm\beta}_{{\rm SRIVW-pleio},\phi}\hat{\bm\beta}_{{\rm SRIVW-pleio},\phi}^TV_j\right\}, \quad \text{and} \\
    &  \hat{\tau}^2_{\phi}   = \frac{1}{\sum_{j=1}^p\sigma_{Yj}^{-2}}\left\{\sum_{j=1}^{p} \left( (\hat \Gamma_j - \hat{\bgamma_j}^T\hat{\bbeta}_{{\rm SRIVW-pleio}, \phi})^2 - \sigma_{Yj}^2 - \hat{\bm\beta}_{{\rm SRIVW-pleio},\phi}^T\Sigma_{Xj}\hat{\bm\beta}_{{\rm SRIVW-pleio},\phi} \right)\sigma_{Yj}^{-2} \right\}.
\end{align*}

To choose $\phi$ under balanced horizontal pleiotropy, we use the modified \textit{Q}-statistic (\citeSupp{sanderson2021testing}), denoted as $Q_+$,  as the objective function:
\begin{align*}
    Q_+ (\phi) = \sum_{j=1}^{p} \frac{(\hat \Gamma_j - \hat \bgamma_j^T\hat \bbeta_{\rm SRIVW-pleio, \phi})^2}{\sigma_{Yj}^2 + \hat \bbeta_{\rm SRIVW-pleio,\phi}^T \Sigma_{Xj} \hat \bbeta_{\rm SRIVW-pleio,\phi} + \hat \tau_{\phi}^2} ,\ \text{ subject to $\phi \in B$}
\end{align*}
where $\hat \tau_{\phi}$ depends on $\phi$ through $\hat \bbeta_{\rm SRIVW-pleio,\phi} $, and in our implementation, $B$ take the same form as that provided in the Section 5.3 of the main text.

A simulation study evaluating the performance of the SRIVW-pleio estimator under balanced horizontal pleiotropy is available in the Section S4.2.

\section{Proofs}\label{sec: proofs}

\subsection{The rate of $\sigma_{Xjk}^2$ and $\sigma_{Yj}^2$, and the  boundedness of variance ratios $\sigma_{Xjk}^2/\sigma_{Yj}^2$}\label{sec: rate sigma}

In Section 2 of the main text, we state that when only common variants are used, every entry in $\Sigma_{Xj}$ and $ \sigma_{Yj}^2, j=1,\dots, p$ are $\Theta (1/n)$, where $n=\min (n_Y, n_{X1}, \dots,  n_{XK}) $. We now provide justifications. Similar argument can be used to show that when rare variants are used, entries in $\Sigma_{Xj}$ and $\sigma_{Yj}^2$ converge to 0 at a rate slower than $1/n$.

From marginal regression, we have
\begin{align}\label{eq: sigma xjk}
    \sigma_{Yj}^2 = \frac{Var(Y) - \Gamma_j^2 Var(Z_j)}{n_Y Var(Z_j)},\ \sigma_{Xjk}^2 = \frac{Var(X_k) - \gamma_{jk}^2 Var(Z_j)}{n_{Xk} Var(Z_j)}
\end{align}
Because $Var(Y) \ge \sum_j \Gamma_{j}^2 Var(Z_j)$ and $Var(X_k) \ge \sum_j \gamma_{jk}^2 Var(Z_j)$ holds for every $k = 1,...,K$, the approximations $\sigma_{Yj}^2 \approx \frac{Var(Y)}{n_Y Var(Z_j)}$ and $\sigma_{Xjk}^2 \approx \frac{Var(X_k)}{n_{Xk} Var(Z_j)}$ for $k = 1,..., K$ are reasonable when $p$ is large and no single IV dominates the summation over $j$. 

When only common variants are used, $Var(Z_j) = \Theta(1)$. $Var(Y)$ and $Var(X_k), k = 1,...,K$ are also $\Theta(1)$. Provided that the sample sizes $n_Y, n_{Xk}, k=1,\dots, K$ diverge to infinity at the same rate,  which is assumed in Assumption 1(i), we get $\sigma_{Yj}^2 = \Theta(1/n)$ and $\sigma_{Xjk}^2 = \Theta(1/n)$. In the general case with either common or rare variants, we can see that the ratio between $  \sigma_{Yj}^2 $ and $\sigma_{Xjk}^2$ is bounded between 0 and infinity, provided that $Var(Y)$ and $Var(X_k), k = 1,...,K$ are bounded and the sample sizes $n_Y, n_{Xk}, k=1,...,K$ diverge to infinity at the same rate. This makes it reasonable to assume the boundedness of the variance ratios $\sigma_{Xjk}^2/\sigma_{Yj}^2$.

\subsection{Existence of an $S_n$ satisfying Assumption 2}

In this subsection, we want to show there is always an $S_n$ matrix that satisfies Assumption 2.

Take the eigen-decomposition of $\sum_j \bgamma_j \bgamma_j^T \sigma_{Yj}^{-2} = U_n \Lambda_n U_n^T$, where $\Lambda_n = {\rm diag}(\lambda_{n1},...,\lambda_{nK})$ and $\lambda_{n1}\ge ...\ge \lambda_{nK}$. We argue that $S_n = U_n {\rm diag}(\sqrt{\lambda_{n1}},...,\sqrt{\lambda_{nK}})$ satisfies the two conditions in Assumption 2.

First, $\tilde S_n = U_n$ is bounded elementwise because $U_n$ is the matrix of eigenvectors. Because $U_n U_n^T = I_K$, the smallest eigenvalue of $\tilde S_n \tilde S_n^T$ is always 1, and hence is bounded away from 0.

Second, $S_n^{-1} \sum_j \bgamma_j \bgamma_j^T \sigma_{Yj}^{-2} S_n^{-T} = I_K$. Hence, it is bounded and its minimum eigenvalue is bounded away from 0.

\subsection{Proof of claims related to the illustrative Example 1 and 2, and the minimum conditional $F$-statistic}

\textbf{Positive definiteness of $S_n^{-1} \bigg\{ \sum_{j=1}^{p} \bgamma_j \bgamma_j^T \sigma_{Yj}^{-2} \bigg\} S_n^{-T}$ in the illustrative Example 1}

In Example 1 of the main text, two exposures $X_1$ and $X_2$ are both strongly predicted by a common set of SNPs, but either one of them has weak conditional associations with the SNPs given the other. To complete this example, we show that when $\mu_{n1} =p$ and $\mu_{n2} = n$, 
\[   S_n^{-1} \bigg\{ \sum_{j=1}^{p} \bgamma_j \bgamma_j^T \sigma_{Yj}^{-2}\bigg\} S_n^{-T}=\frac{1}{\var(Y)}
\begin{pmatrix}
        \frac{c^2p }{\mu_{n1}} & \frac{c \sqrt{n} \sum_{j=1}^{p} \gamma_{j1}}{\sqrt{\mu_{n1}\mu_{n2}}} \\
        \frac{c\sqrt{n}  \sum_{j=1}^{p} \gamma_{j1}}{\sqrt{\mu_{n1}\mu_{n2}}} & \frac{n \sum_{j=1}^{p}\gamma_{j1}^2}{\mu_{n2}}
    \end{pmatrix}= 
    \frac{1}{\var(Y)}
\begin{pmatrix}
        c^2 & \frac{c  \sum_{j=1}^{p} \gamma_{j1}}{\sqrt{p}} \\
       \frac{c  \sum_{j=1}^{p} \gamma_{j1}}{\sqrt{p}} &\sum_{j=1}^{p}\gamma_{j1}^2
    \end{pmatrix} 
\]
is bounded element-wise with the smallest eigenvalue bounded away from 0. 

The element-wise boundedness is from the condition that $\gamma_{j1}\sim N(0,1/p)$. 
Moreover, the determinant of $S_n^{-1} \bigg\{ \sum_{j=1}^{p} \bgamma_j \bgamma_j^T \sigma_{Yj}^{-2} \bigg\} S_n^{-T}$ is 
$$
\frac{c^2}{\var(Y)}
\left[
\sum_{j=1}^p\gamma_{j1}^2 - \frac{1}{p} (\sum_{j=1}^p\gamma_{j1})^2
\right]\geq 0
$$
by the Cauchy-Schwarz inequality, with equality holds when all $\gamma_{j1}$'s are equal to each other, which cannot be true under the condition that $\gamma_{j1}$'s are drawn from $N(0,1/p)$. Hence, the determinant must be positive, implying that the smallest eigenvalue of that matrix is bounded away from 0.

\textbf{Positive definiteness of $S_n^{-1} \bigg\{ \sum_{j=1}^{p} \bgamma_j \bgamma_j^T \sigma_{Yj}^{-2} \bigg\} S_n^{-T}$ in the illustrative Example 2}

For this example, we let $\mu_{n1} = p$, $\mu_{n2} = n$ and $\mu_{n3} = n$,
\begin{align*}
    \tilde S_n^{-1} = 
    \begin{bmatrix}
        1 & -1 & 0 \\
        1 & 1 & 1 \\
        1 & 1 & -2
    \end{bmatrix}.
\end{align*}

Then, we have
\begin{align*}
    & S_n^{-1} \bigg\{ \sum_{j=1}^{p} \bgamma_j \bgamma_j^T \sigma_{Yj}^{-2}\bigg\} S_n^{-T} \\
    & =\frac{1}{\var(Y)}
    \begin{bmatrix}
        \frac{c^2p}{\mu_{n1}} & -\frac{c\sqrt{n}\sum_j (3\gamma_{j1}+\epsilon_j)}{\sqrt{\mu_{n1}\mu_{n2}}} + \frac{1.5c^2}{\sqrt{\mu_{n1}\mu_{n2}}} & -\frac{c\sqrt{n}\sum_j \epsilon_j}{\sqrt{\mu_{n1}\mu_{n3}}} \\
        -\frac{c\sqrt{n}\sum_j (3\gamma_{j1}+\epsilon_j)}{\sqrt{\mu_{n1}\mu_{n2}}} + \frac{1.5c^2}{\sqrt{\mu_{n1}\mu_{n2}}} & \frac{n\sum_j(3\gamma_{j1}+\epsilon_j+1.5\frac{c}{\sqrt{n}})^2}{\mu_{n2}} & \frac{n\sum_j(3\gamma_{j1} + \epsilon_j) \epsilon_j}{\sqrt{\mu_{n2}\mu_{n3}}} + \frac{c\sqrt{n}\sum_j \epsilon_j}{\sqrt{\mu_{n2}\mu_{n3}}} \\
        -\frac{c\sqrt{n}\sum_j \epsilon_j}{\sqrt{\mu_{n1}\mu_{n3}}} & \frac{n\sum_j(3\gamma_{j1} + \epsilon_j) \epsilon_j}{\sqrt{\mu_{n2}\mu_{n3}}} + \frac{c\sqrt{n}\sum_j \epsilon_j}{\sqrt{\mu_{n2}\mu_{n3}}} & \frac{n\sum_j \epsilon_j^2}{\mu_{n3}}
    \end{bmatrix} \\
    & = \frac{1}{\var(Y)}\begin{bmatrix}
        c^2 & o(1) & o(1) \\
        o(1) & \sum_j (3\gamma_{j1} + \epsilon_j + o(1))^2 & \sum_j (3\gamma_{j1} +\epsilon_j)\epsilon_j + o(1) \\
        o(1) & \sum_j (3\gamma_{j1} +\epsilon_j)\epsilon_j + o(1) & \sum_j \epsilon_j^2
    \end{bmatrix}
\end{align*}
Under the conditions that $\sum_j \gamma_{j1} = O(1)$, $\sum_j \epsilon_j = O(1)$, $\sum_j \gamma_j\epsilon_j = O(1)$,$\sum_j \gamma_{j1}^2 = \Theta(1)$, $\sum_j \epsilon_j^2 = \Theta(1)$, we have $\sum_j(3\gamma_{j1} + \epsilon_j + o(1))^2 = \Theta(1)$, $\sum_j(3\gamma_{j1}+\epsilon_j)\epsilon_j = \Theta(1)$. These conditions are implied by the assumptions that $\gamma_{j1}$'s and $\epsilon_j$'s are drawn independently from $N(0, 1/p)$. Therefore, $S_n^{-1} \bigg\{ \sum_{j=1}^{p} \bgamma_j \bgamma_j^T \sigma_{Yj}^{-2}\bigg\} S_n^{-T}$ is element-wise bounded. Moreover, the determinant of $S_n^{-1} \bigg\{ \sum_{j=1}^{p} \bgamma_j \bgamma_j^T \sigma_{Yj}^{-2}\bigg\} S_n^{-T}$ is
\begin{align*}
    \frac{c^2}{\var(Y)}\left[(\sum_j\epsilon_j^2)(\sum_j(3\gamma_{j1}+\epsilon_j)^2) - ( \sum_j(3\gamma_{j1}+\epsilon_j)\epsilon_j)^2 \right] \ge 0
\end{align*}
by Cauchy-Schwartz inequality, with equality holds when all $3\gamma_{j1} + \epsilon_j$ is proportional to $\epsilon_j$ for all $j$, which cannot be true under the given conditions. Hence, the determinant must be positive, implying that the smallest eigenvalue is bounded away from 0.

\textbf{Conditional $F$-statistics in Example 2}

In Example 2, because $\gamma_{j2} = \gamma_{j1} + c/\sqrt{n}$ and $\gamma_{j3} = 0.5 \gamma_{j1} + 0.5 \gamma_{j2} + \epsilon_{j}$, it is easy to see that $\delta_{1} = (-1,1,0)$ $\delta_{2} = (1,-1,0)$ and $\delta_3 = (0.5,0.5,-1)$. Notice that $\delta_1$ and $\delta_2$ are in the same direction as the weakest direction corresponding to the first row of $\tilde S_n^{-1}$. Thus,
$F_1 = \Theta(\delta_1^T (\sum_j \bgamma_j \bgamma_j^T \sigma_{Yj}^{-2}) \delta_1)/p = \Theta(1)$ and $F_2 = \Theta(\delta_2^T (\sum_j \bgamma_j \bgamma_j^T \sigma_{Yj}^{-2}) \delta_2)/p = \Theta(1)$. Since $\delta_3$ is in the same direction as the last row $\tilde S_n^{-1}$, we have $F_3 = \Theta(\delta_3^T (\sum_j \bgamma_j \bgamma_j^T \sigma_{Yj}^{-2}) \delta_3)/p = \Theta(n/p) \to \infty$ when $n/p \to \infty$.

\textbf{Lack of coverage of $\beta_{03}$ by the IVW method in Example 2.}

Suppose $\bbeta_0 = (\beta_{01},\beta_{02},\beta_{03}) = (1,1,0)$, and three exposures are measured in independent samples so that the shared correlation matrix $\Sigma = I_K$. Let $C$ be any positive constant. From the proof of Corollary 1 (section S3.9), we have, for any $\epsilon >0$, there exists a large enough constant $C_{\epsilon}$ as $n \to \infty$,
\begin{align*}
         \Pr(\frac{|\hat \beta_{{\rm IVW},3} - \beta_{03}|}{\sqrt{\mathcal{V}_{{\rm IVW},33}}} > C) \ge 1 - \epsilon - \Pr( \frac{|\tilde \beta_{{\rm IVW},3} - \beta_{03}|}{\sqrt{\mathcal{V}_{{\rm IVW},33}}} \le C_{\epsilon} + C)\\
         \ge 1 - \epsilon - \Pr( \frac{|\tilde \beta_{{\rm IVW},3} - \beta_{03}|}{\sqrt{\lambda_{\rm max, IVW}}} \le C_{\epsilon} + C)
    \end{align*} where $\mathcal{V}_{{\rm IVW},33}$ denotes the third diagonal entry of $\mathcal{V}_{{\rm IVW}}$ and
$\lambda_{\rm max, IVW}$ denotes the maximum eigenvalue of $\mathcal{V}_{\rm IVW}$. From the proof of Corollary 1, we know $\lambda_{\rm max, IVW} = \Theta(1/(\mu_{n,min}+p))$. We have just shown that $\mu_{n,min} = p$ in this example. Thus, $\lambda_{\rm max, IVW} = \Theta(1/p)$.

It remains to determine the order of $\tilde \beta_{IVW,3} - \beta_{03}$. Note that by Proposition 1,
\begin{align*}
    \tilde \bbeta_{\rm IVW} - \bbeta_0 = \left(\{ \sum_{j=1}^{p}(M_j+V_j)\}^{-1}\{\sum_{j=1}^p M_j\}\bbeta_0 - \bbeta_0\right) + o_p(r_n) = -\{\sum_{j=1}^{p}(M_j+V_j)\}^{-1}\{ \sum_{j=1}^p V_j \}\bbeta_0 + o_p(r_n)
\end{align*}
where $r_n \to 0$ represents the rate at which $\tilde \bbeta_{\rm IVW}$ converges to $\bbeta_0$ as $n \to \infty$.
We know $\sum_j V_j = {\rm diag}(\Theta(p))$, and
\begin{align*}
    \sum_j M_j = \frac{n}{\var(Y)}
    \begin{bmatrix}
        \sum_j \gamma_{j1}^2 & \sum_j \gamma_{j1} \gamma_{j2} & \sum_j \gamma_{j1} \gamma_{j3} \\
        \sum_j \gamma_{j1}\gamma_{j2} & \sum_j \gamma_{j2}^2  & \sum_j \gamma_{j2} \gamma_{j3} \\
        \sum_j \gamma_{j1}\gamma_{j3} & \sum_j \gamma_{j2}\gamma_{j3}  & \sum_j \gamma_{j3}^2
    \end{bmatrix}
    = \Theta(n\begin{bmatrix}
        1 & 1 & 1 \\
        1 & 1 + s & 1 \\
        1 & 1 & 2
    \end{bmatrix}),
\end{align*} where $s = p/n$. Hence 
\begin{align*}
    \sum_{j=1}^{p}(M_j+V_j) = \Theta(n
    \begin{bmatrix}
        1+s & 1 & 1 \\
        1 & 1+2s & 1 \\
        1 & 1 & 2 + s
    \end{bmatrix}
    ).
\end{align*}
Some algebras show that
\begin{align*}
    \left( \sum_{j=1}^{p}(M_j+V_j) \right)^{-1} = \Theta(\frac{1}{p(3+7s+2s^2)} \begin{bmatrix}
        1+5s +2s^2 & -(1+s) & -2s \\
        -(1+s) & 1+3s+s^2 & -s \\
        -2s & -s & 3s+s^2
    \end{bmatrix}).
\end{align*}
Therefore, $\tilde \beta_{\rm IVW,3} - \beta_{03} = \Theta(\frac{-2sp -sp}{p(3+7s+2s^2)}) + o_p(r_n) = \Theta(s) + o_p(r_n)$, and $\frac{|\tilde \beta_{{\rm IVW},3} - \beta_{03}|}{\sqrt{\lambda_{\rm max, IVW}}} = \Theta(s\sqrt{p}) + o_p(\sqrt{p}r_n) = \Theta(p^{3/2}/n) + o_p(\sqrt{p}r_n)$. This means that even if $\sqrt{p}r_n \to 0$ as $n \to \infty$, $\frac{|\tilde \beta_{{\rm IVW},3} - \beta_{03}|}{\sqrt{\lambda_{\rm max, IVW}}}$ is still not bounded in probability when $p^{3/2}/n \to \infty$ so that we have $\Pr(\frac{|\hat \beta_{{\rm IVW},3} - \beta_{03}|}{\sqrt{\mathcal{V}_{{\rm IVW},33}}} > C)$ approaches 1 as $n \to \infty$. We remark that in the above argument, for simplicity, we replace $\mathcal{V}_{\rm IVW,33}$ by $\lambda_{\max, IVW}$ which is only $\Theta(1/p)$. Because $X_3$ has the conditional $F$-statistic going to $\infty$, it may be even reasonable to assume $\mathcal{V}_{\rm IVW,33} = O(1/n)$. In that case, we only require $p^2/n \to \infty$ to establish the unboundedness of $\tilde \beta_{\rm IVW,3} - \beta_{03}$ in probability.

\textbf{The order of minimum conditional $F$-statistic is $\Theta(\mu_{n,\min}/p)$}

For each $n$ and $k$, let
\begin{align*}
    \delta_{nk} & = \argmin_{\delta_{k} \in \mathbb{R}^{K}} \delta_k^T \left(\sum_j  \bgamma_j \bgamma_j^T \sigma_{Xjk}^{-2} \right)\delta_k\quad \text{subject to $\delta_{k,k} = -1$}
\end{align*}
where $\delta_{k,k}$ denotes the $k$th element of $\delta_{k}$. By Assumption 1(iii), we have
\begin{align*}
    \delta_{nk} \approx \frac{\var(Y)}{\var(X_k)}\argmin_{\delta_{k} \in \mathbb{R}^{K}} \delta_k^T \left(\sum_j  \bgamma_j \bgamma_j^T \sigma_{Yj}^{-2} \right)\delta_k\quad \text{subject to $\delta_{k,k} = -1$}
\end{align*}
The reasonableness of this approximation is discussed in Section S3.1. To solve this optimization problem, we consider the following loss function:
\begin{align*}
    L(\delta_k, \alpha) = \delta_k^T A_n \delta_k + \alpha(\delta_{k,k} +1)
\end{align*}
where $A_n = \sum_j \bgamma_j \bgamma_j^T \sigma_{Yj}^{-2}$, $\alpha$ is the Lagrangian muliplier. Taking the derivative and setting it to zero, we have the minimizer
\begin{align*}
    \delta_{nk} = -\frac{A_n^{-1}e_k}{[A_n^{-1}]_{kk}}
\end{align*}
where $e_k$ is the unit vector with $k$th element being 1. Then, we can compute $\delta_{nk}^T A_n \delta_{nk}$, denoted as $b_{nk}$:
\begin{align*}
    b_{nk} & = \delta_{nk}^T A_n \delta_{nk} = \frac{e_k^T A_n^{-1} e_k}{[A_n^{-1}]_{kk}^2}
\end{align*}
Because $e_k^T A_n^{-1} e_k = [A_n^{-1}]_{kk}^2$, we have $b_{nk} = \frac{1}{[A_n^{-1}]_{kk}} \ge \lambda_{nK}$, where $\lambda_{n1}\ge \lambda_{n2} \ge \dots \ge \lambda_{nK}$ are eigenvalues of $A_n$. From this, we have
\begin{align*}
    \min_k b_{nk} \ge \lambda_{nK}
\end{align*}
On the other hand, we want to show that
\begin{align*}
    \min_k b_{nk} \le \frac{\lambda_{nK}}{C}
\end{align*}
for some constant $C$. Note that
\begin{align*}
    [A_n^{-1}]_{kk} & = e_k^T A_n^{-1} e_k = \sum_{l=1}^K \frac{u_{nl,k}^2}{\lambda_{nl}} \ge \frac{u_{nK,k}^2}{\lambda_{nK}} = \max \{0, \frac{u_{nK,k}^2}{\lambda_{nK}}\}
\end{align*}
where $u_{nl,k}$ is the $k$th component of the $l$-th eigenvector of $A_n$. Thus,
\begin{align*}
    b_{nk} = \min \{\infty, \frac{\lambda_{nK}}{u_{nK,k}^2} \}
\end{align*}
and hence
\begin{align*}
    \min_k b_{nk} & = \min_k \min\{\infty,  \frac{\lambda_{nK}}{u_{nK,k}^2}\} \\
    & = \min\{\infty,  \min_k \frac{\lambda_{nK}}{u_{nK,k}^2}\}
\end{align*}
Now, note that the eigenvectors are normalized so that $\sum_{k=1}^K u_{nK, k}^2 = 1$. There must be at least one $k$ for which $u_{nK,k}^2 \ge C$ for some constant $0 <C \le 1$ for every $n$. Then, we have
\begin{align*}
    \min_k b_{nk} & \le \frac{\lambda_{nK}}{C}.
\end{align*}
Thus, $\min_k b_{nk} = \Theta(\lambda_{nK})$. Finally, because $F_k = \frac{b_{nk}}{p-(K-1)}$ and $\mu_{n,\min} = \Theta(\lambda_{nK})$ as shown in Section S3.5, this implies that $\mu_{n,\min}/p$ can be interpreted as the rate at which the minimum conditional $F$-statistic across $K$ exposures grows as $n \rightarrow \infty$.

\subsection{Connection between $\sum_j \bgamma_j \bgamma_j^T\sigma_{Yj}^{-2}$ and $\sum_j \Omega_j^{-1} \bgamma_j\bgamma_j^T\Omega_j^{-T}$}\label{connection two matrices}

We establish the connection between $\sum_j \bgamma_j \bgamma_j^T\sigma_{Yj}^{-2}$ and $\sum_j \Omega_j^{-1} \bgamma_j\bgamma_j^T\Omega_j^{-T}$. The latter is referred to as the IV strength matrix in Section 3 of the main text. Specifically, we aim to show that for any $S_n$ satisfying Assumption 2, if $S_n^{-1}\left(\sum_j \bgamma_j \bgamma_j^T\sigma_{Yj}^{-2}\right) S_n^{-T}$ is bounded element-wise and its minimum eigenvalue is bounded away from zero, then $Q_n^{-1}\left(\sum_j\Omega_j^{-1} \bgamma_j\bgamma_j^T\Omega_j^{-T}\right)Q_n^{-T}$ is also bounded element-wise and its minimum eigenvalue is bounded away from 0, where 
\begin{align*}
  Q_n = \tilde Q_n{\rm {\rm diag}}(\sqrt{\mu_{n1}},...,\sqrt{\mu_{nK}}) = \Sigma^{-\frac{1}{2}}D_n^{-1}\tilde S_n {\rm {\rm diag}}(\sqrt{\mu_{n1}},...,\sqrt{\mu_{nK}})
\end{align*}
and
\begin{align*}
    D_n = {\rm {\rm diag}}(\sqrt{\frac{n_YVar(X_1)}{n_{X1}Var(Y)}},...,\sqrt{\frac{n_YVar(X_K)}{n_{XK}Var(Y)}}),
\end{align*}
and $\Sigma$ is the shared correlation matrix.

We now justify this claim. First, the matrix $\tilde Q_n$ satisfies Assumption 2(i). $\tilde Q_n$ is bounded because $\Sigma^{-\frac{1}{2}}$, $D_n^{-1}$ and $\tilde S_n$ are all bounded. The boundedness of $D_n^{-1}$ is by Assumption 1 that all sample sizes diverge to infinity at the same rate. Furthermore, $\tilde Q_n \tilde Q_n^{T}$ has the minimum eigenvalue bounded away from 0. This is because,
\begin{align*}
    \lambda_{\min}(\tilde Q_n \tilde Q_n^{T}) & = \lambda_{\min}(\Sigma^{-\frac{1}{2}}D_n^{-1}\tilde S_n \tilde S_n^{T}D_n^{-1}\Sigma^{-\frac{1}{2}}) \\
    & \ge \lambda_{\min}(\tilde S_n \tilde S_n^{T})\lambda_{\min}(D_n^{-1}D_n^{-1})\lambda_{\min}(\Sigma^{-1}) > 0
\end{align*}
where we use Lemma \ref{lemma0} twice and Assumption 2(i) that $\tilde S_n \tilde S_n^{T}$ has minimum eigenvalue bounded away from 0, and that $D_n$ and $\Sigma$ are both positive definite and bounded.

Next, we show that $Q_n^{-1}\left(\sum_j\Omega_j^{-1} \bgamma_j\bgamma_j^T\Omega_j^{-T}\right)Q_n^{-T}$ is bounded element-wise and its minimum eigenvalue is bounded away from 0. Note that
\begin{align*}
    Q_n^{-1}\left(\sum_j \Omega_j^{-1}\bgamma_j \bgamma_j^T \Omega_j^{-T}\right) Q_n^{-T} & = Q_n^{-1} \left(\sum_j (\Omega_j\sigma_{Yj}^{-1})^{-1} \bgamma_j \bgamma_j^T\sigma_{Yj}^{-2} (\Omega_j\sigma_{Yj}^{-1})^{-1}\right) Q_n^{-T} \\
    & = S_n^{-1}D_n\Sigma^{\frac{1}{2}} \left(\sum_j \Sigma^{-\frac{1}{2}} D_j^{-1} \bgamma_j \bgamma_j^T\sigma_{Yj}^{-2} D_j^{-1} \Sigma^{-\frac{1}{2}}\right) \Sigma^{\frac{1}{2}}D_n S_n^{-T} \\
    & = S_n^{-1} D_n\left(\sum_j D_j^{-1} \bgamma_j \bgamma_j^T\sigma_{Yj}^{-2} D_j^{-1} \right) D_n S_n^{-T}
\end{align*}
where $D_j = {\rm {\rm diag}}(\sigma_{Xj1}/\sigma_{Yj},...,\sigma_{XjK}/\sigma_{Yj})$.
By Assumption 1 and the result from Section \ref{sec: rate sigma}, $\sigma_{Xjk}/\sigma_{Yj} \approx \sqrt{(n_YVar(X_k))/(n_{Xk}Var(Y))}$ and it is bounded between 0 and infinity for every $j$ and $k$. That is, we have $D_j \approx D_n$ for every $j$. Hence, 
\begin{align*}
    Q_n^{-1}\left(\sum_j \Omega_j^{-1}\bgamma_j \bgamma_j^T \Omega_j^{-T}\right) Q_n^{-T} & = S_n^{-1} D_n\left(\sum_j D_j^{-1} \bgamma_j \bgamma_j^T\sigma_{Yj}^{-2} D_j^{-1} \right) D_n S_n^{-T} \\
    & \approx S_n^{-1} \left(\sum_j \bgamma_j \bgamma_j^T\sigma_{Yj}^{-2} \right) S_n^{-T}
\end{align*}
The boundedness of entries in $Q_n^{-1}\left(\sum_j \Omega_j^{-1}\bgamma_j \bgamma_j^T \Omega_j^{-T}\right) Q_n^{-T}$ and the boundedness of its minimum eigenvalue follows from the condition that $S_n^{-1}\left(\sum_j \bgamma_j \bgamma_j^T\sigma_{Yj}^{-2}\right) S_n^{-T}$ is bounded element-wise and its minimum eigenvalue is bounded away from zero.

Note that the sequence $\mu_{n1},...,\mu_{nK}$ are shared in the definition of $S_n$ and $Q_n$. Therefore, the above result suggests that $\sum_j \bgamma_j \bgamma_j^T\sigma_{Yj}^{-2}$ and $\sum_j \Omega_j^{-1} \bgamma_j\bgamma_j^T\Omega_j^{-T}$ have the same rates of divergence.

\subsection{Proof of the minimum eigenvalues of $\sum_j \bgamma_j \bgamma_j^T \sigma_{Yj}^{-2}$ 
 and\\ $\sum_j \Omega_j^{-1} \bgamma_j\bgamma_j^T\Omega_j^{-T}$ are $\Theta(\mu_{n,min})$} \label{proof: minimum eigenvalue}

In this subsection, we prove that the minimum eigenvalue $\lambda_{n,min}$ of $\sum_j \bgamma_j \bgamma_j^T \sigma_{Yj}^{-2}$ is $\Theta(\mu_{n,min})$ by contradiction.

First take the eigen-decomposition of $\sum_j \bgamma_j \bgamma_j^T \sigma_{Yj}^{-2} = U_n\Lambda_n U_n^T$, where $\Lambda_n = {\rm {\rm diag}}(\lambda_{n1},...,\lambda_{nK})$ and $\lambda_{n1}\ge...\ge\lambda_{nK} = \lambda_{n,min}$. Then,   
\begin{align*}
    S_n^{-1} \left( \sum_j \bgamma_j \bgamma_j^T \sigma_{Yj}^{-2}\right) S_n^{-T} & = {\rm {\rm diag}}(\frac{1}{\sqrt{\mu_{n,1}}},...,\frac{1}{\sqrt{\mu_{n,K}}})\tilde S_n^{-1}U_n\Lambda_nU_n^{T}\tilde S_n^{-T}{\rm {\rm diag}}(\frac{1}{\sqrt{\mu_{n,1}}},...,\frac{1}{\sqrt{\mu_{n,K}}}) \\
    & = {\rm {\rm diag}}(\frac{1}{\sqrt{\mu_{n,1}}},...,\frac{1}{\sqrt{\mu_{n,K}}}) \left(\sum_{k=1}^K \lambda_{nk} \bm v_k \bm v_k^T \right) {\rm {\rm diag}}(\frac{1}{\sqrt{\mu_{n,1}}},...,\frac{1}{\sqrt{\mu_{n,K}}})
\end{align*} 
where $\bm v_k = (v_{k1},...,v_{kK})$ corresponds to the $k$-th column of the matrix $\tilde S_n^{-1} U_n$. We further simplify the above display as
\begin{align*}
    S_n^{-1} \left( \sum_j \bgamma_j \bgamma_j^T \sigma_{Yj}^{-2}\right) S_n^{-T} = \sum_{k=1}^K \lambda_{nk} \bm {\bar v}_k \bm {\bar v}_k^T
\end{align*}
where $\bm {\bar v}_k = (v_{k1}/\sqrt{\mu_{n1}},...,v_{kK}/\sqrt{\mu_{nK}})$.

Now consider two scenarios: (1) $\lambda_{n,min}/\mu_{n,min} \rightarrow \infty$ and (2) $\lambda_{n,min} = o(\mu_{n,min})$. Let $m^*$ denote the index such that $\mu_{nm^*} = \mu_{n,min}$
 
For scenario (1), the $(m^*,m^*)$-entry of $S_n^{-1} \left( \sum_j \bgamma_j \bgamma_j^T \sigma_{Yj}^{-2}\right) S_n^{-T}$ is $\sum_{k=1}^K \lambda_{nk} v_{km^*}^2/\mu_{nm^*}$. By the boundedness of $\tilde S_n$ in Assumption 2 and the boundedness of eigenvectors $U_n$, $\bm v_{k}$ is bounded for all $k$. Now, because we assume $\lambda_{n,min}/\mu_{n,mim} \rightarrow \infty$ (i.e., $\lambda_{nk}/\mu_{nm^*} \rightarrow \infty$ for all $k$), we have $\sum_{k=1}^K \lambda_{nk} v_{km^*}^2/\mu_{nm^*} \rightarrow \infty$, which contradicts Assumption 2 that $S_n^{-1} \left( \sum_j \bgamma_j \bgamma_j^T \sigma_{Yj}^{-2}\right) S_n^{-T}$ is bounded.

For scenario (2), if $\lambda_{n,min} = o(\mu_{n,min})$ (i.e., $\lambda_{nK} = o(\mu_{n,min})$), $\lambda_{nK}/\mu_{nk} \rightarrow 0$ for all $k$. Because $\bm v_{k}$ is bounded for all $k$, we have $\lambda_{nK}\bm {\bar v}_K \bm {\bar v}_K^T \rightarrow \bm 0$. Then,
\begin{align*}
    S_n^{-1} \left( \sum_j \bgamma_j \bgamma_j^T \sigma_{Yj}^{-2}\right) S_n^{-T} = {\sum_{k=1}^{K-1} \lambda_{nk} \bm {\bar v}_k \bm {\bar v}_k^T + o(1)}
\end{align*}
{which is the sum of $K-1$ rank-one $K\times K$ matrices and a remainder matrix whose entries are $o(1)$. We know the minimum eigenvalue of $\sum_{k=1}^{K-1} \lambda_{nk} \bm {\bar v}_k \bm {\bar v}_k^T$ is 0 and the maximum eigenvalue of the remainder matrix is  $o(1)$. As a result, by Weyl's inequality, the minimum eigenvalue of  $S_n^{-1} \left( \sum_j \bgamma_j \bgamma_j^T \sigma_{Yj}^{-2}\right) S_n^{-T}$ goes to 0 as $p \rightarrow \infty$, which contradicts Assumption 2 that the minimum eigenvalue of $S_n^{-1} \left( \sum_j \bgamma_j \bgamma_j^T \sigma_{Yj}^{-2}\right) S_n^{-T}$ is bounded away from 0.}
Therefore, it has to be true that $\lambda_{n,min} = \Theta(\mu_{n,min})$.

To show the minimum eigenvalue of $\sum_j \Omega_j^{-1} \bgamma_j\bgamma_j^T\Omega_j^{-T}$ is $\Theta(\mu_{n,min})$, we utilize the result from \ref{connection two matrices}, and then the proof is almost identical and hence omitted. 

\subsection{Proof of $\mu_{n,min} \rightarrow \infty$ implying that $nh^2_k \to \infty$ for every $k$, where $h_k^2$ is the proportion of variance of $X_k$ explained by the $p$ SNPs} \label{proof: heritability}

As discussed after Assumption 2 and justified in Section \ref{proof: minimum eigenvalue},  the minimum eigenvalue $\lambda_{\min}$ of $\sum_j \bgamma_j \bgamma_j^T \sigma_{Yj}^{-2}$ is $\Theta(\mu_{n,min})$. By the definition of the minimum eigenvalue of a symmetric matrix (see Lemma \ref{lemma0}), we know $\lambda_{\min} \le e_k^T \left(\sum_j \bgamma_j \bgamma_j^T \sigma_{Yj}^{-2} \right) e_k = \sum_j \gamma_{jk}^2 \sigma_{Yj}^{-2}$ for every $k$, where $e_k$ is a vector of zeros with the $k$-th entry being 1. Thus, $\sum_j \gamma_{jk}^2\sigma_{Yj}^{-2} \ge \Theta(\mu_{n,min})$ for every $k$. Then, because $\sigma_{Xjk}^2/\sigma_{Yj}^2$ is bounded away from 0 and $\infty$, $\sum_j \gamma_{jk}^2/\sigma_{Xjk}^2 = 
\Theta(\sum_j \gamma_{jk}^2\sigma_{Yj}^{-2})\rightarrow \infty$ as $\mu_{n,min} \rightarrow \infty$ holds for every $k$.

Next, using the equations in $\eqref{eq: sigma xjk}$, we have for every $k$
\begin{align*}
    \sum_j \frac{\gamma_{jk}^2}{\sigma_{Xjk}^2} = n_{Xk}\sum_j \frac{\gamma_{jk}^2Var(Z_j)}{Var(X_k) - \gamma_{jk}^2 Var(Z_j)} \approx n_{Xk} \sum_j \frac{\gamma_{jk}^2Var(Z_j)}{Var(X_k)} = n_{Xk} h_k^2 = 
    \Theta(n) h_k^2
\end{align*}
Therefore, $\mu_{n,min} \rightarrow \infty$ implies $nh^2_k \to \infty$ for every $k$. The approximation is reasonable because each 
$\gamma_{jk}^2Var(Z_j)$ is small compared with $Var(X)$.

\subsection{Connection of Assumption 2 to the concentration parameter in the individual-level data setting}

To draw the connection, we assume the following simultaneous equations model (\citeSupp{Chao:2005aa}):
\begin{align}
	{X}_n &=   Z_n \Pi_n + {V}_n\label{eq: chaoexp}\\ 
	y_n&=  {X}_n \bbeta_0 + u_n  \label{eq: chaoout} 
\end{align}
where $y_n$ is an $n\times 1$ vector of  observations on the outcome, ${X}_n$ is an $n \times K$ matrix of observations of the exposures, ${Z}_n$ is an $n\times p$ matrix of observations on the $p$ instrument variables, $\Pi_n$ is an $p \times K$ coefficient matrix, and $u_n$ and ${V}_n$ are, respectively, an $n \times 1$ vector and an $n \times K$ matrix of random noises. We assume that the vector $(u_i, v_i)^T$, where $u_i$ and $v_i^T$ are respectively the $i$th component of $u_n$ and the $i$th row of ${V}_n$, are independent and identically distributed with mean $ \bm 0 $ and the positive definite variance-covariance matrix which can be partitioned as

\[\begin{pmatrix}
\sigma_{uu} & \sigma_{uV}^T \\
\sigma_{Vu} & \Sigma_{VV} 
\end{pmatrix}.\]

We want to establish that for any $S_n$ satisfying Assumption 2, if $S_n^{-1}\sum_j \bgamma_j \bgamma_j^T \sigma_{Yj}^{-2} S_n^{-T}$ is element-wise bounded with smallest eigenvalue bounded away from 0, then $S_n^{-1}\Pi_n^T Z_n^T Z_n \Pi_n S_n^{-T}$ is also bounded with smallest eigenvalue bounded away from 0 almost surely. This can further imply that the smallest eigenvalue of the concentration parameter $\Sigma_{VV}^{-\frac{1}{2}} \Pi_n^T Z_n^T Z_n \Pi_n \Sigma_{VV}^{-\frac{1}{2}}$ grows at a rate of $\mu_{n,\min}$ (defined in Assumption 2) almost surely. Thus, we can interpret $\mu_{n,min}$ as the rate at which the smallest eigenvalue of the concentration parameter grows as $n$ increases.

Next, we prove the above claims under model \eqref{eq: chaoexp}-\eqref{eq: chaoout} and  the following conditions:
\begin{enumerate}
    \item[(1)] For $j = 1,...,p$, $\sigma_{Yj}^2$ is known and $\sigma_{Yj}^2 \approx \frac{Var(Y)}{nVar(Z_j)}$ as shown in Section \ref{sec: rate sigma}. 
    \item[(2)] $Var(Y) = \Theta(1)$
    \item[(3)] $p$ IVs are independent
    \item[(4)] There exists constants $C_1$ and $C_2$ such that $0 < C_1 \le \lowlim_{n\to \infty}({\rm min}_{j \in [p]}Var(Z_j)) \le \uplim_{n \to \infty}({\rm max}_{j \in [p]}Var(Z_j)) \le C_2 < \infty$, where $[p]$ denotes the set $\{1,...,p\}$
    \item[(5)] There exists constants $D_1$ and $D_2$ such that $0 < D_1 \le \lowlim_{n \to \infty}\lambda_{\min}(\frac{Z_n^TZ_n}{n}) \le \uplim_{n \to \infty}\lambda_{\max}(\frac{Z_n^TZ_n}{n}) \le D_2 < \infty$ almost surely.
\end{enumerate}

First, when all IVs are independent, $\Pi_n^T = (\bgamma_1, \dots, \bgamma_p)$, and thus $S_n^{-1}\sum_j \bgamma_j \bgamma_j^T \sigma_{Yj}^{-2} S_n^{-T}=S_n^{-1}\Pi_n^T W \Pi_n S_n^{-T}$, where $W = {\rm {\rm diag}}(\sigma_{Y1}^{-2},...,\sigma_{Yp}^{-2})$ denotes a {\rm diag}onal matrix with {\rm diag}onal entries being $\sigma_{Yj}^{-2}, j = 1,...,p$. According to Assumption 2 that $S_n^{-1}\Pi_n^T W \Pi_n S_n^{-T}$ has the smallest eigenvalue bounded away from 0, we know, there exist constants $D_3, D_4$ such that $0 < D_3 \le \lowlim_{n \to \infty}\lambda_{\min}(S_n^{-1}\Pi_n^T W \Pi_nS_n^{-T}) \le \uplim_{n \to \infty}\lambda_{\max}(S_n^{-1}\Pi_n^T W \Pi_nS_n^{-T}) \le D_4 < \infty$. 

Next, note that $S_n^{-1}\Pi_n^TZ_n^TZ_n\Pi_nS_n^{-T} = S_n^{-1}\Pi_n^TW^{\frac{1}{2}}W^{-\frac{1}{2}}Z_n^TZ_nW^{-\frac{1}{2}}W^{\frac{1}{2}}\Pi_nS_n^{-T}$. Consider the inner matrix $W^{-\frac{1}{2}}Z_n^TZ_nW^{-\frac{1}{2}}$ which is of size $p \times p$. The $(t, s)$ element of it is $\sum_{i=1}^nZ_{it}Z_{is}\sigma_{Yt}\sigma_{Ys}$, where $t,s \in \{1,...,p\}$. Since  $\sigma_{Yj}^2 \approx \frac{Var(Y)}{nVar(Z_j)}$ as shown in Section \ref{sec: rate sigma}, we have that  $W^{-\frac{1}{2}}Z_n^TZ_nW^{-\frac{1}{2}}$ can be approximated by $Var(Y)W_z^{-\frac{1}{2}}(\frac{1}{n}{Z}_n^T{Z}_n)W_z^{-\frac{1}{2}}$, where
\begin{align*}
	W_z = 
	\left(
	\begin{array}{ccc}
		Var({Z_1}) &&\\
		&	\ddots&\\
		&&Var(Z_{p})  \\
	\end{array}
	\right).
\end{align*}
By the condition (4) and (5), we know for any $0<\epsilon < min(C_1, C_2,D_1, D_2)$, when $n$ is large enough, we have
\begin{align*}
    \lambda_{\min}(W_z^{-\frac{1}{2}}(\frac{1}{n}Z_n^TZ_n)W_z^{-\frac{1}{2}}) &\ge \lambda_{\min}(\frac{Z_n^TZ_n}{n})\lambda_{\min}(W_z^{-1})\\
    & = \lambda_{\min}(\frac{Z_n^TZ_n}{n})/ \lambda_{\max}(W_z) \\
    \ & > (D_1 - \epsilon)/(C_2 + \epsilon) > 0 \ \text{a.s.},
\end{align*}
and 
\begin{align*}
    \lambda_{\max}(W_z^{-\frac{1}{2}}(\frac{1}{n}Z_n^TZ_n)W_z^{-\frac{1}{2}}) &\le \lambda_{\max}(\frac{Z_n^TZ_n}{n})\lambda_{\max}(W_z^{-1})\\
    & = \lambda_{\max}(\frac{Z_n^TZ_n}{n})/\lambda_{\min}(W_z) \\
    \ & < (D_2 + \epsilon)/(C_1 - \epsilon) < \infty \ \text{a.s.}.
\end{align*}
The above derivations use Lemma \ref{lemma0}. Since $\epsilon$ is arbitrary and $Var(Y) = \Theta(1)$, it follows that there exist constants $D_1^{'}$ and $D_2^{'}$ such that, 
\begin{align*}
0< D_1^{'} & \le \lowlim_{n \to \infty}\lambda_{\min}(W^{-\frac{1}{2}}Z_n^TZ_nW^{-\frac{1}{2}}) \\ & \le \uplim_{n \to \infty}\lambda_{\max}(W^{-\frac{1}{2}}Z_n^TZ_nW^{-\frac{1}{2}}) \le D_2^{'} < \infty  \ \text{a.s.} 
\end{align*}

Then, for any $\epsilon >0$, $\epsilon < min(D_1^{'}, D_3)$, when $n$ is large enough, we have, 
\begin{align*}
	\lambda_{\min}(S_n^{-1}\Pi_n^TZ_n^TZ_n\Pi_nS_n^{-T}) &= \lambda_{\min}(S_n^{-1}\Pi_n^TW^{\frac{1}{2}}W^{-\frac{1}{2}}Z_n^TZ_nW^{-\frac{1}{2}}W^{\frac{1}{2}}\Pi_nS_n^{-T}) \\
        \ &\ge \lambda_{\min}(S_n^{-1}\Pi_n^T W \Pi_nS_n^{-T})\lambda_{\min}(W^{-\frac{1}{2}}Z_n^TZ_nW^{-\frac{1}{2}}) \\
        \ &> (D_1^{'} - \epsilon)(D_3 - \epsilon) > 0  \ \text{a.s.}
\end{align*}
where the first inequality uses Lemma \ref{lemma0}. It follows that $S_n^{-1}\Pi_n^TZ_n^TZ_n\Pi_nS_n^{-T}$ has smallest eigenvalue bounded away from 0 almost surely.

We next show that $S_n^{-1}\Pi_n^TZ_n^TZ_n\Pi_nS_n^{-T}$ is bounded almost surely. We use the fact that $\lambda_{\max}(S_n^{-1}\Pi_n^T W \Pi_nS_n^{-T})$ is bounded and $\lambda_{\max}(W^{-\frac{1}{2}}Z_n^TZ_nW^{-\frac{1}{2}})$ is bounded almost surely. Note that
\begin{align*}
	\lambda_{\max}(S_n^{-1}\Pi_n^TZ_n^TZ_n\Pi_nS_n^{-T}) &= \lambda_{\max}(S_n^{-1}\Pi_n^TW^{\frac{1}{2}}W^{-\frac{1}{2}}Z_n^TZ_nW^{-\frac{1}{2}}W^{\frac{1}{2}}\Pi_nS_n^{-T}) \\
        \ &\le \lambda_{\max}(S_n\Pi_n^T W \Pi_nS_n^{-T})\lambda_{\max}(W^{-\frac{1}{2}}Z_n^TZ_nW^{-\frac{1}{2}}) \\
        \ &< \infty 
\end{align*}
holds almost surely for large enough $n$.

Let $e_t$ be a vector of length K with $t$th entry being 1 and the other entries being 0. We write the $(t,s)$ entry of $S_n^{-1}\Pi_n^TZ_n^TZ_n\Pi_nS_n^{-T}$ as $   e_t^TS_n^{-1}\Pi_n^TZ_n^TZ_n\Pi_n S_n^{-T} e_s $, and 
\begin{align*}
    |e_t^TS_n^{-1}\Pi_n^TZ_n^TZ_n\Pi_n S_n^{-T} e_s| &\le \sqrt{e_t^T(S_n^{-1}\Pi_n^TZ_n^TZ_n\Pi_nS_n^{-T}) e_t}\sqrt{e_s^T(S_n^{-1}\Pi_n^TZ_n^TZ_n\Pi_nS_n^{-T}) e_s} \\
    \ &\le \lambda_{\max}(S_n^{-1}\Pi_n^TZ_n^TZ_n\Pi_nS_n^{-T}) \\
    \ &<\infty
\end{align*}
holds almost surely for large enough $n$. The first inequality follows from the Cauchy-Schwartz inequality. 

Next, we want to establish that the concentration parameter $\Sigma_{VV}^{-\frac{1}{2}}\Pi_n^TZ_n^TZ_n\Pi_n\Sigma_{VV}^{-\frac{1}{2}}$ has the same rates of divergence as the $\sum_j \bgamma_j \bgamma_j^T \sigma_{Yj}^{-2}$, in the sense that there exists $\tilde Q_n$ such that $Q_n = \tilde Q_n  {\rm {\rm diag}}(\sqrt{\mu_{n1}},...,\sqrt{\mu_{nK}})$ and $Q_n^{-1} \Sigma_{VV}^{-\frac{1}{2}}\Pi_n^TZ_n^TZ_n\Pi_n\Sigma_{VV}^{-\frac{1}{2}} Q_n^{-T}$ is bounded with smallest eigenvalue bounded away from 0 almost surely. To show this, we let $\tilde Q_n = \Sigma_{VV}^{-\frac{1}{2}} \tilde S_n $. Note that
\begin{align*}
    Q_n^{-1} \Sigma_{VV}^{-\frac{1}{2}}\Pi_n^TZ_n^TZ_n\Pi_n\Sigma_{VV}^{-\frac{1}{2}} Q_n^{-T} = S_n^{-1}\Pi_n^TZ_n^TZ_n\Pi_nS_n^{-T}, 
\end{align*}
where the right-hand side is bounded with the smallest eigenvalue bounded away from 0 almost surely as showed above. 

Finally, following the same arguments as in \ref{proof: minimum eigenvalue} (i.e., first taking eigen-decomposition on $\Sigma_{VV}^{-\frac{1}{2}}\Pi_n^TZ_n^TZ_n\Pi_n\Sigma_{VV}^{-\frac{1}{2}}$ and then giving proof by contradiction), we can show that
the smallest eigenvalue of $\Sigma_{VV}^{-\frac{1}{2}}\Pi_n^TZ_n^TZ_n\Pi_n\Sigma_{VV}^{-\frac{1}{2}}$ and $\mu_{n,min}$ have the same rate of divergence.

\subsection{Connection of our adjustment to {\cite{van2016ridge}}}

In the estimation of high dimensional precision matrices, when the dimension of the covariates is close to or exceeds the sample size, inverting the sample covariance matrix faces a similar singularity issue as inverting $\sum_j \hat M_j - V_j$. \citeSupp{van2016ridge} introduced a precision matrix estimator $S/2 + \sqrt{S^2/4 + \phi I}$, where $S$ is a positive semi-definite sample covariance matrix and $I$ is the identity matrix. An eigen-decomposition demonstrates that their estimator essentially adjusts each eigenvalue $\lambda$ of $S$ to $\lambda/2 + \sqrt{\lambda^2/4 + \phi}$. In our case, the matrix $\sum_j \hat M_j - V_j$ takes the role of $S$. Interestingly, if we view this adjustment as a function of $\phi$, our proposed adjustment (adjusting each eigenvalue $\lambda$ to $\lambda +  {\phi}{\lambda^{-1}}$) can be viewed as a first-order Taylor expansion of their adjustment. 

\subsection{Proof of Theorem 1 and claims related to MV-IVW}

\vspace{3mm}
\textbf{Proof of Theorem 1}\label{proof: theorem 1a}

\begin{proof}
    We aim to show that $\mathbb{V}^{-\frac{1}{2}} (\sum_j M_j + V_j)(\hat \bbeta_{\rm IVW} - \tilde \bbeta_{\rm IVW}) \xrightarrow[]{D} N(\bm 0, I_K)$. Note that
    \begin{align*}
         \mathbb{V}^{-\frac{1}{2}} (\sum_j M_j + V_j)(\hat \bbeta_{\rm IVW} - \tilde \bbeta_{\rm IVW}) & = \mathbb{V}^{-\frac{1}{2}} (\sum_j M_j + V_j) (\sum_j \hat M_j)^{-1}(\sum_j \hat \bgamma_j \hat \Gamma_j \sigma_{Yj}^{-2} - \hat M_j\bbeta_0 + V_j\bbeta_0) \\
        & = \mathbb{V}^{-\frac{1}{2}} (\sum_j M_j + V_j) (\sum_j \hat M_j)^{-1}\mathbb{V}^{\frac{1}{2}}\mathbb{V}^{-\frac{1}{2}}(\sum_j \hat \bgamma_j \hat \Gamma_j \sigma_{Yj}^{-2} - \hat M_j\bbeta_0 + V_j\bbeta_0)
    \end{align*}
    
    We have shown in Lemma \ref{lemma5} that $\mathbb{V}^{-\frac{1}{2}}(\sum_j \hat \bgamma_j \hat \Gamma_j \sigma_{Yj}^{-2} - \hat M_j\bbeta_0 + V_j\bbeta_0)\xrightarrow[]{D} N(\bm 0,I_K)$. We now show that 
    \begin{align}\label{converge in prob result M hat}
        \mathbb{V}^{-\frac{1}{2}} (\sum_j M_j + V_j) (\sum_j \hat M_j)^{-1}\mathbb{V}^{\frac{1}{2}} \xrightarrow[]{P} I_K.
    \end{align}
     Note that
    \begin{align*}
       \mathbb{V}^{-\frac{1}{2}} (\sum_j \hat M_j) (\sum_j M_j + V_j)^{-1}\mathbb{V}^{\frac{1}{2}} & =  \mathbb{V}^{-\frac{1}{2}} (\sum_j \hat M_j - (M_j + V_j))(\sum_j M_j + V_j)^{-1} \mathbb{V}^{\frac{1}{2}} + I_K
    \end{align*}
    It suffices to show the first term on the right of above equation converges in probability to $\bm 0$. Expand it as follows:
    \begin{align*}
       \mathbb{V}^{-\frac{1}{2}}T_n \bigg\{T_n^{-1} (\sum_j \hat M_j - (M_j + V_j))T_n^{-T}\bigg\}\bigg\{T_n^{-1}(\sum_j M_j + V_j)T_n^{-T}\bigg\}^{-1} T_n^{-1}\mathbb{V}^{\frac{1}{2}}
    \end{align*}
    We know from Lemma \ref{lemma1} that $(T_n^{-1}(\sum_j M_j + V_j)T_n^{-T})^{-1}$ is bounded, from Lemma \ref{lemma2} that $T_n^{-1} (\sum_j \hat M_j - (M_j + V_j))T_n^{-T}$ converges in probability to $\bm 0$. Because $T_n^{-1}\mathbb{V}T_n^{-T}$ is bounded with eigenvalues bounded away from 0 by Lemma \ref{lemma1}, $\mathbb{V}^{-\frac{1}{2}}T_n$ and $T_n^{-1}\mathbb{V}^{\frac{1}{2}}$ are also bounded. Hence, $\mathbb{V}^{-\frac{1}{2}} (\sum_j \hat M_j) (\sum_j M_j + V_j)^{-1}\mathbb{V}^{\frac{1}{2}}$ converges in probability to $I_K$. By the continuous mapping theorem, we have $\mathbb{V}^{-\frac{1}{2}} (\sum_j M_j + V_j) (\sum_j \hat M_j)^{-1}\mathbb{V}^{\frac{1}{2}}$ also converges in probability to $I_K$. Lastly, by Slutsky's theorem, we have the desired result.
\end{proof}

\textbf{Proof of Proposition 1}

In the following, we prove claims directly related to $\hat \bbeta_{\rm IVW}$. The claims in Proposition 1 then follow by noting that $\hat \bbeta_{\rm IVW} - \tilde \bbeta_{\rm IVW} \xrightarrow[]{P} \bm 0$ by Theorem 1.

\vspace{3mm}
\textbf{Asymptotic normality and consistency of MV-IVW when ${\mu_{n,min}}/{p^2} \rightarrow \infty$}

We aim to show that $\hat \bbeta_{IVW}$ is consistent and asymptotically normal under the conditions of Theorem 1 and the condition ${\mu_{n,min}}/{p^2} \rightarrow \infty$. In particular, it holds that
\begin{align*}
    \mathbb{V}^{-\frac{1}{2}} (\sum_j M_j + V_j)(\hat \bbeta_{\rm IVW} - \bbeta_0) \xrightarrow[]{D} N(\bm 0, I_K)
\end{align*}
\begin{proof}
    Note that \begin{align*}
         \mathbb{V}^{-\frac{1}{2}} (\sum_j M_j + V_j)(\hat \bbeta_{\rm IVW} - \bbeta_0) & = \mathbb{V}^{-\frac{1}{2}} (\sum_j M_j + V_j) (\sum_j \hat M_j)^{-1}(\sum_j \hat \bgamma_j \hat \Gamma_j \sigma_{Yj}^{-2} - \hat M_j\bbeta_0 + V_j\bbeta_0) \\
        & - \mathbb{V}^{-\frac{1}{2}}(\sum_j M_j + V_j)(\sum_j \hat M_j)^{-1} (\sum_j V_j) \bbeta_0 \\
        & = \underbrace{\mathbb{V}^{-\frac{1}{2}} (\sum_j M_j + V_j) (\sum_j \hat M_j)^{-1}\mathbb{V}^{\frac{1}{2}}\mathbb{V}^{-\frac{1}{2}}(\sum_j \hat \bgamma_j \hat \Gamma_j \sigma_{Yj}^{-2} - \hat M_j\bbeta_0 + V_j\bbeta_0)}_{A_1} \\
        & - \underbrace{\mathbb{V}^{-\frac{1}{2}}(\sum_j M_j + V_j)(\sum_j \hat M_j)^{-1} (\sum_j V_j) \bbeta_0}_{A_2} 
    \end{align*}
    We have shown that $A_1$ converges in distribution to $N(\bm 0, I_K)$. It remains to show $A_2$ converges in probability to $\bm 0$. Then, the asymptotic normality of the MV-IVW estimator follows by the Slutsky's theorem.

    For $A_2$, note that
    \begin{align*}
       \mathbb{V}^{-\frac{1}{2}}(\sum_j M_j + V_j)(\sum_j \hat M_j)^{-1} (\sum_j V_j) =  \mathbb{V}^{-\frac{1}{2}} (\sum_j M_j + V_j) (\sum_j \hat M_j)^{-1} \mathbb{V}^{\frac{1}{2}} \mathbb{V}^{-\frac{1}{2}} (\sum_j V_j)
    \end{align*}
    We have just shown that $\mathbb{V}^{-\frac{1}{2}} (\sum_j M_j + V_j) (\sum_j \hat M_j)^{-1} \mathbb{V}^{\frac{1}{2}}$ converges in probability to $I_K$. By Lemma \ref{lemma 3}, we know $\mathbb{V}^{-\frac{1}{2}} (\sum_j V_j)$ converges to $\bm 0$. Thus, the entire second term converges in probability to $\bm 0$.
    
    Applying the Slutsky's theorem, we get the desired asymptotic normality result. Consistency follows directly from Lemma \ref{lemma 4} i.e., $\mathcal{V}_{\rm IVW} =  (\sum_j M_j + V_j)^{-1}\mathbb{V}(\sum_j M_j + V_j)^{-1} \rightarrow \bm0$.
\end{proof}

\textbf{Consistency of variance estimators}

    We next show that under the conditions of Theorem 1 and an additional condition that $\mu_{n,min}/p \rightarrow \infty$, $\mathbb{V}^{-\frac{1}{2}} (\sum_j M_j + V_j)(\hat \bbeta_{\rm IVW} - \tilde \bbeta_{\rm IVW}) \xrightarrow[]{D} N(\bm 0, I_K)$ still holds if we replace $\sum_j M_j + V_j$ and $\mathbb{V}$ respectively by $\sum_j \hat M_j$ and
    \begin{align*}
        \hat{\mathbb{V}}  = \sum_{j}\{(1+\hat{\bm\beta}_{\rm IVW}^TV_j\hat{\bm\beta}_{\rm IVW})\hat{M_j}+V_j\hat{\bm\beta}_{\rm IVW}\hat{\bm\beta}_{\rm IVW}^TV_j\}
    \end{align*}
    \begin{proof}
    Note that
  \begin{align*}
      & \hat{\mathbb{V}}^{-1/2} (\sum_{j}\hat{M}_j) \left(\bm{\hat{\bbeta}}_{\rm IVW} - \tilde\bbeta_{\rm IVW}\right)\\
      & = \hat{\mathbb{V}}^{-1/2} (\sum_{j}\hat{M}_j) (\sum_j M_j + V_j)^{-1} \mathbb{V}^{\frac{1}{2}} \mathbb{V}^{-\frac{1}{2}}
        ( \sum_{j}M_j+V_j)\left(\bm{\hat{\bbeta}}_{\rm IVW} - \tilde\bbeta_{\rm IVW}\right) \\
     & = \underbrace{\hat{\mathbb{V}}^{-1/2} \mathbb{V}^{\frac{1}{2}}}_{B_1} \underbrace{ \mathbb{V}^{-\frac{1}{2}}(\sum_{j}\hat{M}_j) (\sum_j M_j+V_j)^{-1} \mathbb{V}^{\frac{1}{2}}}_{B_2} \underbrace{\mathbb{V}^{-\frac{1}{2}}
        ( \sum_{j}M_j+V_j)\left(\bm{\hat{\bbeta}}_{\rm IVW} - \tilde\bbeta_{\rm IVW}\right)}_{B_3}
  \end{align*}
  We have just shown in the proof of Theorem 1 that $B_2 \xrightarrow[]{P} I_K$ and $B_3 \xrightarrow[]{D} N(\bm 0, I_K)$. It remains to show $B_1 \xrightarrow[]{P} I_K$. 

    Before that, we first show that $\mathbb{V}^{-\frac{1}{2}} \hat{\mathbb{V}} \mathbb{V}^{-\frac{1}{2}}$ converges in probability to $I_K$.  
    Let $\hat{\mathbb{V}}^{'} = \sum_{j}\{(1+\bbeta_0^TV_j\bbeta_0)\hat{M_j}+V_j\bbeta_0 \bbeta_0^TV_j\}$. It is sufficient to show that $\mathbb{V}^{-\frac{1}{2}} \hat{\mathbb{V}}^{'} \mathbb{V}^{-\frac{1}{2}}$ converges in probability to $I_K$, because the desired result follows by replacing $\bbeta_0$ with its consistent estimator $\hat \bbeta_{\rm IVW}$, the consistency of which will be established in the next subsection.
    We can expand $\mathbb{V}^{-\frac{1}{2}} \hat{\mathbb{V}}^{'} \mathbb{V}^{-\frac{1}{2}}$ as follows:
    \begin{align*}
        \mathbb{V}^{-\frac{1}{2}} \hat{\mathbb{V}}^{'} \mathbb{V}^{-\frac{1}{2}} = 
        \mathbb{V}^{-\frac{1}{2}}T_n \bigg\{T_n^{-1} \hat{\mathbb{V}}^{'} T_n^{-T}\bigg\} T_n^T  \mathbb{V}^{-\frac{1}{2}}
    \end{align*}
    Note that
    \begin{align*}
        T_n^{-1} \hat{\mathbb{V}}^{'} T_n^{-T} = T_n^{-1} \mathbb{V} T_n^{-T} + T_n^{-1}  \{\sum_j (1 + \bbeta_0^T V_j \bbeta_0)(\hat M_j - (M_j+V_j))\} T_n^{-T}
    \end{align*}
    Then,
    \begin{align*}
        \mathbb{V}^{-\frac{1}{2}} \hat{\mathbb{V}}^{'} \mathbb{V}^{-\frac{1}{2}} = 
        I_K + \mathbb{V}^{-\frac{1}{2}}T_n \bigg\{T_n^{-1}  \{\sum_j (1 + \bbeta_0^T V_j \bbeta_0)(\hat M_j - (M_j+V_j))\} T_n^{-T}\bigg\} T_n^T  \mathbb{V}^{-\frac{1}{2}}
    \end{align*}
    Because $(1 + \bbeta_0^T V_j \bbeta_0) = \Theta(1)$ by Assumption 1 (i.e., $V_j = O(1)$), there exists a positive constant $c$ such that $1 \le 1 + \bbeta_0^T V_j \bbeta_0 \le c$ for every $j$. Then, consider any $(t,s)$ entry of $T_n^{-1}  \{\sum_j(\hat M_j - (M_j+V_j))\} T_n^{-T}$ and $T_n^{-1}  \{\sum_j (1 + \bbeta_0^T V_j \bbeta_0)(\hat M_j - (M_j+V_j))\} T_n^{-T}$, denoted respectively as $D_{1[t,s]}$ and $D_{2[t,s]}$. We know
    {$|D_{1[t,s]}| \le |D_{2[t,s]}| \le c|D_{1[t,s]}|$.}
     By Lemma \ref{lemma2}, $D_{1[t,s]} \xrightarrow[]{P} 0$ for every $t, s = 1,...,K$. Thus, $T_n^{-1}  \{\sum_j (1 + \bbeta_0^T V_j \bbeta_0)(\hat M_j - (M_j+V_j))\} T_n^{-T}$ converges in probability to $\bm 0$ as well. Finally, because $\mathbb{V}^{-\frac{1}{2}}T_n$ is bounded, $\mathbb{V}^{-\frac{1}{2}} \hat{\mathbb{V}}^{'} \mathbb{V}^{-\frac{1}{2}}$ converges in probability to $I_K$. Hence, $\mathbb{V}^{-\frac{1}{2}} \hat{\mathbb{V}} \mathbb{V}^{-\frac{1}{2}} \xrightarrow[]{P} I_K$.
  
  Then, applying the continuous mapping theorem, we have $\mathbb{V}^{\frac{1}{2}}\hat{\mathbb{V}}^{-1}\mathbb{V}^{\frac{1}{2}} \xrightarrow[]{P} I_K$. That is, $B_1 B_1^T \xrightarrow[]{P} I_K$. In other words, $B_1$ is an orthogonal matrix with probability tending 1. As a result, $B_1^{-1} = B_1^T$ with probability tending 1. This implies
  \begin{align*}
      B_1^{-1} =  \mathbb{V}^{-\frac{1}{2}} \hat{\mathbb{V}}^{1/2} = \mathbb{V}^{\frac{1}{2}} \hat{\mathbb{V}}^{-1/2} = B_1^T
  \end{align*}
  with probability tending 1. From this equality, we conclude that $\mathbb{V} = \hat{\mathbb{V}}$ with probability tending 1. Therefore, $B_1 \xrightarrow[]{P} I_K$. This concludes the proof.
\end{proof}

\textbf{Consistency of MV-IVW when $\mu_{n,min}/p \rightarrow \infty$}

    We aim to show $\hat{\bm{\beta}}_{\rm IVW} - \bm{\beta}_0 \xrightarrow[]{P} 0$ when $\mu_{n,min}/p \to \infty$. 
\begin{proof}
    Note that
    \begin{align*}
        \hat{\bm{\beta}}_{\rm IVW} - \bm{\beta}_0 & = (\sum_j M_j + V_j)^{-1}\mathbb{V}^{\frac{1}{2}}\mathbb{V}^{-\frac{1}{2}} (\sum_j M_j + V_j)(\hat{\bm{\beta}}_{\rm IVW} - \bm{\beta}_0) \\
        & = \underbrace{(\sum_j M_j + V_j)^{-1}\mathbb{V}^{\frac{1}{2}} \mathbb{V}^{-\frac{1}{2}} (\sum_j M_j + V_j) (\sum_j \hat M_j)^{-1}\mathbb{V}^{\frac{1}{2}}\mathbb{V}^{-\frac{1}{2}}(\sum_j \hat \bgamma_j \hat \Gamma_j \sigma_{Yj}^{-2} - \hat M_j\bbeta_0 + V_j\bbeta_0)}_{A_1}\\
        & - \underbrace{(\sum_j \hat M_j)^{-1} (\sum_j V_j) \bbeta_0}_{A_2}
    \end{align*}
    We want to show that both $A_1$ and $A_2$ converge in probability to $\bm 0$. For $A_1$, from \eqref{converge in prob result M hat} in the proof of Theorem 1, $\mathbb{V}^{-\frac{1}{2}} (\sum_j M_j + V_j) (\sum_j \hat M_j)^{-1}\mathbb{V}^{\frac{1}{2}}$ converges in probability to $I_K$. Let $\hat{\bm \theta} = \mathbb{V}^{-\frac{1}{2}}(\sum_j \hat \bgamma_j \hat \Gamma_j \sigma_{Yj}^{-2} - \hat M_j\bbeta_0 + V_j\bbeta_0)$. We know $E(\hat{\bm \theta}) = \bm 0$ and $Var(\hat{\bm \theta}) = I_K$. Then, for any $\epsilon > 0$, there exists a large  enough constant $M$ such that
    \begin{align*}
        P(||\hat {\bm \theta}||_2^2 > M) \le \frac{E(\hat {\bm \theta}^T \hat {\bm \theta})}{M} = \frac{K}{M} < \epsilon
    \end{align*}
    where the first inequality uses the Markov inequality, so $\hat{\bm \theta}$ is bounded in probability, i.e., $O_p(1)$. Because $\mathcal{V}_{\rm IVW} \rightarrow \bm 0$ by Lemma \ref{lemma 4}, we have $(\sum_j M_j + V_j)^{-1}\mathbb{V}^{\frac{1}{2}} \rightarrow \bm 0$ and hence $A_1$ converges in probability to $\bm 0$.
    
    To show $A_2$ also converges to $\bm 0$. first note that $(\sum_j V_j) \bbeta_0/p = O(1)$ by the boundedness of variance ratios between $\sigma_{Xjk}^2$ and $\sigma_{Yj}^2$ from Assumption 1. Hence, it suffices to show that $p (\sum_j \hat M_j)^{-1}\xrightarrow[]{P} \bm 0$. To this end, we write $p (\sum_j \hat M_j)^{-1}$ as 
    \begin{align*}
        p (\sum_j \hat M_j)^{-1} & = p T_n^{-T} \bigg\{ T_n^{-1} (\sum_j \hat M_j) T_n^{-T}\bigg\}^{-1} T_n^{-1} \\
        & = (\sqrt{p} T_n^{-T}) \bigg\{ T_n^{-1} (\sum_j \hat M_j - (M_j + V_j)) T_n^{-T} + T_n^{-1}(\sum_j M_j + V_j)T_n^{-T}\bigg\}^{-1} (\sqrt{p}T_n^{-1})
    \end{align*}
    By Lemma \ref{lemma1}, we know $T_n^{-1}(\sum_j M_j + V_j)T_n^{-T}$ is bounded with minimum eigenvalue bounded away from 0. By Lemma \ref{lemma2}, we know
    $T_n^{-1} (\sum_j \hat M_j - (M_j + V_j)) T_n^{-T}$ converges in probability to $\bm 0$. Hence, $T_n^{-1} (\sum_j \hat M_j) T_n^{-T}$ is bounded in probability with the smallest eigenvalue bounded away from 0 with probability tending 1. Lastly, both terms $\sqrt{p} T_n^{-T}$ and $\sqrt{p}T_n^{-1}$ go to $\bm 0$ when $\mu_{n,min}/p \to \infty$, because,
    \begin{align}\label{proof theorem1: pT}
        \sqrt{p} T_n^{-T} = \tilde S_n^{-T} {\rm{{\rm diag}}}(\sqrt{\frac{p}{\mu_{n1} + p}},...,\sqrt{\frac{p}{\mu_{nK} + p}})
    \end{align}
Therefore, $p (\sum_j \hat M_j)^{-1}\xrightarrow[]{P} \bm 0$, and hence $A_2$ also converges in probability to $\bm 0$. This establishes the consistency.
\end{proof}

\textbf{Asymptotic bias of $\hat \bbeta_{\rm IVW}$ when ${\mu_{n,\min}}/{p} \rightarrow \tau \in[0,\infty)$}

We now show that when ${\mu_{n,\min}}/{p} \rightarrow \tau \in[0,\infty)$, $\hat{\bm{\beta}}_{\rm IVW} - (\sum_j M_j + V_j)^{-1}(\sum_j M_j)\bbeta_0 \xrightarrow[]{P} \bm 0$. 
\begin{proof}
Note that
\begin{align*}
    & \hat{\bm{\beta}}_{\rm IVW} - (\sum_j M_j + V_j)^{-1}(\sum_j M_j)\bbeta_0 \\  & =  \hat{\bm{\beta}}_{\rm IVW} - \bm{\beta}_0 + (\sum_j \hat M_j)^{-1}(\sum_j V_j)\bm{\beta}_0  - (\sum_j \hat M_j)^{-1} (\sum_j V_j) \bbeta_0  + (\sum_j M_j + V_j)^{-1} (\sum_j V_j) \bbeta_0 \\
    & = \underbrace{\hat{\bm{\beta}}_{\rm IVW} - \bm{\beta}_0 + (\sum_j \hat M_j)^{-1}(\sum_j V_j)\bm{\beta}_0}_{A_1}  + \underbrace{\bigg\{(\sum_j M_j + V_j)^{-1} - (\sum_j \hat M_j)^{-1}\bigg\}(\sum_j V_j) \bbeta_0}_{A_2}
\end{align*}
From the above proof of consistency of $\hat \bbeta_{\rm IVW}$, we know $A_1$ converges in probability to $\bm 0$. This is true even when ${\mu_{n,\min}}/{p} \rightarrow \tau \in[0,\infty)$. It then suffices to show $A_2$ converges in probability to $\bm 0$. Since $(\sum_j V_j)\bbeta_0/p = O(1)$, we want to establish
\begin{align*}
    p\bigg\{(\sum_j M_j + V_j)^{-1} - (\sum_j \hat M_j)^{-1}\bigg\} \xrightarrow[]{P} \bm 0
\end{align*}
Expand the left hand side as follows:
\begin{align*}
    (\sqrt{p} T_n^{-T})\bigg\{(T_n^{-1}(\sum_j M_j + V_j)T_n^{-T})^{-1} - (T_n^{-1}\sum_j \hat M_j T_n^{-T})^{-1} \bigg\}(\sqrt{p}T_n^{-1})
\end{align*}
Note that $T_n^{-1}\sum_j \hat M_j T_n^{-T} = T_n^{-1}\sum_j (\hat M_j - (M_j + V_j))T_n^{-T} + T_n^{-1}(\sum_j M_j + V_j) T_n^{-T}$. We known from Lemma \ref{lemma1} that $T_n^{-1}(\sum_j M_j + V_j) T_n^{-T}$ is bounded with the smallest eigenvalue bounded away from 0 and from Lemma \ref{lemma2} that $T_n^{-1}\sum_j (\hat M_j - (M_j + V_j))T_n^{-T}$ converges in probability to $\bm 0$. Thus, $(T_n^{-1}(\sum_j M_j + V_j)T_n^{-T})^{-1} - (T_n^{-1}\sum_j \hat M_j T_n^{-T})^{-1}$ converges in probability to $\bm 0$. We further know that $\sqrt{p} T_n^{-T} = O(1)$ in light of \eqref{proof theorem1: pT}. Hence, $A_2$ converges in probability to $\bm 0$. Consequently, $\hat{\bm{\beta}}_{\rm IVW} - (\sum_j M_j + V_j)^{-1}(\sum_j M_j)\bbeta_0 \xrightarrow[]{P} \bm 0$.
\end{proof}

\textbf{Direction of bias for $\hat \bbeta_{\rm IVW}$}

We construct a numerical example to show that the asymptotic bias for $\hat \bbeta_{\rm IVW}$ can be either towards or away from the null.

Suppose there are only two SNPs and two exposures which are measured in the same sample and have a correlation of $0.8$. This implies that the shared correlation matrix $P = \begin{bmatrix}
    1 & 0.8 \\
    0.8 & 1
\end{bmatrix}$. Let $\gamma_{11} = \gamma_{12} = 0.071$, and $\gamma_{21} = 0.071$ and $\gamma_{22} = 0.142$. $\sigma_{Xj1}^2 = \sigma_{Xj2}^2 = \sigma_{Yj}^2 = 1/N$ with $N = 5000$ for $j = 1, 2$. In this setting, two exposures have heritability of approximately 1\% and 1.5\%. Then, we have
\begin{align*}
     M_1 = \begin{bmatrix}
        25 & 25 \\
        25 & 25 
    \end{bmatrix}, \quad
    M_2 = \begin{bmatrix}
        25 & 50 \\
        50 & 100 
    \end{bmatrix}, \quad V_1 = V_2 = \begin{bmatrix}
        1   & 0.8 \\
        0.8 & 1
    \end{bmatrix}
\end{align*}
Then, some algebra shows that
\begin{align*}
    (\sum_j M_j + V_j)^{-1}(\sum_j M_j)\bbeta_0 = \begin{bmatrix}
        0.754 \\
        1.120
    \end{bmatrix}, \quad \text{when $\bbeta_0 = (1,1)$}
\end{align*}
\begin{align*}
    (\sum_j M_j + V_j)^{-1}(\sum_j M_j)\bbeta_0 = \begin{bmatrix}
        0.822 \\
        0.095
    \end{bmatrix}, \quad \text{when $\bbeta_0 = (1,0)$}
\end{align*}
Clearly, the second component of $\bbeta_0$ is magnified, while the first component of $\bbeta_0$ is attenuated.

\vspace{3mm}
\textbf{Proof of Corollary 1}
\begin{proof}
    We only consider the case when $\bbeta_0 \neq \bm 0$, because when $\bbeta_0 = \bm 0$, $\hat \bbeta_{\rm IVW}$ is consistent and asymptotically normal and thus the confidence intervals based on the MV-IVW estimator are valid. 

    When $\bbeta_0 \neq \bm 0$, by Theorem 1, we know, for any $\epsilon > 0$, there exists a large enough constant $C_{\epsilon}$, such that as $n \rightarrow \infty$,
    \begin{align*}
        \Pr(\forall k, \frac{|\hat \beta_{{\rm IVW},k} - \tilde \beta_{{\rm IVW},k}|}{\sqrt{\mathcal{V}_{{\rm IVW},kk}}} > C_{\epsilon}) \le \epsilon
    \end{align*}
    
    Then, we have
    \begin{align*}
        & \Pr(\exists k \text{ s.t. }\frac{|\hat \beta_{{\rm IVW},k} - \beta_{0k}|}{\sqrt{\mathcal{V}_{{\rm IVW},kk}}} > C) \\
        & = 1 - \Pr(\forall k, \ \frac{|\hat \beta_{{\rm IVW},k} - \beta_{0k}|}{\sqrt{\mathcal{V}_{{\rm IVW},kk}}} \le C) \\
        & \ge 1 - \Pr(\forall k, \frac{|\hat \beta_{{\rm IVW},k} - \tilde \beta_{{\rm IVW},k}|}{\sqrt{\mathcal{V}_{{\rm IVW},kk}}} \ge \frac{|\tilde \beta_{{\rm IVW},k} - \beta_{0k}|}{\sqrt{\mathcal{V}_{{\rm IVW},kk}}} - C) \\
        & \ge 1 - \Pr(\forall k, \frac{|\hat \beta_{{\rm IVW},k} - \tilde \beta_{{\rm IVW},k}|}{\sqrt{\mathcal{V}_{{\rm IVW},kk}}} \ge \frac{|\tilde \beta_{{\rm IVW},k} - \beta_{0k}|}{\sqrt{\mathcal{V}_{{\rm IVW},kk}}} - C \text{ and } \frac{|\hat \beta_{{\rm IVW},k} - \tilde \beta_{{\rm IVW},k}|}{\sqrt{\mathcal{V}_{{\rm IVW},kk}}} \le C_{\epsilon}) \\
        & - \Pr(\forall k, \frac{|\hat \beta_{{\rm IVW},k} - \tilde \beta_{{\rm IVW},k}|}{\sqrt{\mathcal{V}_{{\rm IVW},kk}}} \ge \frac{|\tilde \beta_{{\rm IVW},k} - \beta_{0k}|}{\sqrt{\mathcal{V}_{{\rm IVW},kk}}} - C \text{ and } \frac{|\hat \beta_{{\rm IVW},k} - \tilde \beta_{{\rm IVW},k}|}{\sqrt{\mathcal{V}_{{\rm IVW},kk}}} > C_{\epsilon}) \\
        & \ge 1 - \epsilon - \Pr(\forall k,  \frac{|\tilde \beta_{{\rm IVW},k} - \beta_{0k}|}{\sqrt{\mathcal{V}_{{\rm IVW},kk}}} \le C_{\epsilon} + C) \\
        & \ge 1 - \epsilon - \Pr(\underbrace{\forall k, \frac{|\tilde \beta_{{\rm IVW},k} - \beta_{0k}|}{\sqrt{\lambda_{\rm max, IVW}}} \le C_{\epsilon} + C}_{E_1})
    \end{align*}
    where $\lambda_{\rm max, IVW}$ denotes the maximum eigenvalue of $\mathcal{V}_{\rm IVW}$.
    The first inequality is because
    \begin{align*}
        \frac{|\tilde \beta_{{\rm IVW},k} - \beta_{0k}|}{\sqrt{\mathcal{V}_{{\rm IVW},kk}}} \le \frac{|\hat \beta_{{\rm IVW},k} - \beta_{0k}|}{\sqrt{\mathcal{V}_{{\rm IVW},kk}}}  + \frac{|\hat \beta_{{\rm IVW},k} - \tilde \beta_{{\rm IVW},k}|}{\sqrt{\mathcal{V}_{{\rm IVW},kk}}}
    \end{align*}
    by triangle inequality, and therefore the event $\frac{|\hat \beta_{{\rm IVW},k} - \beta_{0k}|}{\sqrt{\mathcal{V}_{{\rm IVW},kk}}} \le C$ implies the event $\frac{|\tilde \beta_{{\rm IVW},k} - \beta_{0k}|}{\sqrt{\mathcal{V}_{{\rm IVW},kk}}} \le C + \frac{|\hat \beta_{{\rm IVW},k} - \tilde \beta_{{\rm IVW},k}|}{\sqrt{\mathcal{V}_{{\rm IVW},kk}}}$.

For the rest of the proof, we show that if $\mu_{n,min}$ diverges to infinity at a rate slower than $p^2$, the probability that $E_1$ happens will tend to 0. Because $\epsilon > 0$ is arbitrary, it implies that $\Pr(\exists k \text{ s.t. }\frac{|\hat \beta_{{\rm IVW},k} - \beta_{0k}|}{\sqrt{\mathcal{V}_{{\rm IVW},kk}}} > C)$ approaches 1, concluding the proof.

It remains to derive the probability of the event $E_1$. Note that 
a necessary condition for the event $E_1$ to happen is that $(\tilde \bbeta_{\rm IVW} - \bbeta_0)/\sqrt{\lambda_{\rm max, IVW}} = O_p(1)$. From Lemma \ref{lemma 4} and the same argument in the section S3.5, we know that the minimum eigenvalue of $\mathcal{V}_{\rm IVW}^{-1}$ is $\Theta(\mu_{n,min} + p)$, which implies that $\lambda_{\rm max, IVW} = \Theta(1/(\mu_{n,min}+p))$. Thus, the necessary condition for the event $E_1$ to happen becomes $\sqrt{\mu_{n,min} + p}(\tilde \bbeta_{\rm IVW} - \tilde \bbeta_0) = O_p(1)$.
    
   Note that
    \begin{align*}
        & \sqrt{\mu_{n,min} + p}(\tilde \bbeta_{\rm IVW} - \bbeta_0) = (\mu_{n,min} + p)(\sum_j \hat M_j)^{-1} (\sum_j V_j)  \bbeta_0 \\
        & = \sqrt{\mu_{n,min} + p} (\sum_j M_j + V_j)^{-1} \mathbb{V}^{\frac{1}{2}} \mathbb{V}^{-\frac{1}{2}} (\sum_j M_j + V_j) (\sum_j \hat M_j)^{-1} \mathbb{V}^{\frac{1}{2}} \mathbb{V}^{-\frac{1}{2}} (\sum_j V_j)  \bbeta_0 \\
        & = \sqrt{\mu_{n,min} + p} T_n^{-T} \underbrace{T_n^T (\sum_j M_j + V_j)^{-1} \mathbb{V}^{\frac{1}{2}} \mathbb{V}^{-\frac{1}{2}} (\sum_j M_j + V_j) (\sum_j \hat M_j)^{-1} \mathbb{V}^{\frac{1}{2}} \mathbb{V}^{-\frac{1}{2}} T_n}_{\hat A} T_n^{-1}p (\sum_j V_j/p)  \bbeta_0
    \end{align*}
    By the result S8, we know $\mathbb{V}^{-\frac{1}{2}} (\sum_j M_j + V_j) (\sum_j \hat M_j)^{-1} \mathbb{V}^{\frac{1}{2}} \xrightarrow[]{P} I_K$. By Lemma \ref{lemma 4}, we known $T_n^{T} (\sum_j M_j + V_j)^{-1} \mathbb{V}^{\frac{1}{2}}$ is bounded. We further know that $\sum_j V_j/p = O(1)$ and $\mathbb{V}^{-\frac{1}{2}} T_n$ is bounded by Lemma 2. Thus, $\hat A$ is bounded in probability. Note that, 
    \begin{align*}
    \sqrt{\mu_{n,min}+p} T_n^{-T} & =  \tilde S_n^{-T} {\rm diag}(\frac{\sqrt{\mu_{n,min}+p}}{\sqrt{\mu_{n1}+p}},\dots, \frac{\sqrt{\mu_{n,min}+p}}{\sqrt{\mu_{nK}+p}}),
    \end{align*}
    where $1$ corresponds to the index $m$ such that $\mu_{nm} = \mu_{n,min}$, and
    \begin{align*}
        T_n^{-1} p=  {\rm diag}(\frac{p}{\sqrt{\mu_{n1} +p}},\dots,\frac{p}{\sqrt{\mu_{nK} +p}}) \tilde S_n^{-1}
    \end{align*}
    
    Therefore, we have
    \begin{align*}
        & \sqrt{\mu_{n,min} + p}(\tilde \bbeta_{\rm IVW} - \bbeta_0) \\
        & = \tilde S_n^{-T} \underbrace{{\rm diag}(\frac{\sqrt{\mu_{n,min}+p}}{\sqrt{\mu_{n1}+p}},\dots, \frac{\sqrt{\mu_{n,min}+p}}{\sqrt{\mu_{nK}+p}}) \hat A {\rm diag}(\frac{p}{\sqrt{\mu_{n1} +p}},\dots,\frac{p}{\sqrt{\mu_{nK} +p}})}_{\hat B} \tilde S_n^{-1} (\sum_j V_j/p) \bbeta_0.
    \end{align*}
    Let $\hat B_{[t,s]}$ denote the $(t,s)$ entry of the matrix $B$. We have
    \begin{align*}
        \hat B_{[t,s]} = O_p(\frac{\sqrt{\mu_{n,min}+p}}{\sqrt{\mu_{nt}+p}} \frac{p}{\sqrt{\mu_{ns} +p}}).
    \end{align*}

    In order for $\sqrt{\mu_{n,min} + p}(\tilde \bbeta_{\rm IVW} - \bbeta_0)$ to be bounded in probability, a necessary condition is that $\hat B_{[t,s]} = O_p(1)$ for every $t, s$. Because for every $t$, $\mu_{nt}$ diverges to infinity as least as fast as $\mu_{n,min}$, we only need to make sure $\hat B_{[m,m]}$ is $O_p(1)$. That is, we require $p/\sqrt{\mu_{nm}+p} = O(1)$, or equivalently, $p/\sqrt{\mu_{n,min}+p} = O(1)$, which implies that $\mu_{n,min}$ must diverge to infinity at a rate of at least $O(p^2)$. In other words, if $\mu_{n,min}$ diverges to infinity at a rate slower than $p^2$, the probability that $E_1$ happens will tend to 0. 
\end{proof}

\subsection{Proof of Theorem 2}\label{proof: theorem 2}

\vspace{3mm}
\textbf{Main proof of Theorem 2}
\begin{proof}
Note that
\begin{align*}
      & \mathbb{V}^{-\frac{1}{2}}
        ( \sum_{j}M_j)\left(\bm{\hat{\bbeta}}_{{\rm SRIVW},\phi} - \bm\bbeta_0\right) \\
 & = \underbrace{\mathbb{V}^{-\frac{1}{2}}
        ( \sum_{j}M_j)
{R_{\phi}^{-1}\left(\sum_{j}\hat{M}_j - V_j\right)}\mathbb{V}^{\frac{1}{2}}}_{A_1} \underbrace{\mathbb{V}^{-\frac{1}{2}}\left( \sum_j \hat \bgamma_j \hat \Gamma_j \sigma_{Yj}^{-2} - \hat M_j\bbeta_0 + V_j\bbeta_0\right)}_{A_2} \\
& - \underbrace{\mathbb{V}^{-\frac{1}{2}}
        ( \sum_{j}M_j)
{R_{\phi}^{-1}\left(\sum_{j}\hat{M}_j - V_j\right)}\mathbb{V}^{\frac{1}{2}}}_{A1}\underbrace{ \mathbb{V}^{-\frac{1}{2}}\phi(\sum_{j}\hat{M}_j - V_j)^{-1}}_{A_3}\bbeta_0
\end{align*}
where $R_\phi^{-1} (\cdot):= (R_\phi (\cdot))^{-1}$.

By Lemma \ref{lemma5}, we have that $A_2$ converges in distribution to $N(\bm 0, I_K)$. Hence, we aim to show that $A_1$ converges in probability to $I_K$, and $A_3$ converges in probability to $\bm 0$. Once these are proved, the desired asymptotic normality result follows by applying Slutsky's theorem.

We first establish that 
\begin{align} \label{proof theorem 2: key identity}
    A_1 = \mathbb{V}^{-\frac{1}{2}}
        ( \sum_{j}M_j){R_{\phi}^{-1}\left(\sum_{j}\hat{M}_j - V_j\right)}\mathbb{V}^{\frac{1}{2}} \xrightarrow[]{P} I_{K}
\end{align}
We study the inverse of the left-hand side i.e., $A_1^{-1}$, in \eqref{proof theorem 2: key identity} and write it as follows:
\begin{align*}
  & \mathbb{V}^{-\frac{1}{2}}
        {R_{\phi}\left(\sum_{j}\hat{M}_j - V_j\right)}( \sum_{j}M_j)^{-1}\mathbb{V}^{\frac{1}{2}} = \mathbb{V}^{-\frac{1}{2}}
        \left(\sum_{j} (\hat{M}_j - V_j) + \phi (\sum_{j}\hat{M}_j - V_j)^{-1} \right)( \sum_{j}M_j)^{-1}\mathbb{V}^{\frac{1}{2}} \\
        & = I_K + \mathbb{V}^{-\frac{1}{2}}
        \left(\sum_{j} (\hat{M}_j - (M_j + V_j)) + \phi (\sum_{j}\hat{M}_j - V_j)^{-1}\right)( \sum_{j}M_j)^{-1}\mathbb{V}^{\frac{1}{2}} \\
        & = I_K + \underbrace{\mathbb{V}^{-\frac{1}{2}}
        \left(\sum_{j} (\hat{M}_j - (M_j + V_j))\right)( \sum_{j}M_j)^{-1}\mathbb{V}^{\frac{1}{2}}}_{B_1} + \underbrace{\phi \mathbb{V}^{-\frac{1}{2}}
         (\sum_{j}\hat{M}_j - V_j)^{-1}( \sum_{j}M_j)^{-1}\mathbb{V}^{\frac{1}{2}}}_{B_2}
\end{align*}
For the term $B_1$, we further write it as follows:
\begin{align}\label{proof theorem2: a1 term}
   & \mathbb{V}^{-\frac{1}{2}}
        \left(\sum_{j} (\hat{M}_j - (M_j + V_j))\right)( \sum_{j}M_j)^{-1}\mathbb{V}^{\frac{1}{2}} \nonumber \\
        & =  \mathbb{V}^{-\frac{1}{2}} S_n \left(S_n^{-1}\sum_{j} (\hat{M}_j - (M_j + V_j))S_n^{-T}\right)(S_n^{-1} \sum_{j}M_j S_n^{-T})^{-1}S_n^{-1}\mathbb{V}^{\frac{1}{2}} \nonumber \\
        & = \mathbb{V}^{-\frac{1}{2}} T_n {\rm{{\rm diag}}}(\sqrt{\frac{\mu_{n1}}{\mu_{n1} + p}},...,\sqrt{\frac{\mu_{nK}}{\mu_{nK} + p}}) \left(S_n^{-1}\sum_{j} (\hat{M}_j - (M_j + V_j))S_n^{-T}\right) \nonumber \\
        & (S_n^{-1} \sum_{j}M_j S_n^{-T})^{-1}{\rm{{\rm diag}}}(\sqrt{\frac{\mu_{n1} + p}{\mu_{n1}}},...,\sqrt{\frac{\mu_{nK} + p}{\mu_{nk}}})T_n^{-1}\mathbb{V}^{\frac{1}{2}} 
\end{align}
where $T_n = \tilde S_n {\rm diag}(\sqrt{\mu_{n1} + p},...,\sqrt{\mu_{nk} + p})$. To establish that $B_1$ converges in probability to $\bm 0$, we note that for any $t, s = 1,...,K$, the $(t,s)$ entry of the inner matrix (the matrix between $\mathbb{V}^{-\frac{1}{2}}T_n$ and $T_n^{-1}\mathbb{V}^{\frac{1}{2}}$ in \eqref{proof theorem2: a1 term}) is
\begin{align*}
    \sqrt{\frac{\mu_{nt}}{\mu_{nt} + p}}\sqrt{\frac{\mu_{ns} + p}{\mu_{ns}}} D_{[t,s]} \ \text{where} \ D = \underbrace{\left(S_n^{-1}\sum_{j} (\hat{M}_j - (M_j + V_j))S_n^{-T}\right)}_{D_1}
    \underbrace{(S_n^{-1} \sum_{j}M_j S_n^{-T})^{-1}}_{D_2}
\end{align*}
Note that $D_{[t,s]} = \sum_{l=1}^K D_{1[t,l]}D_{2[l,s]}$. By Assumption 2, $D_{2}$ is bounded. Hence, we aim to show that, for every $l = 1,...,K$,
\begin{align}\label{proof theorem 2: dts op}
    \sqrt{\frac{\mu_{nt}}{\mu_{nt} + p}}\sqrt{\frac{\mu_{ns} + p}{\mu_{ns}}} D_{1[t,l]} \xrightarrow[]{P} 0
\end{align}
From the proof of Lemma \ref{lemma6}, we have $Var(D_{1[t,l]}) \le O((\mu_{nt} + \mu_{nl} + p)/(\mu_{nt}\mu_{nl}))$. Then, for any $\epsilon > 0$,
\begin{align*}
  & P(|\sqrt{\frac{\mu_{nt}}{\mu_{nt} + p}}\sqrt{\frac{\mu_{ns} + p}{\mu_{ns}}} D_{1[t,l]}| \ge \epsilon)  \le \frac{1}{\epsilon} Var(\sqrt{\frac{\mu_{nt}}{\mu_{nt} + p}}\sqrt{\frac{\mu_{ns} + p}{\mu_{ns}}} D_{1[t,l]}) \\
  &\le  O(\frac{\mu_{nt}}{\mu_{nt} + p}\frac{\mu_{ns} + p}{\mu_{ns}}\frac{\mu_{nt} + \mu_{nl} + p}{\mu_{nt}\mu_{nl}})\\
  & = O(\frac{\mu_{ns}\mu_{nt} + \mu_{ns}\mu_{nl} + \mu_{ns}p + \mu_{nt}p + \mu_{nl}p + p^2}{(\mu_{nt}+p)\mu_{ns}\mu_{nl}}) \\
  & = O(\frac{\mu_{nt}}{(\mu_{nt}+p)\mu_{nl}} + \frac{1}{\mu_{nt}+p} + \frac{p}{(\mu_{nt}+p)\mu_{nl}} + \frac{\mu_{nt}p}{(\mu_{nt}+p)\mu_{ns}\mu_{nl}} + \frac{p}{(\mu_{nt}+p)\mu_{ns}} + \frac{p^2}{(\mu_{nt}+p)\mu_{ns}\mu_{nl}})\\
  & \to 0
\end{align*}
 where $\mu_{nk} \rightarrow \infty$ for all $k$ and $p/(\mu_{ns}\mu_{nl}) \rightarrow 0$ for every $s$ and $l$, which are implied by $\mu_{n,min}/\sqrt{p} \rightarrow \infty$, and $\mu_{nt}/(\mu_{nt}+p) = O(1)$, and $p/(\mu_{nt}+p) = O(1)$. Hence, we have $ \sqrt{\frac{\mu_{nt}}{\mu_{nt} + p}}\sqrt{\frac{\mu_{ns} + p}{\mu_{ns}}} D_{[t,s]}$ converges in probability to 0 for any $t,s = 1,...,K$. Lastly, because $\mathbb{V}^{-\frac{1}{2}}T_n$ and $T_n^{-1}\mathbb{V}^{\frac{1}{2}}$ are bounded by Lemma \ref{lemma1}, we have $A_1$ converges in probability to $\bm 0$.

For the term $B_2$, we expand it as follows:
\begin{align*}
    & B_2 = \phi \mathbb{V}^{-\frac{1}{2}}
         (\sum_{j}\hat{M}_j - V_j)^{-1}( \sum_{j}M_j)^{-1}\mathbb{V}^{\frac{1}{2}} = \\  
        &  \underbrace{\sqrt{\phi}\mathbb{V}^{-\frac{1}{2}} (\sum_j M_j)^{-1} \mathbb{V}^{\frac{1}{2}}}_{C_1} \underbrace{\mathbb{V}^{-\frac{1}{2}} (\sum_j M_j) (\sum_j \hat M_j - V_j)^{-1} \mathbb{V}^{\frac{1}{2}}}_{C_2} \underbrace{\sqrt{\phi} \mathbb{V}^{-\frac{1}{2}} (\sum_j M_j)^{-1} \mathbb{V}^{\frac{1}{2}}}_{C_1}
\end{align*}
Note that $C_2^{-1} = I_K + B_1$. We have just shown that $B_1$ converges in probability to $\bm 0$ when $\mu_{n,min}/\sqrt{p} \rightarrow \infty$. Hence, $C_2$ converges in probability to $I_K$ by the continuous mapping theorem. 

It remains to show $C_1 \xrightarrow[]{P} \bm 0$.
Note that
\begin{align}\label{proof theorem2: a2 term}
    & C_1 = \sqrt{\phi} \mathbb{V}^{-\frac{1}{2}} S_n S_n^{-1}S_n^{-T}(S_n^{-1}\sum_j M_j S_n^{-T})^{-1} S_n^{-1} \mathbb{V}^{\frac{1}{2}} \nonumber \\
    & = \sqrt{\phi} \mathbb{V}^{-\frac{1}{2}} T_n {\rm{{\rm diag}}}(\sqrt{\frac{\mu_{n1}}{\mu_{n1} + p}},...,\sqrt{\frac{\mu_{nK}}{\mu_{nK} + p}}) S_n^{-1}S_n^{-T}(S_n^{-1}\sum_j M_j S_n^{-T})^{-1} \nonumber \\
    & \quad \quad {\rm{{\rm diag}}}(\sqrt{\frac{\mu_{n1} + p}{\mu_{n1}}},...,\sqrt{\frac{\mu_{nK} + p}{\mu_{nk}}})T_n^{-1} \mathbb{V}^{\frac{1}{2}} \nonumber \\
    & = \sqrt{\phi} \mathbb{V}^{-\frac{1}{2}} T_n {\rm{{\rm diag}}}(\sqrt{\frac{\mu_{n1}}{\mu_{n1} + p}},...,\sqrt{\frac{\mu_{nK}}{\mu_{nK} + p}}) E\ {\rm{{\rm diag}}}(\sqrt{\frac{\mu_{n1} + p}{\mu_{n1}}},...,\sqrt{\frac{\mu_{nK} + p}{\mu_{nk}}})  T_n^{-1} \mathbb{V}^{\frac{1}{2}} 
\end{align}
Consider any $(t,s)$ entry of the inner matrix (the matrix between $\mathbb{V}^{-\frac{1}{2}}T_n$ and $T_n^{-1}\mathbb{V}^{\frac{1}{2}}$ in \eqref{proof theorem2: a2 term}), which is equal to
\begin{align*}
    \sqrt{\frac{\mu_{nt}}{\mu_{nt} + p}}\sqrt{\frac{\mu_{ns} + p}{\mu_{ns}}} E_{ts} \ \text{where} \ E = \underbrace{S_n^{-1}S_n^{-T}}_{E_1}
    \underbrace{(S_n^{-1} \sum_{j}M_j S_n^{-T})^{-1}}_{E_2}
\end{align*}
Note that $E_{ts} = \sum_{l=1}^K E_{1[t,l]} E_{2[l,s]}$. By Assumption 2, $E_2$ is bounded. Hence, we aim to show that, for every $l = 1,..., K$,
\begin{align}\label{proof theorem 2: ets op}
    \sqrt{\phi} \sqrt{\frac{\mu_{nt}}{\mu_{nt} + p}}\sqrt{\frac{\mu_{ns} + p}{\mu_{ns}}} E_{1[t,l]} \xrightarrow[]{P} 0
\end{align}
By Assumption 2, $E_{1[t,l]} = O(1/\sqrt{\mu_{nt}\mu_{nl}})$. Then, we have
\begin{align*}
    \sqrt{\phi} \sqrt{\frac{\mu_{nt}}{\mu_{nt} + p}}\sqrt{\frac{\mu_{ns} + p}{\mu_{ns}}} E_{1[t,l]} & = \sqrt{\phi O(\frac{\mu_{nt} (\mu_{ns} + p)}{(\mu_{nt} + p)\mu_{ns}\mu_{nt}\mu_{nl}})} \\
    & = \sqrt{\phi O(\frac{1}{(\mu_{nt}+p)\mu_{nl}} + \frac{p}{(\mu_{nt}+p)\mu_{ns}\mu_{nl}})}
\end{align*}
Because $\phi = O_p(\mu_{n,min}+p)$ and $\mu_{n,min}/\sqrt{p} \rightarrow \infty$, \eqref{proof theorem 2: ets op} holds for every $l$. As a result,\\ $\sqrt{\phi}\sqrt{\frac{\mu_{nt}}{\mu_{nt} + p}}\sqrt{\frac{\mu_{ns} + p}{\mu_{ns}}} E_{ts} \xrightarrow[]{P} 0$ for every $(t,s)$. Now, because $\mathbb{V}^{-\frac{1}{2}}T_n$ and $T_n^{-1}\mathbb{V}^{\frac{1}{2}}$ are bounded by Lemma \ref{lemma1}, we have $C_1$ converges in probability to $\bm 0$. As a result, $B_2$ converges in probability to $\bm 0$.

Putting the above results together, we have \eqref{proof theorem 2: key identity} holds for any $\phi = O_p(\mu_{n,min}+p)$. That is, $A_1 \xrightarrow[]{P} I_K$.

{To conclude the proof of asymptotic normality of the SRIVW estimator, it remains to show $A_3 =  \phi \mathbb{V}^{-\frac{1}{2}}(\sum_{j}\hat{M}_j - V_j)^{-1} \xrightarrow[]{P} \bm 0$. Note that
\begin{align*}
    & A_3 = \phi \mathbb{V}^{-\frac{1}{2}}(\sum_{j}\hat{M}_j - V_j)^{-1} \\
    & =  \underbrace{\sqrt{\phi} \mathbb{V}^{-\frac{1}{2}} (\sum_j M_j)^{-1} \mathbb{V}^{\frac{1}{2}}}_{C_1} \underbrace{\mathbb{V}^{-\frac{1}{2}} (\sum_j M_j) (\sum_j \hat M_j - V_j)^{-1} \mathbb{V}^{\frac{1}{2}}}_{C_2} \sqrt{\phi}\mathbb{V}^{-\frac{1}{2}}.
\end{align*}
We have just shown from the analysis of $B_2$ that $C_1 \xrightarrow[]{P} \bm 0$, and $C_2 \xrightarrow[]{P} I_K$. Now, note that
\begin{align*}
    \sqrt{\phi}\mathbb{V}^{-\frac{1}{2}} & = \sqrt{\phi}\mathbb{V}^{-\frac{1}{2}}T_nT_n^{-1} \\
    & = \mathbb{V}^{-\frac{1}{2}}T_n {\rm diag}(\sqrt{\frac{\phi}{\mu_{n1}+p}},...,\sqrt{\frac{\phi}{\mu_{nK}+p}}) \tilde S_n^{-1}
\end{align*}
By Lemma \ref{lemma1}, $\mathbb{V}^{-\frac{1}{2}}T_n$ is bounded, and by Assumption 2, $\tilde S_n$ is also bounded. Combining it with the assumption that $\phi = O_p(\mu_{n,min}+p)$, we have that $\sqrt{\phi}\mathbb{V}^{-\frac{1}{2}} = O_p(1)$. Thus, $A_3 \xrightarrow[]{P} \bm 0$. }

So far, we have shown that under Theorem 2's conditions, $A_1 \xrightarrow[]{P} I_K$, $A_2 \xrightarrow[]{D} N(\bm 0, I_K)$, and $A_3 \xrightarrow[]{P} \bm 0$. Because $\mathbb{V}^{-\frac{1}{2}}
        ( \sum_{j}M_j)\left(\bm{\hat{\bbeta}}_{{\rm SRIVW},\phi} - \bm\bbeta_0\right) = A_1 A_2  - A_1 A_3$, it thus converges in distribution to $N(\bm 0, I_K)$ by  the Slutsky's theorem. Consistency follows directly from Lemma \ref{lemma7} that $\mathcal{V}_{\rm SRIVW} \rightarrow \bm 0$.   
\end{proof}

 \textbf{Consistency of variance estimators}

Lastly, we want to show that $\mathbb{V}^{-\frac{1}{2}}
        ( \sum_{j}M_j)\left(\bm{\hat{\bbeta}}_{{\rm SRIVW},\phi} - \bm\bbeta_0\right) \xrightarrow[]{D} N(\bm 0, I_K)$ holds if we replace $\sum_j M_j$ and $\mathbb{V}$ with respectively $R_{\phi}\left(\sum_{j}\hat{M}_j - V_j\right)$ and 
        \begin{align*}
            \hat{\mathbb{V}}_{\phi} = \sum_{j}\{(1+\hat{\bm\beta}_{\rm SRIVW, \phi}^TV_j\hat{\bm\beta}_{\rm SRIVW, \phi})\hat{M_j}+V_j\hat{\bm\beta}_{\rm SRIVW, \phi}\hat{\bm\beta}_{\rm SRIVW, \phi}^TV_j.
        \end{align*}
Note that
  \begin{align*}
      & \hat{\mathbb{V}}_{\phi}^{-1/2} {R_{\phi}\left(\sum_{j}\hat{M}_j - V_j\right)} \left(\bm{\hat{\bbeta}}_{{\rm SRIVW},\phi} - \bm\bbeta_0\right)\\
      & = \hat{\mathbb{V}}_{\phi}^{-1/2} {R_{\phi}\left(\sum_{j}\hat{M}_j - V_j\right)} (\sum_j M_j)^{-1} \mathbb{V}^{\frac{1}{2}} \mathbb{V}^{-\frac{1}{2}}
        ( \sum_{j}M_j)\left(\bm{\hat{\bbeta}}_{{\rm SRIVW},\phi} - \bm\bbeta_0\right) \\
     & = \underbrace{\hat{\mathbb{V}}_{\phi}^{-1/2} \mathbb{V}^{\frac{1}{2}}}_{F_1} \underbrace{ \mathbb{V}^{-\frac{1}{2}}{R_{\phi}\left(\sum_{j}\hat{M}_j - V_j\right)} (\sum_j M_j)^{-1} \mathbb{V}^{\frac{1}{2}}}_{F_2} \underbrace{\mathbb{V}^{-\frac{1}{2}}
        ( \sum_{j}M_j)\left(\bm{\hat{\bbeta}}_{{\rm SRIVW},\phi} - \bm\bbeta_0\right)}_{F_3}
  \end{align*}
  Using \eqref{proof theorem 2: key identity} and the continuous mapping theorem, we have $F_2 \xrightarrow[]{P} I_K$. We have just shown that $F_3 \xrightarrow[]{D} N(\bm 0, I_K)$. It remains to show $F_1 \xrightarrow[]{P} I_K$. 
  
  Following almost the same argument as in the proof of consistency of the variance estimator for MV-IVW (section S3.9) and using the fact that $\hat{\bm\beta}_{\rm SRIVW, \phi}$ is consistent under Theorem 2's conditions, we can first show that $\mathbb{V}^{-\frac{1}{2}}\hat{\mathbb{V}}_{\phi}\mathbb{V}^{-\frac{1}{2}} \xrightarrow[]{P} I_K$. Then, by the continuous mapping theorem, we have $\mathbb{V}^{\frac{1}{2}}\hat{\mathbb{V}}_{\phi}^{-1}\mathbb{V}^{\frac{1}{2}} \xrightarrow[]{P} I_K$. This implies $F_1 F_1^T \xrightarrow[]{P} I_K$. In other words, $F_1$ is an orthogonal matrix with probability tending 1. As a result, $F_1^{-1} = F_1^T$ with probability tending 1. This implies
  \begin{align*}
      F_1^{-1} =  \mathbb{V}^{-\frac{1}{2}} \hat{\mathbb{V}}_{\phi}^{1/2} = \mathbb{V}^{\frac{1}{2}} \hat{\mathbb{V}}_{\phi}^{-1/2} = F_1^T
  \end{align*}
  with probability tending 1. From this equality, we conclude that $\mathbb{V} = \hat{\mathbb{V}}_{\phi}$ with probability tending 1. Therefore, $F_1 \xrightarrow[]{P} I_K$. This completes the proof.

\subsection{Proof of claims related to SRIVW-pleio}

We show that the SRIVW estimators are still asymptotically normal and consistent under balanced horizontal pleiotropy. The proof basically mirrors the proof of Theorem 2, and is based on slight modifications of several established results.

We first state without proof those results and some important facts under balanced horizontal pleiotropy.  Under the balanced horizontal pleiotropy, we have $\hat{\Gamma}_j \sim N(\bgamma_j^T\bbeta_0, \sigma_{Yj}^2 + \tau_0^2)$ and $\sigma_{Xjk}^2/(\sigma_{Yj}^2 + \tau_0^2) = \Theta(1)$ for every $j, k$. In addition, we have
\begin{align}
    Cov(\sum_j \hat \bgamma_j\hat \Gamma_j\sigma_{Yj}^{-2} - \hat M_j\bbeta_0) = \mathbb{U} \label{eq: Uomega}
\end{align}
where $\mathbb{U} = \sum_{j=1}^p\{(1 + \tau_0^2\sigma_{Yj}^{-2} + \bm{\beta}_0^TV_j\bm{\beta}_0)(M_j + V_j) + V_j\bm{\beta}_0\bm{\beta}_0^TV_j\}$.

Under Theorem 2's conditions and the above balanced horizontal pleiotropy conditions,  we have the following intermediate results:
\begin{enumerate}
    \item [(1)] $T_n^{-1} \mathbb{U} T_n^{-T}$ is bounded with the minimum eigenvalue bounded away from 0.
    \item [(2)] $T_n^{-1}  \{\sum_j (1 + \tau_0^2\sigma_{Yj}^{-2} + \bbeta_0^T V_j \bbeta_0)(\hat M_j - (M_j+V_j))\} T_n^{-T} \xrightarrow[]{P} \bm 0$.
    \item [(3)] When $\max_j \frac{\gamma_{jk}^2}{\sigma_{Xjk}^2}/(\mu_{n,min} + p) \rightarrow 0$ for every $k = 1, ..., K$, $\mathbb{U}^{-\frac{1}{2}}(\sum_j \hat \bgamma_j\hat \Gamma_j\sigma_{Yj}^{-2} - \hat M_j\bbeta_0 + V_j\bbeta_0) \xrightarrow[]{D} N(\bm0, I_{K})$.
    \item [(4)] $\mathcal{V}_{\rm SRIVW} = (\sum_j M_j)^{-1}\mathbb{U}(\sum_j M_j)^{-1} \rightarrow \bm 0$
    \item  [(5)] $\mathbb{U}^{-\frac{1}{2}}( \sum_{j}M_j){R_{\phi}^{-1}\left(\sum_{j}\hat{M}_j - V_j\right)}\mathbb{U}^{\frac{1}{2}} \xrightarrow[]{P} I_{K}$
    \item [(6)] $\phi \mathbb{U}^{-\frac{1}{2}}(\sum_{j}\hat{M}_j - V_j)^{-1} \xrightarrow[]{P} \bm 0$
\end{enumerate}
The proof of these results is almost identical to their counterparts (with $\mathbb{U}$ replaced by $\mathbb{V}$) presented in earlier sections, by simply utilizing the fact that $1 + \tau_0^2 \sigma_{Yj}^{-2} + \bbeta_0^T V_j \bbeta_0 = \Theta(1)$ for all $j$. That is, this additional term $\tau_0^2 \sigma_{Yj}^{-2}$ has negligible impacts on the main arguments of the existing proofs.

Next, we show that 
\begin{align*}
    \mathbb{U}^{-\frac{1}{2}}
        ( \sum_{j}M_j) \left(\bm{\hat{\bbeta}}_{{\rm SRIVW},\phi} - \bm\bbeta_0\right) \xrightarrow[]{D} N(\bm 0, I_K)
\end{align*}
Note that 
\begin{align*}
       & \mathbb{U}^{-\frac{1}{2}}
        ( \sum_{j}M_j)\left(\hat{\bbeta}_{SRIVW} - \bm\bbeta_0\right) \\
        & = \underbrace{\mathbb{U}^{-\frac{1}{2}}
        ( \sum_{j}M_j)
{R_{\phi}^{-1}\left(\sum_{j}\hat{M}_j - V_j\right)}\mathbb{U}^{\frac{1}{2}}}_{A_1} \underbrace{\mathbb{U}^{-\frac{1}{2}}\left( \sum_j \hat \bgamma_j \hat \Gamma_j \sigma_{Yj}^{-2} - \hat M_j\bbeta_0 + V_j\bbeta_0\right)}_{A_2} \\
& - \underbrace{\mathbb{U}^{-\frac{1}{2}}
        ( \sum_{j}M_j)
{R_{\phi}^{-1}\left(\sum_{j}\hat{M}_j - V_j\right)}\mathbb{U}^{\frac{1}{2}}}_{A1}\underbrace{ \mathbb{U}^{-\frac{1}{2}}\phi(\sum_{j}\hat{M}_j - V_j)^{-1}}_{A_3}\bbeta_0
\end{align*}
Using the above results (3), (5) and (6), we have $A_1 \xrightarrow[]{P} I_K$, $A_2 \xrightarrow[]{D} N(\bm 0, I_K)$, and $A_3 \xrightarrow[]{P} \bm 0$. Then, applying Slutsky's theorem, we conclude that the SRIVW estimator is asymptotically normal under balanced horizontal pleiotropy. This result suggests the asymptotic variance of the SRIVW estimator is given by $\mathcal{V}_{\rm SRIVW} = (\sum_j M_j)^{-1}\mathbb{U}(\sum_j M_j)^{-1}$. The consistency of the SRIVW estimator follows from that $\mathcal{V}_{\rm SRIVW} \rightarrow \bm 0$.

For the rest of this section, we show that $\mathbb{U}^{-\frac{1}{2}}
        ( \sum_{j}M_j) \left(\bm{\hat{\bbeta}}_{{\rm SRIVW},\phi} - \bm\bbeta_0\right) \xrightarrow[]{D} N(\bm 0, I_K)$ holds if $\sum_j M_j$ and $\mathbb{U}$ are replaced with respectively $R_{\phi}(\sum_j \hat M_j - V_j)$ and $\hat{\mathbb{U}}_{\phi}$, where
\begin{align*}
     & \hat{\mathbb{U}}_{\phi} = \sum_{j=1}^p\left\{\left(1+\hat{\tau}_{\phi}^2\sigma_{Yj}^{-2}+\hat{\bm\beta}_{{\rm SRIVW},\phi}^TV_j\hat{\bm\beta}_{{\rm SRIVW},\phi}\right)\hat{M_j}+V_j\hat{\bm\beta}_{{\rm SRIVW},\phi}\hat{\bm\beta}_{{\rm SRIVW},\phi}^TV_j\right\}, \quad \text{and} \\
    &  \hat{\tau}^2_{\phi}   = \frac{1}{\sum_{j=1}^p\sigma_{Yj}^{-2}}\left\{\sum_{j=1}^{p} \left( (\hat \Gamma_j - \hat{\bgamma_j}^T\hat{\bbeta}_{{\rm SRIVW}, \phi})^2 - \sigma_{Yj}^2 - \hat{\bm\beta}_{{\rm SRIVW},\phi}^T\Sigma_{Xj}\hat{\bm\beta}_{{\rm SRIVW},\phi} \right)\sigma_{Yj}^{-2} \right\}.
\end{align*} This result implies that a consistent estimator of $\mathcal{V}_{\rm SRIVW}$ is given by $\hat{ \mathcal{V}}_{\rm SRIVW} = R_{\phi}^{-1}(\sum_j \hat M_j - V_j) \hat{\mathbb{U}}_{\phi} R_{\phi}^{-1}(\sum_j \hat M_j - V_j)$.

Note that
  \begin{align*}
      & \hat{\mathbb{U}}_{\phi}^{-1/2} {R_{\phi}\left(\sum_{j}\hat{M}_j - V_j\right)} \left(\bm{\hat{\bbeta}}_{{\rm SRIVW},\phi} - \bm\bbeta_0\right)\\
      & = \hat{\mathbb{U}}_{\phi}^{-1/2} {R_{\phi}\left(\sum_{j}\hat{M}_j - V_j\right)} (\sum_j M_j)^{-1} \mathbb{U}^{\frac{1}{2}} \mathbb{U}^{-\frac{1}{2}}
        ( \sum_{j}M_j)\left(\bm{\hat{\bbeta}}_{{\rm SRIVW},\phi} - \bm\bbeta_0\right) \\
     & = \underbrace{\hat{\mathbb{U}}_{\phi}^{-1/2} \mathbb{U}^{\frac{1}{2}}}_{F_1} \underbrace{ \mathbb{U}^{-\frac{1}{2}}{R_{\phi}\left(\sum_{j}\hat{M}_j - V_j\right)} (\sum_j M_j)^{-1} \mathbb{U}^{\frac{1}{2}}}_{F_2} \underbrace{\mathbb{U}^{-\frac{1}{2}}
        ( \sum_{j}M_j)\left(\bm{\hat{\bbeta}}_{{\rm SRIVW},\phi} - \bm\bbeta_0\right)}_{F_3}
  \end{align*}
  From the above stated intermediate results, we have $F_2 \xrightarrow[]{P} I_K$ and $F_3 \xrightarrow[]{D} N(\bm 0, I_K)$. It remains to show $F_1 \xrightarrow[]{P} I_K$. 

    Before that, we first show $\mathbb{U}^{-\frac{1}{2}}\hat{\mathbb{U}}_{\phi}\mathbb{U}^{-\frac{1}{2}} \xrightarrow[]{P} I_K$. Let $\hat{\mathbb{U}}^{'} = \sum_{j}\{(1+\tau_0^2\sigma_{Yj}^{-2}+\bbeta_0^TV_j\bbeta_0)\hat{M_j}+V_j\bbeta_0 \bbeta_0^TV_j\}$. It is sufficient to show that $\mathbb{U}^{-\frac{1}{2}} \hat{\mathbb{U}}^{'} \mathbb{V}^{-\frac{1}{2}}$ converges in probability to $I_K$, because the desired result follows by replacing $\bbeta_0$ and $\tau_0^2$ with respectively their consistent estimator $\hat \bbeta_{{\rm SRIVW},\phi}$ and $\hat \tau_{\phi}^2$. The consistency of $\hat \bbeta_{{\rm SRIVW},\phi}$ has already been established under the current setting. We will show the consistency of $\hat \tau_{\phi}^2$ at the end of this proof.
    
    To show $\mathbb{U}^{-\frac{1}{2}}\hat{\mathbb{U}}^{'}\mathbb{U}^{-\frac{1}{2}} \xrightarrow[]{P} I_K$, we can expand $\mathbb{U}^{-\frac{1}{2}} \hat{\mathbb{U}}^{'} \mathbb{U}^{-\frac{1}{2}}$ as follows:
    \begin{align*}
        \mathbb{U}^{-\frac{1}{2}} \hat{\mathbb{U}}^{'} \mathbb{U}^{-\frac{1}{2}} = 
        \mathbb{U}^{-\frac{1}{2}}T_n \bigg\{T_n^{-1} \hat{\mathbb{U}}^{'} T_n^{-T}\bigg\} T_n^T  \mathbb{U}^{-\frac{1}{2}}
    \end{align*}
    Note that
    \begin{align*}
        T_n^{-1} \hat{\mathbb{U}}^{'} T_n^{-T} = T_n^{-1} \mathbb{U} T_n^{-T} + T_n^{-1}  \{\sum_j (1 + \tau_0^2\sigma_{Yj}^{-2} + \bbeta_0^T V_j \bbeta_0)(\hat M_j - (M_j+V_j))\} T_n^{-T}
    \end{align*}
    Then,
    \begin{align*}
        \mathbb{U}^{-\frac{1}{2}} \hat{\mathbb{U}}^{'} \mathbb{U}^{-\frac{1}{2}} = 
        I_K + \mathbb{U}^{-\frac{1}{2}}T_n \bigg\{T_n^{-1}  \{\sum_j (1 + \tau_0^2\sigma_{Yj}^{-2} + \bbeta_0^T V_j \bbeta_0)(\hat M_j - (M_j+V_j))\} T_n^{-T}\bigg\} T_n^T  \mathbb{U}^{-\frac{1}{2}}
    \end{align*}
    Since $(1 + \tau_0^2\sigma_{Yj}^{-2} + \bbeta_0^T V_j \bbeta_0) = \Theta(1)$, by the same argument as in the proof of Lemma \ref{lemma2}, we can show that $T_n^{-1}  \{\sum_j (1 + \tau_0^2\sigma_{Yj}^{-2} + \bbeta_0^T V_j \bbeta_0)(\hat M_j - (M_j+V_j))\} T_n^{-T}$ converges in probability to $\bm 0$. Because $\mathbb{U}^{-\frac{1}{2}}T_n$ is bounded, $\mathbb{U}^{-\frac{1}{2}} \hat{\mathbb{U}}^{'} \mathbb{U}^{-\frac{1}{2}}$ converges in probability to $I_K$. Hence, given that $\hat \bbeta_{{\rm SRIVW},\phi}$ and $\hat \tau_{\phi}^2$ are consistent (the consistency of $\hat \tau_{\phi}^2$ will be proved soon), $\mathbb{U}^{-\frac{1}{2}} \hat{\mathbb{U}}_{\phi} \mathbb{U}^{-\frac{1}{2}} \xrightarrow[]{P} I_K$. 
  
  Then, by the continuous mapping theorem, we have $\mathbb{U}^{\frac{1}{2}}\hat{\mathbb{U}}_{\phi}^{-1}\mathbb{U}^{\frac{1}{2}} \xrightarrow[]{P} I_K$. This implies $F_1 F_1^T \xrightarrow[]{P} I_K$. In other words, $F_1$ is an orthogonal matrix with probability tending 1. As a result, $F_1^{-1} = F_1^T$ with probability tending 1. This implies
  \begin{align*}
      F_1^{-1} =  \mathbb{U}^{-\frac{1}{2}} \hat{\mathbb{U}}_{\phi}^{1/2} = \mathbb{U}^{\frac{1}{2}} \hat{\mathbb{U}}_{\phi}^{-1/2} = F_1^T
  \end{align*}
  with probability tending 1. From this equality, we conclude that $\mathbb{U} = \hat{\mathbb{U}}_{\phi}$ with probability tending 1. Therefore, $F_1 \xrightarrow[]{P} I_K$.

    To complete the proof, it remains to show that $\hat{\tau}^2_\phi$ is a consistent estimator of $\tau_0^2$. We only prove consistency in the case where $\phi = 0$, because the SRIVW estimator is consistent as long as $\phi = O_p (\mu_{n,min}+p)$. We hereafter drop the subscript $\phi$ in $\hat \tau_{\phi}^2$ to ease the notation. Recall that
    \begin{align*}
        \hat{\tau}^2   = \frac{1}{\sum_{j}\sigma_{Yj}^{-2}}\left\{\sum_j \left( (\hat \Gamma_j - \hat{\bgamma_j}^T\hat{\bbeta}_{\rm SRIVW})^2 - \sigma_{Yj}^2 - \hat{\bm\beta}_{\rm SRIVW}^T\Sigma_{Xj}\hat{\bm\beta}_{\rm SRIVW} \right)\sigma_{Yj}^{-2} \right\}
    \end{align*}
     We then define $\tilde  \tau $ as an analogue of $\hat \tau^2$ with $\hat \bbeta_{{\rm SRIVW}}$ replaced by $\bbeta_0$ i.e.
    \begin{align*}
        \tilde {\tau}^2   = \frac{1}{\sum_{j}\sigma_{Yj}^{-2}}\left\{\sum_j \left( (\hat \Gamma_j - \hat{\bgamma_j}^T\bbeta_0)^2 - \sigma_{Yj}^2 - \bbeta_0^T\Sigma_{Xj}\bbeta_0 \right)\sigma_{Yj}^{-2} \right\}
    \end{align*}
    Some algebra shows that
    \begin{align*}
        & E(\tilde \tau^2) = \tau_0^2 \\
        & Var(\tilde \tau^2) = \frac{2 \sum_j (1 + \bbeta_0^T V_j \bbeta_0 + \tau_0^2 \sigma_{Yj}^{-2})^2}{(\sum_j \sigma_{Yj}^{-2})^2}= \frac{\Theta(p)}{(\sum_j \sigma_{Yj}^{-2})^2} 
    \end{align*}
    In this proof, for a random variable X with finite second moments, we can write
    $X = E(X)+ O_P ((Var(X))^{1/2})$. With this, we can then write
    \begin{align*}
        (\sum_j \sigma_{Yj}^{-2}) (\tilde \tau^2 - \tau_0^2) = O_p(\sqrt{p})
    \end{align*}
    Furthermore, some algebra reveals that
    \begin{align*}
        & (\sum_j \sigma_{Yj}^{-2}) (\hat \tau^2 - \tilde \tau^2)  = -2(\hat \bbeta_{\rm SRIVW} - \bbeta_0)^T\bigg\{\sum_j \hat \bgamma_j \hat \Gamma_j \sigma_{Yj}^{-2} - \hat M_j\bbeta_0 + V_j\bbeta_0\bigg\}\\
        & + (\hat \bbeta_{\rm SRIVW} - \bbeta_0)^T (\sum_j \hat M_j - V_j) (\hat \bbeta_{\rm SRIVW} - \bbeta_0) \\
        & = - \bigg\{\sum_j \hat \bgamma_j \hat \Gamma_j \sigma_{Yj}^{-2} - \hat M_j\bbeta_0 + V_j\bbeta_0\bigg\}^T (\sum_j \hat M_j - V_j)^{-1} \bigg\{\sum_j \hat \bgamma_j \hat \Gamma_j \sigma_{Yj}^{-2} - \hat M_j\bbeta_0 + V_j\bbeta_0\bigg\}\\
        & = - \bigg\{\sum_j \hat \bgamma_j \hat \Gamma_j \sigma_{Yj}^{-2} - \hat M_j\bbeta_0 + V_j\bbeta_0\bigg\}^T \mathbb{U}^{-\frac{1}{2}}  \mathbb{U}^{\frac{1}{2}} (\sum_j \hat M_j - V_j)^{-1} \mathbb{U}^{\frac{1}{2}} \mathbb{U}^{-\frac{1}{2}} \bigg\{\sum_j \hat \bgamma_j \hat \Gamma_j \sigma_{Yj}^{-2} - \hat M_j\bbeta_0 + V_j\bbeta_0\bigg\}
    \end{align*}
    We know $\mathbb{U}^{-\frac{1}{2}} \bigg\{\sum_j \hat \bgamma_j \hat \Gamma_j \sigma_{Yj}^{-2} - \hat M_j\bbeta_0 + V_j\bbeta_0\bigg\}$ is asymptotically normal and hence $O_p(1)$. It remains to study $\mathbb{U}^{\frac{1}{2}} (\sum_j \hat M_j - V_j)^{-1} \mathbb{U}^{\frac{1}{2}}$. Expand it as follows:
    \begin{align*}
    & \mathbb{U}^{\frac{1}{2}} S_n^{-T} \bigg(S_n^{-1}(\sum_j \hat M_j - (M_j + V_j))S_n^{-T} + S_n^{-1} \sum_j M_j S_n^{-T}\bigg)^{-1} S_n^{-1}\mathbb{U}^{\frac{1}{2}}\\
    & = \mathbb{U}^{\frac{1}{2}} T_n^{-T} {\rm{{\rm diag}}}(\sqrt{\frac{\mu_{n1} + p}{\mu_{n1}}},...,\sqrt{\frac{\mu_{nK} + p}{\mu_{nk}}}) \underbrace{\bigg(S_n^{-1}(\sum_j \hat M_j - (M_j + V_j))S_n^{-T} + S_n^{-1} \sum_j M_j S_n^{-T}\bigg)^{-1}}_{B}\\
    & {\rm{{\rm diag}}}(\sqrt{\frac{\mu_{n1} + p}{\mu_{n1}}},...,\sqrt{\frac{\mu_{nK} + p}{\mu_{nk}}}) T_n^{-1}\mathbb{U}^{\frac{1}{2}}
    \end{align*}
    Consider any $(t,s)$ entry of the inner matrix (the matrix between $\mathbb{U}^{\frac{1}{2}}T_n^{-T}$ and $T_n^{-1}\mathbb{U}^{\frac{1}{2}}$ in the above display), which is equal to
\begin{align*}
    \sqrt{\frac{(\mu_{nt}+p)(\mu_{ns} + p)}{\mu_{nt}\mu_{ns}}} B_{ts}
\end{align*}
     By Lemma \ref{lemma6} and Assumption 2, $B$ is bounded in probability. Hence, when $\mu_{n,min}/\sqrt{p} \rightarrow \infty$,
     \begin{align}
      \sqrt{\frac{(\mu_{nt}+p)(\mu_{ns} + p)}{\mu_{nt}\mu_{ns}}} B_{ts} = O_p(\sqrt{\frac{(\mu_{nt}+p)(\mu_{ns} + p)}{\mu_{nt}\mu_{ns}}}) = o_p(\sqrt{p})   
     \end{align}
     That is, $\mathbb{U}^{\frac{1}{2}} (\sum_j \hat M_j - V_j)^{-1} \mathbb{U}^{\frac{1}{2}} = o_p(\sqrt{p})$. Because $T_n^{-1}\mathbb{U}^{\frac{1}{2}}$ is bounded, we have $(\sum_j \sigma_{Yj}^{-2}) (\hat \tau^2 - \tilde \tau^2) = o_p(\sqrt{p})$. Consequently,
     \begin{align*}
         \sum_j \sigma_{Yj}^{-2} (\hat \tau^2 - \tau_0^2) =\sum_j \sigma_{Yj}^{-2} (\hat \tau^2 - \tilde \tau_0^2) + \sum_j \sigma_{Yj}^{-2} (\tilde \tau^2 - \tau_0^2) = o_p(\sqrt{p}) +  O_p(\sqrt{p}) = O_p(\sqrt{p})
     \end{align*} 
    
    Finally, because $\sigma_{Yj}^2 = \Theta(\frac{1}{n})$ for every $j$ (see \ref{sec: rate sigma}), $\sum_j \sigma_{Yj}^{-2} = \Theta(np)$, which leads to
    \begin{align*}
        \hat \tau^2 - \tau_0^2 = O_p(\frac{1}{n\sqrt{p}}) = o_p(1)
    \end{align*}
    This completes the proof.

\section{Additional simulation details and results}\label{sec: additional sim}

\subsection{Main simulation study: three exposures and summary-data}\label{exposure-simulation-results.-all-beta.exposures-by-d.}

In the following, we present the additional simulation results not shown in the main text.

\begin{table}[H]

\caption{Simulation results for six MVMR estimators in the main simulation study with 10,000 repetitions, when Factor 1 = (ii), Factor 2 = (i), Factor 3 = (i). In other words, $\bbeta_0 = (\beta_{01}, \beta_{02}, \beta_{03}) = (0.1, -0.5, -0.9)$, and the first exposure has much weaker IV strength than the other two. $p = 145$. $\hat \lambda_{\min}$ is the minimum eigenvalue of the sample IV strength matrix $\sum_{j=1}^{p} \Omega_j^{-1} \hat\bgamma_j\hat \bgamma_j^T\Omega_j^{-T} - pI_K$.}
\centering
\resizebox{\linewidth}{!}{
\begin{tabular}[t]{ccccccccccccc}
\toprule
\multicolumn{1}{c}{ } & \multicolumn{4}{c}{$\beta_{01}$ = 0.1} & \multicolumn{4}{c}{$\beta_{02}$ = -0.5} & \multicolumn{4}{c}{$\beta_{03}$ = -0.9} \\
\cmidrule(l{3pt}r{3pt}){2-5} \cmidrule(l{3pt}r{3pt}){6-9} \cmidrule(l{3pt}r{3pt}){10-13}
estimator & Est & SD & SE & CP & Est & SD & SE & CP & Est & SD & SE & CP\\
\midrule
\addlinespace[0.3em]
\multicolumn{13}{l}{\textit{Mean $\hat \lambda_{\min}/\sqrt{p}$ = 103.4, mean conditional $F$-statistics = 9.9, 38.3, 23.5}}\\
\hspace{1em}\cellcolor{gray!6}{IVW} & \cellcolor{gray!6}{0.099} & \cellcolor{gray!6}{0.035} & \cellcolor{gray!6}{0.035} & \cellcolor{gray!6}{0.952} & \cellcolor{gray!6}{-0.495} & \cellcolor{gray!6}{0.011} & \cellcolor{gray!6}{0.011} & \cellcolor{gray!6}{0.925} & \cellcolor{gray!6}{-0.883} & \cellcolor{gray!6}{0.016} & \cellcolor{gray!6}{0.016} & \cellcolor{gray!6}{0.815}\\
\hspace{1em}MV-Egger & 0.099 & 0.046 & 0.046 & 0.946 & -0.495 & 0.011 & 0.010 & 0.899 & -0.883 & 0.016 & 0.016 & 0.814\\
\hspace{1em}\cellcolor{gray!6}{MV-Median} & \cellcolor{gray!6}{0.099} & \cellcolor{gray!6}{0.042} & \cellcolor{gray!6}{0.045} & \cellcolor{gray!6}{0.968} & \cellcolor{gray!6}{-0.494} & \cellcolor{gray!6}{0.014} & \cellcolor{gray!6}{0.016} & \cellcolor{gray!6}{0.952} & \cellcolor{gray!6}{-0.884} & \cellcolor{gray!6}{0.021} & \cellcolor{gray!6}{0.023} & \cellcolor{gray!6}{0.910}\\
\hspace{1em}GRAPPLE & 0.099 & 0.039 & 0.040 & 0.958 & -0.500 & 0.011 & 0.011 & 0.952 & -0.900 & 0.017 & 0.017 & 0.956\\
\hspace{1em}\cellcolor{gray!6}{MRBEE} & \cellcolor{gray!6}{0.100} & \cellcolor{gray!6}{0.040} & \cellcolor{gray!6}{0.040} & \cellcolor{gray!6}{0.944} & \cellcolor{gray!6}{-0.500} & \cellcolor{gray!6}{0.012} & \cellcolor{gray!6}{0.012} & \cellcolor{gray!6}{0.937} & \cellcolor{gray!6}{-0.901} & \cellcolor{gray!6}{0.017} & \cellcolor{gray!6}{0.018} & \cellcolor{gray!6}{0.948}\\
\hspace{1em}SRIVW & 0.100 & 0.039 & 0.039 & 0.953 & -0.500 & 0.012 & 0.012 & 0.949 & -0.901 & 0.017 & 0.017 & 0.950\\
\addlinespace[0.3em]
\multicolumn{13}{l}{\textit{Mean $\hat \lambda_{\min}/\sqrt{p}$ = 21.7, mean conditional $F$-statistics = 2.9, 28.0, 14.7}}\\
\hspace{1em}\cellcolor{gray!6}{MV-IVW} & \cellcolor{gray!6}{0.060} & \cellcolor{gray!6}{0.065} & \cellcolor{gray!6}{0.065} & \cellcolor{gray!6}{0.906} & \cellcolor{gray!6}{-0.495} & \cellcolor{gray!6}{0.011} & \cellcolor{gray!6}{0.011} & \cellcolor{gray!6}{0.928} & \cellcolor{gray!6}{-0.886} & \cellcolor{gray!6}{0.016} & \cellcolor{gray!6}{0.016} & \cellcolor{gray!6}{0.857}\\
\hspace{1em}MV-Egger & 0.065 & 0.091 & 0.090 & 0.924 & -0.495 & 0.011 & 0.010 & 0.898 & -0.886 & 0.016 & 0.016 & 0.858\\
\hspace{1em}\cellcolor{gray!6}{MV-Median} & \cellcolor{gray!6}{0.061} & \cellcolor{gray!6}{0.081} & \cellcolor{gray!6}{0.081} & \cellcolor{gray!6}{0.923} & \cellcolor{gray!6}{-0.494} & \cellcolor{gray!6}{0.014} & \cellcolor{gray!6}{0.016} & \cellcolor{gray!6}{0.954} & \cellcolor{gray!6}{-0.887} & \cellcolor{gray!6}{0.020} & \cellcolor{gray!6}{0.022} & \cellcolor{gray!6}{0.926}\\
\hspace{1em}GRAPPLE & 0.097 & 0.101 & 0.103 & 0.958 & -0.500 & 0.011 & 0.012 & 0.955 & -0.900 & 0.017 & 0.018 & 0.955\\
\hspace{1em}\cellcolor{gray!6}{MRBEE} & \cellcolor{gray!6}{0.104} & \cellcolor{gray!6}{0.108} & \cellcolor{gray!6}{0.111} & \cellcolor{gray!6}{0.955} & \cellcolor{gray!6}{-0.500} & \cellcolor{gray!6}{0.012} & \cellcolor{gray!6}{0.012} & \cellcolor{gray!6}{0.938} & \cellcolor{gray!6}{-0.901} & \cellcolor{gray!6}{0.018} & \cellcolor{gray!6}{0.019} & \cellcolor{gray!6}{0.950}\\
\hspace{1em}SRIVW & 0.102 & 0.107 & 0.107 & 0.958 & -0.500 & 0.012 & 0.012 & 0.950 & -0.900 & 0.018 & 0.018 & 0.951\\
\addlinespace[0.3em]
\multicolumn{13}{l}{\textit{Mean $\hat \lambda_{\min}/\sqrt{p}$ = 7.6, mean conditional $F$-statistics = 1.7, 28.7, 17.3}}\\
\hspace{1em}\cellcolor{gray!6}{MV-IVW} & \cellcolor{gray!6}{0.011} & \cellcolor{gray!6}{0.087} & \cellcolor{gray!6}{0.086} & \cellcolor{gray!6}{0.820} & \cellcolor{gray!6}{-0.496} & \cellcolor{gray!6}{0.011} & \cellcolor{gray!6}{0.011} & \cellcolor{gray!6}{0.934} & \cellcolor{gray!6}{-0.887} & \cellcolor{gray!6}{0.015} & \cellcolor{gray!6}{0.015} & \cellcolor{gray!6}{0.867}\\
\hspace{1em}MV-Egger & 0.018 & 0.129 & 0.122 & 0.878 & -0.495 & 0.011 & 0.010 & 0.906 & -0.887 & 0.015 & 0.015 & 0.871\\
\hspace{1em}\cellcolor{gray!6}{MV-Median} & \cellcolor{gray!6}{0.015} & \cellcolor{gray!6}{0.107} & \cellcolor{gray!6}{0.100} & \cellcolor{gray!6}{0.846} & \cellcolor{gray!6}{-0.494} & \cellcolor{gray!6}{0.014} & \cellcolor{gray!6}{0.016} & \cellcolor{gray!6}{0.954} & \cellcolor{gray!6}{-0.888} & \cellcolor{gray!6}{0.020} & \cellcolor{gray!6}{0.022} & \cellcolor{gray!6}{0.931}\\
\hspace{1em}GRAPPLE & 0.098 & 0.227 & 0.228 & 0.958 & -0.500 & 0.011 & 0.012 & 0.960 & -0.900 & 0.020 & 0.020 & 0.958\\
\hspace{1em}\cellcolor{gray!6}{MRBEE} & \cellcolor{gray!6}{0.127} & \cellcolor{gray!6}{0.299} & \cellcolor{gray!6}{0.354} & \cellcolor{gray!6}{0.976} & \cellcolor{gray!6}{-0.500} & \cellcolor{gray!6}{0.014} & \cellcolor{gray!6}{0.013} & \cellcolor{gray!6}{0.949} & \cellcolor{gray!6}{-0.899} & \cellcolor{gray!6}{0.025} & \cellcolor{gray!6}{0.026} & \cellcolor{gray!6}{0.969}\\
\hspace{1em}SRIVW & 0.104 & 0.223 & 0.223 & 0.962 & -0.500 & 0.012 & 0.012 & 0.956 & -0.900 & 0.019 & 0.019 & 0.961\\
\bottomrule
\multicolumn{13}{l}{\small Abbreviations: Est = estimated causal effect; SD = standard deviation; SE = standard error; CP = coverage probability.}
\end{tabular}}
\end{table}

\begin{table}[H]
\caption{Simulation results for six MVMR estimators in the main simulation study with 10,000 repetitions, when Factor 1 = (i), Factor 2 = (ii), Factor 3 = (ii). In other words, $\bbeta_0 = (\beta_{01}, \beta_{02}, \beta_{03}) = (0.8, 0.4, 0)$,  and all three exposures have similar IV strength. $p = 145$. $\hat \lambda_{\min}$ is the minimum eigenvalue of the sample IV strength matrix $\sum_{j=1}^{p} \Omega_j^{-1} \hat\bgamma_j \hat\bgamma_j^T \Omega_j^{-T} - pI_K$.}
\centering
\resizebox{\linewidth}{!}{
\begin{tabular}[t]{c|cccc|cccc|cccc}
\toprule
\multicolumn{1}{c}{ } & \multicolumn{4}{c}{$\beta_{01}$ = 0.8} & \multicolumn{4}{c}{$\beta_{02}$ = 0.4} & \multicolumn{4}{c}{$\beta_{03}$ = 0} \\
\cmidrule(l{3pt}r{3pt}){2-5} \cmidrule(l{3pt}r{3pt}){6-9} \cmidrule(l{3pt}r{3pt}){10-13}
Estimator & Est & SD & SE & CP & Est & SD & SE & CP & Est & SD & SE & CP\\
\midrule
\addlinespace[0.3em]
\multicolumn{13}{l}{\textit{Mean $\hat \lambda_{\min}/\sqrt{p}$ = 109.1, mean conditional $F$-statistics = 13.3, 19.0, 11.1}}\\
\hspace{1em}\cellcolor{gray!6}{MV-IVW} & \cellcolor{gray!6}{0.746} & \cellcolor{gray!6}{0.022} & \cellcolor{gray!6}{0.023} & \cellcolor{gray!6}{0.327} & \cellcolor{gray!6}{0.383} & \cellcolor{gray!6}{0.017} & \cellcolor{gray!6}{0.018} & \cellcolor{gray!6}{0.832} & \cellcolor{gray!6}{-0.023} & \cellcolor{gray!6}{0.025} & \cellcolor{gray!6}{0.025} & \cellcolor{gray!6}{0.840}\\
\hspace{1em}MV-Egger & 0.762 & 0.029 & 0.030 & 0.755 & 0.382 & 0.017 & 0.017 & 0.795 & -0.025 & 0.025 & 0.025 & 0.824\\
\hspace{1em}\cellcolor{gray!6}{MV-Median} & \cellcolor{gray!6}{0.749} & \cellcolor{gray!6}{0.029} & \cellcolor{gray!6}{0.034} & \cellcolor{gray!6}{0.690} & \cellcolor{gray!6}{0.381} & \cellcolor{gray!6}{0.022} & \cellcolor{gray!6}{0.026} & \cellcolor{gray!6}{0.918} & \cellcolor{gray!6}{-0.021} & \cellcolor{gray!6}{0.033} & \cellcolor{gray!6}{0.037} & \cellcolor{gray!6}{0.938}\\
\hspace{1em}GRAPPLE & 0.799 & 0.025 & 0.025 & 0.954 & 0.399 & 0.019 & 0.019 & 0.954 & -0.001 & 0.028 & 0.029 & 0.952\\
\hspace{1em}\cellcolor{gray!6}{MRBEE} & \cellcolor{gray!6}{0.804} & \cellcolor{gray!6}{0.026} & \cellcolor{gray!6}{0.027} & \cellcolor{gray!6}{0.942} & \cellcolor{gray!6}{0.401} & \cellcolor{gray!6}{0.019} & \cellcolor{gray!6}{0.020} & \cellcolor{gray!6}{0.941} & \cellcolor{gray!6}{0.001} & \cellcolor{gray!6}{0.029} & \cellcolor{gray!6}{0.030} & \cellcolor{gray!6}{0.948}\\
\hspace{1em}SRIVW & 0.802 & 0.026 & 0.026 & 0.952 & 0.400 & 0.019 & 0.019 & 0.954 & 0.001 & 0.029 & 0.029 & 0.950\\
\addlinespace[0.3em]
\multicolumn{13}{l}{\textit{Mean $\hat \lambda_{\min}/\sqrt{p}$ = 21.1, mean conditional $F$-statistics = 3.6, 4.7, 3.1}}\\
\hspace{1em}\cellcolor{gray!6}{MV-IVW} & \cellcolor{gray!6}{0.588} & \cellcolor{gray!6}{0.038} & \cellcolor{gray!6}{0.041} & \cellcolor{gray!6}{0.002} & \cellcolor{gray!6}{0.327} & \cellcolor{gray!6}{0.031} & \cellcolor{gray!6}{0.032} & \cellcolor{gray!6}{0.380} & \cellcolor{gray!6}{-0.070} & \cellcolor{gray!6}{0.043} & \cellcolor{gray!6}{0.042} & \cellcolor{gray!6}{0.609}\\
\hspace{1em}MV-Egger & 0.622 & 0.054 & 0.057 & 0.129 & 0.328 & 0.032 & 0.032 & 0.378 & -0.073 & 0.044 & 0.045 & 0.635\\
\hspace{1em}\cellcolor{gray!6}{MV-Median} & \cellcolor{gray!6}{0.585} & \cellcolor{gray!6}{0.055} & \cellcolor{gray!6}{0.056} & \cellcolor{gray!6}{0.036} & \cellcolor{gray!6}{0.323} & \cellcolor{gray!6}{0.042} & \cellcolor{gray!6}{0.043} & \cellcolor{gray!6}{0.569} & \cellcolor{gray!6}{-0.069} & \cellcolor{gray!6}{0.058} & \cellcolor{gray!6}{0.058} & \cellcolor{gray!6}{0.769}\\
\hspace{1em}GRAPPLE & 0.792 & 0.062 & 0.062 & 0.941 & 0.398 & 0.045 & 0.046 & 0.955 & -0.004 & 0.073 & 0.074 & 0.955\\
\hspace{1em}\cellcolor{gray!6}{MRBEE} & \cellcolor{gray!6}{0.821} & \cellcolor{gray!6}{0.078} & \cellcolor{gray!6}{0.080} & \cellcolor{gray!6}{0.961} & \cellcolor{gray!6}{0.406} & \cellcolor{gray!6}{0.049} & \cellcolor{gray!6}{0.051} & \cellcolor{gray!6}{0.949} & \cellcolor{gray!6}{0.010} & \cellcolor{gray!6}{0.083} & \cellcolor{gray!6}{0.086} & \cellcolor{gray!6}{0.962}\\
\hspace{1em}SRIVW & 0.813 & 0.076 & 0.075 & 0.957 & 0.404 & 0.049 & 0.049 & 0.956 & 0.007 & 0.081 & 0.081 & 0.957\\
\addlinespace[0.3em]
\multicolumn{13}{l}{\textit{Mean $\hat \lambda_{\min}/\sqrt{p}$ = 7.2, mean conditional $F$-statistics = 2.0, 2.4, 1.8}}\\
\hspace{1em}\cellcolor{gray!6}{MV-IVW} & \cellcolor{gray!6}{0.406} & \cellcolor{gray!6}{0.046} & \cellcolor{gray!6}{0.048} & \cellcolor{gray!6}{0.000} & \cellcolor{gray!6}{0.246} & \cellcolor{gray!6}{0.039} & \cellcolor{gray!6}{0.039} & \cellcolor{gray!6}{0.038} & \cellcolor{gray!6}{-0.089} & \cellcolor{gray!6}{0.053} & \cellcolor{gray!6}{0.048} & \cellcolor{gray!6}{0.541}\\
\hspace{1em}MV-Egger & 0.442 & 0.072 & 0.073 & 0.003 & 0.246 & 0.039 & 0.042 & 0.043 & -0.090 & 0.053 & 0.054 & 0.623\\
\hspace{1em}\cellcolor{gray!6}{MV-Median} & \cellcolor{gray!6}{0.395} & \cellcolor{gray!6}{0.066} & \cellcolor{gray!6}{0.060} & \cellcolor{gray!6}{0.000} & \cellcolor{gray!6}{0.238} & \cellcolor{gray!6}{0.053} & \cellcolor{gray!6}{0.051} & \cellcolor{gray!6}{0.122} & \cellcolor{gray!6}{-0.087} & \cellcolor{gray!6}{0.069} & \cellcolor{gray!6}{0.062} & \cellcolor{gray!6}{0.692}\\
\hspace{1em}GRAPPLE & 0.781 & 0.125 & 0.123 & 0.933 & 0.393 & 0.087 & 0.089 & 0.959 & -0.007 & 0.156 & 0.158 & 0.957\\
\hspace{1em}\cellcolor{gray!6}{MRBEE} & \cellcolor{gray!6}{0.881} & \cellcolor{gray!6}{0.202} & \cellcolor{gray!6}{0.288} & \cellcolor{gray!6}{0.982} & \cellcolor{gray!6}{0.419} & \cellcolor{gray!6}{0.112} & \cellcolor{gray!6}{0.126} & \cellcolor{gray!6}{0.981} & \cellcolor{gray!6}{0.050} & \cellcolor{gray!6}{0.226} & \cellcolor{gray!6}{0.306} & \cellcolor{gray!6}{0.983}\\
\hspace{1em}SRIVW & 0.798 & 0.169 & 0.178 & 0.959 & 0.405 & 0.126 & 0.108 & 0.972 & -0.017 & 0.239 & 0.191 & 0.961\\
\bottomrule
\multicolumn{13}{l}{\small Abbreviations: Est = estimated causal effect; SD = standard deviation; SE = standard error; CP = coverage probability.}
\end{tabular}}
\end{table}

\clearpage

\begin{table}[H]

\caption{Simulation results for six MVMR estimators in simulation study 1 with 10,000 repetitions, when Factor 1 = (ii), Factor 2 = (ii), Factor 3 = (ii). In other words, $\bbeta_0 = (\beta_{01}, \beta_{02}, \beta_{03}) = (0.1, -0.5, -0.9)$, and all three exposures have similar IV strength. $p = 145$. $\hat \lambda_{\min}$ is the minimum eigenvalue of the sample IV strength matrix $\sum_{j=1}^{p} \Omega_j^{-1} \hat\bgamma_j\hat \bgamma_j^T\Omega_j^{-T} - pI_K$.}
\centering
\resizebox{\linewidth}{!}{
\begin{tabular}[t]{cccccccccccccc}
\toprule
\multicolumn{1}{c}{ } & \multicolumn{4}{c}{$\beta_{01}$ = 0.1} & \multicolumn{4}{c}{$\beta_{02}$ = -0.5} & \multicolumn{4}{c}{$\beta_{03}$ = -0.9} \\
\cmidrule(l{3pt}r{3pt}){2-5} \cmidrule(l{3pt}r{3pt}){6-9} \cmidrule(l{3pt}r{3pt}){10-13}
Estimator & Est & SD & SE & CP & Est & SD & SE & CP & Est & SD & SE & CP\\
\midrule
\addlinespace[0.3em]
\multicolumn{13}{l}{\textit{Mean $\hat \lambda_{\min}/\sqrt{p}$ = 109.1, mean Conditional $F$-statistics = 13.3, 19.0, 11.1}}\\
\hspace{1em}\cellcolor{gray!6}{MV-IVW} & \cellcolor{gray!6}{\red 0.118} & \cellcolor{gray!6}{0.027} & \cellcolor{gray!6}{0.027} & \cellcolor{gray!6}{0.895} & \cellcolor{gray!6}{-0.480} & \cellcolor{gray!6}{0.021} & \cellcolor{gray!6}{0.022} & \cellcolor{gray!6}{0.833} & \cellcolor{gray!6}{-0.828} & \cellcolor{gray!6}{0.030} & \cellcolor{gray!6}{0.031} & \cellcolor{gray!6}{0.348}\\
\hspace{1em}MV-Egger & 0.113 & 0.036 & 0.036 & 0.938 & -0.480 & 0.021 & 0.020 & 0.798 & -0.828 & 0.030 & 0.031 & 0.340\\
\hspace{1em}\cellcolor{gray!6}{MV-Median} & \cellcolor{gray!6}{0.117} & \cellcolor{gray!6}{0.033} & \cellcolor{gray!6}{0.036} & \cellcolor{gray!6}{0.935} & \cellcolor{gray!6}{-0.477} & \cellcolor{gray!6}{0.027} & \cellcolor{gray!6}{0.029} & \cellcolor{gray!6}{0.896} & \cellcolor{gray!6}{-0.827} & \cellcolor{gray!6}{0.041} & \cellcolor{gray!6}{0.045} & \cellcolor{gray!6}{0.647}\\
\hspace{1em}GRAPPLE & 0.101 & 0.030 & 0.031 & 0.959 & -0.500 & 0.023 & 0.023 & 0.952 & -0.898 & 0.034 & 0.035 & 0.954\\
\hspace{1em}\cellcolor{gray!6}{MRBEE} & \cellcolor{gray!6}{0.099} & \cellcolor{gray!6}{0.031} & \cellcolor{gray!6}{0.032} & \cellcolor{gray!6}{0.944} & \cellcolor{gray!6}{-0.502} & \cellcolor{gray!6}{0.024} & \cellcolor{gray!6}{0.024} & \cellcolor{gray!6}{0.937} & \cellcolor{gray!6}{-0.904} & \cellcolor{gray!6}{0.036} & \cellcolor{gray!6}{0.037} & \cellcolor{gray!6}{0.949}\\
\hspace{1em}SRIVW & 0.100 & 0.031 & 0.031 & 0.955 & -0.501 & 0.024 & 0.024 & 0.949 & -0.902 & 0.036 & 0.036 & 0.948\\
\addlinespace[0.3em]
\multicolumn{13}{l}{\textit{Mean $\hat \lambda_{\min}/\sqrt{p}$ = 21.1, mean conditional $F$-statistics = 3.6, 4.7, 3.1}}\\
\hspace{1em}\cellcolor{gray!6}{MV-IVW} & \cellcolor{gray!6}{\red 0.148} & \cellcolor{gray!6}{0.047} & \cellcolor{gray!6}{0.046} & \cellcolor{gray!6}{0.799} & \cellcolor{gray!6}{-0.410} & \cellcolor{gray!6}{0.037} & \cellcolor{gray!6}{0.038} & \cellcolor{gray!6}{0.345} & \cellcolor{gray!6}{-0.630} & \cellcolor{gray!6}{0.049} & \cellcolor{gray!6}{0.053} & \cellcolor{gray!6}{0.002}\\
\hspace{1em}MV-Egger & 0.136 & 0.067 & 0.069 & 0.923 & -0.410 & 0.038 & 0.039 & 0.343 & -0.630 & 0.050 & 0.054 & 0.001\\
\hspace{1em}\cellcolor{gray!6}{MV-Median} & \cellcolor{gray!6}{0.146} & \cellcolor{gray!6}{0.058} & \cellcolor{gray!6}{0.056} & \cellcolor{gray!6}{0.858} & \cellcolor{gray!6}{-0.401} & \cellcolor{gray!6}{0.050} & \cellcolor{gray!6}{0.049} & \cellcolor{gray!6}{0.465} & \cellcolor{gray!6}{-0.599} & \cellcolor{gray!6}{0.077} & \cellcolor{gray!6}{0.073} & \cellcolor{gray!6}{0.024}\\
\hspace{1em}GRAPPLE & 0.103 & 0.077 & 0.080 & 0.958 & -0.497 & 0.055 & 0.056 & 0.958 & -0.889 & 0.086 & 0.086 & 0.948\\
\hspace{1em}\cellcolor{gray!6}{MRBEE} & \cellcolor{gray!6}{0.092} & \cellcolor{gray!6}{0.087} & \cellcolor{gray!6}{0.090} & \cellcolor{gray!6}{0.960} & \cellcolor{gray!6}{-0.508} & \cellcolor{gray!6}{0.061} & \cellcolor{gray!6}{0.062} & \cellcolor{gray!6}{0.950} & \cellcolor{gray!6}{-0.929} & \cellcolor{gray!6}{0.109} & \cellcolor{gray!6}{0.112} & \cellcolor{gray!6}{0.966}\\
\hspace{1em}SRIVW & 0.095 & 0.085 & 0.085 & 0.957 & -0.505 & 0.060 & 0.060 & 0.958 & -0.917 & 0.104 & 0.105 & 0.960\\
\addlinespace[0.3em]
\multicolumn{13}{l}{\textit{Mean $\hat \lambda_{\min}/\sqrt{p}$ = 7.2, mean conditional $F$-statistics = 2.0, 2.4, 1.8}}\\
\hspace{1em}\cellcolor{gray!6}{MV-IVW} & \cellcolor{gray!6}{\red 0.144} & \cellcolor{gray!6}{0.056} & \cellcolor{gray!6}{0.051} & \cellcolor{gray!6}{0.836} & \cellcolor{gray!6}{-0.310} & \cellcolor{gray!6}{0.045} & \cellcolor{gray!6}{0.045} & \cellcolor{gray!6}{0.021} & \cellcolor{gray!6}{-0.418} & \cellcolor{gray!6}{0.058} & \cellcolor{gray!6}{0.059} & \cellcolor{gray!6}{0.000}\\
\hspace{1em}MV-Egger & 0.140 & 0.084 & 0.086 & 0.929 & -0.310 & 0.045 & 0.049 & 0.023 & -0.418 & 0.059 & 0.064 & 0.000\\
\hspace{1em}\cellcolor{gray!6}{MV-Median} & \cellcolor{gray!6}{0.139} & \cellcolor{gray!6}{0.069} & \cellcolor{gray!6}{0.061} & \cellcolor{gray!6}{0.866} & \cellcolor{gray!6}{-0.291} & \cellcolor{gray!6}{0.061} & \cellcolor{gray!6}{0.056} & \cellcolor{gray!6}{0.057} & \cellcolor{gray!6}{-0.364} & \cellcolor{gray!6}{0.086} & \cellcolor{gray!6}{0.073} & \cellcolor{gray!6}{0.000}\\
\hspace{1em}GRAPPLE & 0.104 & 0.164 & 0.168 & 0.962 & -0.493 & 0.105 & 0.108 & 0.957 & -0.879 & 0.175 & 0.177 & 0.945\\
\hspace{1em}\cellcolor{gray!6}{MRBEE} & \cellcolor{gray!6}{0.057} & \cellcolor{gray!6}{0.224} & \cellcolor{gray!6}{0.285} & \cellcolor{gray!6}{0.986} & \cellcolor{gray!6}{-0.524} & \cellcolor{gray!6}{0.136} & \cellcolor{gray!6}{0.153} & \cellcolor{gray!6}{0.982} & \cellcolor{gray!6}{-1.010} & \cellcolor{gray!6}{0.265} & \cellcolor{gray!6}{0.379} & \cellcolor{gray!6}{0.986}\\
\hspace{1em}SRIVW & 0.119 & 0.180 & 0.185 & 0.966 & -0.518 & 0.116 & 0.123 & 0.972 & -0.896 & 0.201 & 0.238 & 0.961\\
\bottomrule
\multicolumn{13}{l}{\small Abbreviations: Est = estimated causal effect; SD = standard deviation; SE = standard error; CP = coverage probability.}
\end{tabular}}
\end{table}

Across all settings in this simulation study, we use the following shared correlation matrix
\begin{align*}
    \Sigma = \begin{bmatrix}
        1 & -0.1 & -0.05\\
     -0.1 &    1 & 0.2 \\
    -0.05 &  0.2 & 1
\end{bmatrix}.
\end{align*}
The rows correspond to LDL-C, HDL-C and TG.

\subsection{Simulation study under balanced horizontal pleiotropy}\label{exposure-simulation-results.-all-beta.exposures-by-d.-balanced-horizontal-pleiotropy.}

The simulation setup is identical to the simulation study presented in the main Table 1, except now we only consider the \textit{very small} IV strength scenario by dividing the true $\gamma_{j1}$'s by 8. To introduce balanced horizontal pleiotropy, we let $\Gamma_j = \bbeta_0^T \bgamma_j + \alpha_j$ for $j = 1, ..., p$, where $\alpha_j \sim N(0,\tau_0^2)$ with $\tau_0 = 2p^{-1}\sum_{j=1}^p\sigma_{Yj}$.

From the Table S4, we can see that when balanced horizontal pleiotropy is present, the SRIVW estimator still has the smallest bias of $\beta_{01}$ among all the methods, and it effectively accounts for the additional variability, yielding SEs that adequately capture the estimation uncertainty. In comparison, GRAPPLE has slight bias and under-coverage of $\beta_{01}$ and MRBEE is noticeably biased.

\begin{table}[H]

\caption{Simulation results for six MVMR estimators in simulation study 1 with 10,000 repetitions, when $\bbeta_0 = (\beta_{01}, \beta_{02}, \beta_{03}) = (0.8, 0.4, 0)$, and the first exposure has weaker IV strength than the other two. $p = 145$. The simulation mean $\hat \lambda_{\min}/\sqrt{p}$ is 10.2 and the mean conditional $F$-statistics are respectively 1.9, 28.9, 15.8. $\hat \lambda_{\min}$ is the minimum eigenvalue of the sample IV strength matrix $\sum_{j=1}^{p} \Omega_j^{-1} \hat\bgamma_j\hat \bgamma_j^T\Omega_j^{-T} - pI_K$.}
\centering
\resizebox{\linewidth}{!}{
\begin{tabular}[t]{cccccccccccccc}
\toprule
\multicolumn{1}{c}{ } & \multicolumn{4}{c}{$\beta_{01}$ = 0.8} & \multicolumn{4}{c}{$\beta_{02}$ = 0.4} & \multicolumn{4}{c}{$\beta_{03}$ = 0} \\
\cmidrule(l{3pt}r{3pt}){2-5} \cmidrule(l{3pt}r{3pt}){6-9} \cmidrule(l{3pt}r{3pt}){10-13}
Estimator & Est & SD & SE & CP & Est & SD & SE & CP & Est & SD & SE & CP\\
\midrule
\cellcolor{gray!6}{MV-IVW} & \cellcolor{gray!6}{0.385} & \cellcolor{gray!6}{0.096} & \cellcolor{gray!6}{0.053} & \cellcolor{gray!6}{0.002} & \cellcolor{gray!6}{0.393} & \cellcolor{gray!6}{0.011} & \cellcolor{gray!6}{0.007} & \cellcolor{gray!6}{0.707} & \cellcolor{gray!6}{-0.024} & \cellcolor{gray!6}{0.018} & \cellcolor{gray!6}{0.010} & \cellcolor{gray!6}{0.384}\\
MV-Egger & 0.425 & 0.139 & 0.140 & 0.237 & 0.393 & 0.011 & 0.013 & 0.951 & -0.024 & 0.018 & 0.019 & 0.768\\
\cellcolor{gray!6}{MV-Median} & \cellcolor{gray!6}{0.386} & \cellcolor{gray!6}{0.129} & \cellcolor{gray!6}{0.091} & \cellcolor{gray!6}{0.040} & \cellcolor{gray!6}{0.394} & \cellcolor{gray!6}{0.016} & \cellcolor{gray!6}{0.012} & \cellcolor{gray!6}{0.833} & \cellcolor{gray!6}{-0.024} & \cellcolor{gray!6}{0.024} & \cellcolor{gray!6}{0.017} & \cellcolor{gray!6}{0.637}\\
GRAPPLE & 0.838 & 0.271 & 0.226 & 0.917 & 0.400 & 0.013 & 0.013 & 0.950 & 0.002 & 0.024 & 0.023 & 0.939\\
\cellcolor{gray!6}{MRBEE} & \cellcolor{gray!6}{0.872} & \cellcolor{gray!6}{0.295} & \cellcolor{gray!6}{0.316} & \cellcolor{gray!6}{0.970} & \cellcolor{gray!6}{0.400} & \cellcolor{gray!6}{0.016} & \cellcolor{gray!6}{0.015} & \cellcolor{gray!6}{0.960} & \cellcolor{gray!6}{0.004} & \cellcolor{gray!6}{0.026} & \cellcolor{gray!6}{0.028} & \cellcolor{gray!6}{0.965}\\
SRIVW & 0.808 & 0.237 & 0.250 & 0.952 & 0.400 & 0.013 & 0.013 & 0.953 & 0.000 & 0.023 & 0.024 & 0.954\\
\bottomrule
\multicolumn{13}{l}{\small Abbreviations: Est = estimated causal effect; SD = standard deviation; SE = standard error; CP = coverage probability.}
\end{tabular}}
\end{table}

\subsection{Simulation study based on individual-level data}\label{subsec: est sigma xj}

In this simulation study, the goal is to closely mimic the process of generating and analyzing summary-level data in a real MR study. We perform the simulation and analysis as follows:
\begin{enumerate}
    \item \textbf{Generation of independent SNPs}: A set of $p = 2000$ independent SNPs are simulated from a multinomial distribution with $P(Z_j = 0) = 0.25$, $P(Z_j = 1) = 0.5$, $P(Z_j = 2) = 0.25$ with sample size $n = 10000$.
    \item \textbf{Generation of exposures and outcome}: Three exposure variables ($X_1$, $X_2$, $X_3$) and an outcome variable ($Y$) are generated according to the following model: $
	    X_k  = \sum_{j=1}^p\gamma_{jk}Z_j + \eta_{X}U + E_X, k=1,\dots, K$, and  $
	    Y = 10+\bm X^T \bbeta_0 + \eta_{Y}U + E_Y $,
where $\eta_X = \eta_Y = 1$, $\bbeta_0 = (1, -1, 0.5)$, $U \sim N(0, 0.6(1-h^2))$ and $E_X, E_Y \sim N(0, 0.4(1-h^2))$ and $K = 3$. For each $k = 1, 2, 3$, we construct $\gamma_{jk} = \phi_j\sqrt{2h^2/s}$ for $j = 1,...,s$ and $\gamma_{jk} = 0$ for $j = s+1,...,p$, where $s$ is the number of nonnull SNPs, $h^2$ is the proportion of variance in $X$ that is explained by the $s$ nonnull SNPs, and $\phi_j$'s are constants generated once from a standard normal distribution. In this simulation study, we set $s = 1000$ (i.e. 50\% null IVs) and  $h^2=0.1$.
  \item \textbf{Generation of individual-level data}: Three individual-level datasets are generated according to the data-generating process in Steps 1-2, which are respectively designated as the exposure, outcome, and selection datasets. 
    \item \textbf{Generation of summary statistics}: We generate summary statistics as follows: (i) $\hat \gamma_{jk}, \hat \sigma_{Xjk}, $ for $j = 1,...,p$ and $ k=1, \dots, K$ from the exposure dataset; (ii) $\hat \Gamma_j, \hat\sigma_{Yj},$ for $ j = 1,...,p$ from the outcome dataset; (iii) $\hat \gamma_{jk}^{*}, \hat \sigma_{Xjk}^{*} $,  $j = 1,...,p$ and $ k=1, \dots, K$ from the selection dataset. All the summary statistics are computed through simple linear regression from the three datasets, respectively, where $\hat \sigma_{Xjk}$, $\hat \sigma_{Yj}$ and $\hat \sigma_{Xjk}^{*}$ are the corresponding SEs.
    \item \textbf{IV selection}: The selection dataset is used for IV selection by applying a p-value thresholding approach. A SNP is selected for MR estimation if it is associated with at least one of the three exposures (with a p-value threshold of 0.01/$K$).
    \item \textbf{Estimation of variance matrices}: We again use the selection dataset to select SNPs whose selection p-values for all exposures are larger than 0.5. Their summary statistics from the exposure dataset are used to estimate the shared correlation matrix $\Sigma$ using the approach in \citeSupp{Wang2021grapple}. The variances $\Sigma_{Xj}$'s for SNP-exposure associations are then obtained according to Assumption 1(iii). See more details below.
    \item \textbf{Causal effect estimation}: The exposure and outcome datasets are restricted to the SNPs selected during Step 5. Causal effects are then estimated using SRIVW and the other five methods. 
\end{enumerate}
The above procedure (Steps 1-7) is carried out for each simulation run. The results are presented in Table S5. Again, SRIVW performs the best in terms of all metrics. It is worth noting that its SEs adequately capture the variability of the estimates, indicating that ignoring the uncertainty in estimating $\sigma_{Xjk}$'s, $\sigma_{Yj}$'s, and $\Sigma$ has negligible impact on inference.

\begin{table}[ht]

\caption{\protect\label{tab: ind lvl} Simulation results for six MVMR estimators in the individual-level data simulation with 10,000 repetitions. $\bbeta_0 = (\beta_{01},\beta_{02},\beta_{03}) = (1,-0.5,0.5)$.  $p = 2000$ (including 50\% null IVs). The mean and SD for the number of SNPs selected for causal effect estimation in the simulations are  111 and 10, respectively. Among the selected SNPs, the simulation mean of $\hat\lambda_{\min}/\sqrt{p}\approx 8$. The mean conditional $F$-statistics for the three exposures are respectively 3.0, 3.0, 3.4.}
\centering
\resizebox{\linewidth}{!}{
\begin{tabular}[t]{c|cccc|cccc|cccc}
\toprule
\multicolumn{1}{c}{ } & \multicolumn{4}{c}{$\beta_{01}$ = 1} & \multicolumn{4}{c}{$\beta_{02}$ = -0.5} & \multicolumn{4}{c}{$\beta_{03}$ = 0.5} \\
\cmidrule(l{3pt}r{3pt}){2-5} \cmidrule(l{3pt}r{3pt}){6-9} \cmidrule(l{3pt}r{3pt}){10-13}
Estimator & Est & SD & SE & CP & Est & SD & SE & CP & Est & SD & SE & CP\\
\midrule
\cellcolor{gray!6}{MV-IVW} & \cellcolor{gray!6}{0.665} & \cellcolor{gray!6}{0.135} & \cellcolor{gray!6}{0.132} & \cellcolor{gray!6}{0.285} & \cellcolor{gray!6}{-0.525} & \cellcolor{gray!6}{0.132} & \cellcolor{gray!6}{0.130} & \cellcolor{gray!6}{0.944} & \cellcolor{gray!6}{0.298} & \cellcolor{gray!6}{0.132} & \cellcolor{gray!6}{0.130} & \cellcolor{gray!6}{0.644}\\
MV-Egger & 0.678 & 0.219 & 0.219 & 0.677 & -0.525 & 0.133 & 0.134 & 0.944 & 0.298 & 0.132 & 0.133 & 0.665\\
\cellcolor{gray!6}{MV-Median} & \cellcolor{gray!6}{0.668} & \cellcolor{gray!6}{0.167} & \cellcolor{gray!6}{0.170} & \cellcolor{gray!6}{0.494} & \cellcolor{gray!6}{-0.525} & \cellcolor{gray!6}{0.167} & \cellcolor{gray!6}{0.169} & \cellcolor{gray!6}{0.982} & \cellcolor{gray!6}{0.298} & \cellcolor{gray!6}{0.166} & \cellcolor{gray!6}{0.169} & \cellcolor{gray!6}{0.775}\\
GRAPPLE & 0.937 & 0.249 & 0.238 & 0.921 & -0.520 & 0.254 & 0.241 & 0.946 & 0.459 & 0.227 & 0.219 & 0.936\\
\cellcolor{gray!6}{MRBEE} & \cellcolor{gray!6}{1.060} & \cellcolor{gray!6}{0.292} & \cellcolor{gray!6}{0.346} & \cellcolor{gray!6}{0.969} & \cellcolor{gray!6}{-0.508} & \cellcolor{gray!6}{0.282} & \cellcolor{gray!6}{0.335} & \cellcolor{gray!6}{0.960} & \cellcolor{gray!6}{0.534} & \cellcolor{gray!6}{0.251} & \cellcolor{gray!6}{0.288} & \cellcolor{gray!6}{0.965}\\
SRIVW & 0.991 & 0.226 & 0.245 & 0.964 & -0.511 & 0.237 & 0.236 & 0.953 & 0.498 & 0.210 & 0.219 & 0.960\\
\bottomrule
\multicolumn{13}{l}{\footnotesize Abbreviations: Est = estimated causal effect; SD = standard deviation; SE = standard error; CP = coverage probability.}
\end{tabular}}
\end{table}

\section{Additional details for real data applications}
\subsection{Data sources}\label{sec: additional real data}

Table \ref{supp table: data source} summarizes data sources for the real data applications.

{\small
\begin{landscape}
\begin{longtable}{p{3cm} p{2cm} p{2.5cm} p{3cm} p{3cm} p{8cm}}
\caption{Data sources for exposures and outcomes.} \\
\toprule
\textbf{Phenotype} & \textbf{Sample size} & \textbf{Ancestry} & \textbf{Source} & \textbf{Reference} & \textbf{Download link} \\
\midrule
\endfirsthead

\multicolumn{6}{c}{{\textit{(Continued from previous page)}}} \\
\toprule
\textbf{Exposure} & \textbf{Sample size} & \textbf{Ancestry} & \textbf{Source} & \textbf{Reference} & \textbf{Download link} \\
\midrule
\endhead

\midrule
\multicolumn{6}{r}{{\textit{Continued on next page}}} \\
\endfoot

\bottomrule \label{supp table: data source}
\endlastfoot

Lipoprotein (a) & 284,044 & Mostly European & GWAS Catalog & \citeSupp{sinnott2021genetics} & \texttt{https://www.ebi.ac.uk/gwas/studies/GCST90019513} \\
Body fat mass & 337,196 & European & GWAS Catalog & \citeSupp{harris2024new} & \texttt{https://www.ebi.ac.uk/gwas/studies/GCST90428121} \\
Body fat-free mass & 337,196 & European & GWAS Catalog & \citeSupp{harris2024new} & \texttt{https://www.ebi.ac.uk/gwas/studies/GCST90428120} \\
Adult BMI & 806,834 & European & UKB + GIANT & \citeSupp{pulit2019meta} & \texttt{https://zenodo.org/records/1251813} \\
Childhood BMI & 39,620 & European & GWAS Catalog & \citeSupp{vogelezang2020novel} & \texttt{https://www.ebi.ac.uk/gwas/studies/GCST90002409} \\
8-year-old BMI & 28,681 & European & MoBa & \citeSupp{helgeland2022characterization} & \texttt{https://www.fhi.no/en/ch/studies/moba/} \\
Visceral adipose tissue & 9,076 & European & GWAS Catalog & \citeSupp{agrawal2022inherited} & \texttt{https://cvd.hugeamp.org} \\
Gluteofemoral adipose tissue & 9,076 & European & GWAS Catalog & \citeSupp{agrawal2022inherited} & \texttt{https://cvd.hugeamp.org} \\
Traditional lipid traits & 94,674 & Multi-ethnic & GWAS Catalog & \citeSupp{hoffmann2018large} & \texttt{https://www.ebi.ac.uk/gwas/publications/29507422} \\
 & 188,578 & European & GLGC & \citeSupp{global2013discovery} & \texttt{http://csg.sph.umich.edu/abecasis/public/lipids2013/} \\
Lipoprotein subfraction traits & 8,372 & Finnish &  & \citeSupp{davis2017common}  & \texttt{http://csg.sph.umich.edu/boehnke/public/} \\
 & 24,925 & European & IEU & \citeSupp{kettunen2016genome} & \texttt{https://gwas.mrcieu.ac.uk/datasets/} \\

Coronary artery disease &  & Mostly European & CARDIoGRAM-plusC4D + UKB & \citeSupp{nelson2017association} & \texttt{https://www.cardiogramplusc4d.org/data-downloads/} \\
Atrial fibrillation & 60,620/970,216 & \ European &  & \citeSupp{nielsen2018biobank} & \texttt{https://csg.sph.umich.edu/willer/public/afib2018/} \\
Breast Cancer & 76,192/63,082 & European & GWAS Catalog & \citeSupp{michailidou2017association} & \texttt{https://www.ebi.ac.uk/gwas/studies/GCST004988} \\

\midrule
\multicolumn{6}{l}{\footnotesize Note: for binary traits, sample sizes are reported as cases / controls}

\end{longtable}
\end{landscape}
}

\subsection{List of lipoprotein subfraction traits in Example 5}

Table \ref{supp table: list of subfrac} includes the 19 lipoprotein subfractions in Example 5.  A genetic correlation matrix among the 19 lipoprotein subfractions and the three traditional lipid traits are available in the Figure 1 of \citeSupp{zhao2021Mendelian}.

\begin{table}[h]
\centering
\caption{Lipoprotein subfractions in the real data application}
\begin{tabular}{ll}
\toprule
\textbf{Abbreviation} & \textbf{Description} \\
\midrule
XS-VLDL-PL  & Extra small VLDL phospholipids \\
LDL-D       & Mean diameter for LDL particles \\
L-LDL-PL    & Large LDL phospholipids \\
L-LDL-P     & Large LDL particle concentration \\
L-LDL-FC    & Large LDL free cholesterol \\
L-LDL-C     & Large LDL total cholesterol \\
M-LDL-P     & Medium LDL particle concentration \\
M-LDL-CE    & Medium LDL cholesterol esters \\
S-LDL-C     & Small LDL total cholesterol \\
S-LDL-P     & Small LDL particle concentration \\
IDL-FC      & IDL free cholesterol \\
IDL-C       & IDL total cholesterol \\
HDL-D       & Mean diameter for HDL particles \\
XL-HDL-TG   & Extra large HDL triglycerides \\
L-HDL-PL    & Large HDL phospholipids \\
M-HDL-PL    & Medium HDL phospholipids \\
M-HDL-P     & Medium HDL particle concentration \\
M-HDL-CE    & Medium HDL cholesterol esters \\
S-HDL-P     & Small HDL particle concentration \\
\bottomrule \label{supp table: list of subfrac}
\end{tabular}
\end{table}

\subsection{Estimation of the shared correlation matrix}

In practice, when the shared correlation matrix $\Sigma$ in Assumption 1 is not available, we adopt the approach from  \citeSupp{Wang2021grapple} to estimate it. As shown in \citeSupp{bulik2015atlas} and  S2 Text in the supporting information of \citeSupp{Wang2021grapple}, when each genetic effect is very small, 
\begin{align*}
    {\rm Cor}(\hat \gamma_{jt}, \hat \gamma_{js}) \approx \frac{n_{ts}}{\sqrt{ n_tn_s}} {\rm Cor} (X_t, X_s)
\end{align*}
where $ n_s $ is the sample size for the $ s $-th exposure variable, $ n_{ts} $ is the number of shared samples between the datasets for exposure variables $X_t$ and $X_s$. This suggests that $\hat \bgamma_j$'s all share the same correlation, which makes it reasonable to assume 
\begin{align*}
\Sigma_{Xj} = 
	 {\rm {\rm diag}}(\sigma_{Xj1}, \dots, \sigma_{Xjk}) \ \Sigma \  {\rm {\rm diag}}(\sigma_{Xj1}, \dots, \sigma_{Xjk}) 
\end{align*}
in our Assumption 1.

We adopt the approach in \citeSupp{Wang2021grapple} to estimate the shared correlation matrix $\Sigma$ from summary statistics. We first choose SNPs whose selection p-values $p_{jk} \ge 0.5$ for all exposures, where $p_{jk}$ is the p-value for the marginal SNP-exposure association between SNP $j$ and exposure $X_k$. For these selected SNPs, denote the corresponding Z-values of $\hat \gamma_{jk}$'s in the exposure dataset as matrix $Z_{T\times K}$, where $T$ is the number of selected SNPs. Then we calculate the sample correlation matrix of $Z_{T\times K}$ as $\hat \Sigma$. Lastly, for any SNP $j$ selected for MVMR estimation, we obtain $\hat \Sigma_{Xj}$ as
\begin{align}\label{Sigma_xj}
	\hat \Sigma_{Xj}  = 
	 {\rm {\rm diag}}(\hat \sigma_{Xj1}, \dots, \hat \sigma_{Xjk}) \ \hat \Sigma \  {\rm {\rm diag}}(\hat \sigma_{Xj1}, \dots, \hat \sigma_{Xjk}) .
\end{align}
where $\hat \sigma_{Xjk}$ is the standard error estimate of $\hat \gamma_{jk}$ from the marginal linear regression of exposure $k$ on $Z_j$.

\subsection{Additional results and discussion for Examples 1-5}

Below, we present and discuss the results of each example in detail. Detailed data sources for the five examples can be found in Section S5.1.

\vspace{3mm}
\textbf{Data preprocessing}

For each analysis, we focus on variants with a minor allele frequency exceeding 0.5\% and an imputation score greater than 0.3 when such information is available. In cases where duplicate variants appear in the summary datasets, we retain only the instance with the smallest p-value. We use the \textsf{getInput} function from the \textsf{GRAPPLE} package to perform LD clumping and harmonization. The shared correlation matrix $\Sigma$ is also estimated based on summary statistics using the \textsf{getInput} function. To quantify IV strength, we calculate both conditional F-statistics using the \textsf{strength\_mvmr} function from the \textsf{MVMR} package and the IV strength parameter, consistently defined across all analyses as $\hat \lambda_{\min} / \sqrt{p}$, where $\hat \lambda_{\min}$ is the minimum eigenvalue of the sample IV strength matrix (see Section 3.3 of the main text), and $p$ is the number of SNPs. 

\vspace{6mm}
\textbf{Example 1: Lp(a) and LDL-C on CAD}

In this example, the exposures are the concentrations of Lp(a) and LDL-C and the outcome is coronary artery disease (CAD). We use summary statistics from \citeSupp{sinnott2021genetics} for Lp(a) and \citeSupp{global2013discovery} for LDL-C, which are obtained through GWAS Catalog with study accession number GCST90019513, and Global Lipids Genetics Consortium, respectively. The summary statistics for CAD are from \citeSupp{nelson2017association}. The results are presented in Figure \ref{fig lpa}. After SNP selection and harmonization, we retain 74 SNPs with an IV strength parameter of 795.3. The conditional F-statistics are 87.1 for Lp(a) and 134.3 for LDL-C. Both the IV strength parameter and conditional F-statistics indicate strong IV strength. As expected, all estimators provide very similar point estimates, suggesting that Lp(a) and LDL-C independently have harmful effects on CAD, consistent with previous findings from \citeSupp{burgess2018association}. When IV strength is strong, IVW and SRIVW yield nearly identical point estimates and inferences. Note that both IVW and the default SRIVW (Equation (4) in the main text) assume no pleiotropy, leading to narrower confidence intervals (CIs) compared to GRAPPLE, MRBEE, and SRIVW-pleio (see Section S2), which allow for balanced horizontal pleiotropy. Among these three estimators, MRBEE has slightly wider CIs than GRAPPLE and SRIVW-pleio.
\begin{figure}[h]
    \centering
    \includegraphics[width=13cm]{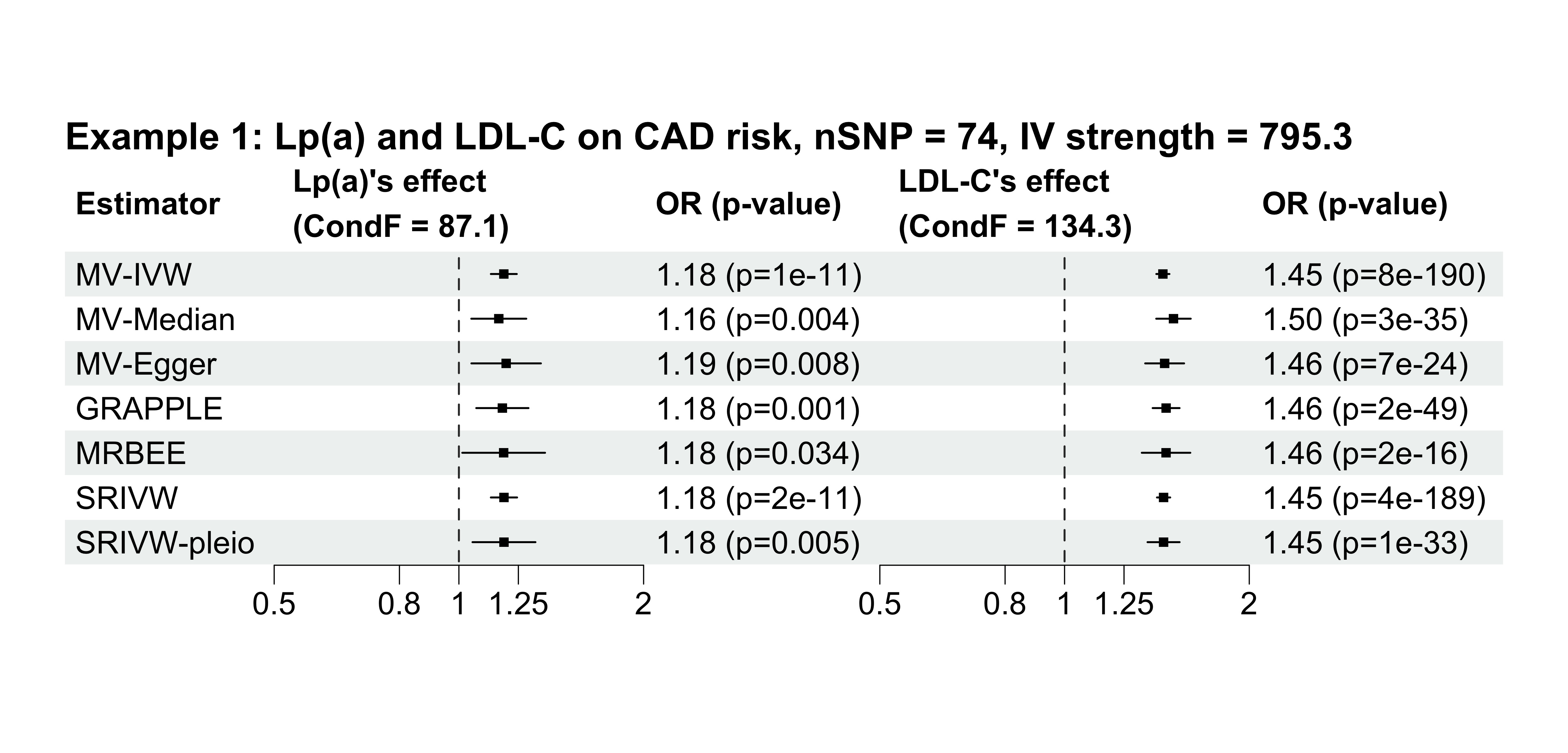}
    \caption{{Estimated effects of Lipoprotein(a) (Lp(a)) and low-density lipoprotein cholesterol (LDL-C) on coronary artery disease (CAD) across different estimators. The odds ratios (ORs) and corresponding p-values are reported for each estimator, with the horizontal bars representing 95\% confidence intervals. IV strength is calculated as $\hat \lambda_{\mathrm{min}}/\sqrt{p}$. Abbreviations: nSNP (number of SNPs); CondF (conditional $F$-statistic).}} \label{fig lpa}
\end{figure}

\vspace{6mm}
\textbf{Example 2: Body fat mass and body fat-free mass on AF}

In this example, the exposures are body fat mass and body fat-free mass and the outcome is AF. We used summary statistics from \citeSupp{harris2024new} for body fat mass and body fat-free mass, which are obtained through GWAS Catalog with study accession number GCST90428121 and GCST90428120, respectively. The summary statistics for AF are from \citeSupp{nielsen2018biobank}. The results are presented in Figure \ref{fig fat mass}. After SNP selection and harmonization, we retain 437 SNPs with an IV strength parameter of 203.8. The conditional F-statistics for body fat mass and body fat-free mass are 11.5 and 12.1, respectively. This example represents a scenario with moderate IV strength. All estimators consistently indicate that both body fat mass and body fat-free mass have independent harmful effects, aligning with findings from \citeSupp{tikkanen2019body}. However, there are slight differences in the point estimates of body fat-free mass's effect across different methods. Specifically, the IVW and SRIVW estimator provide identical estimates, while GRAPPLE and MRBEE yield slightly larger estimates in magnitude.

\begin{figure}[h]
    \centering
    \includegraphics[width=14cm]{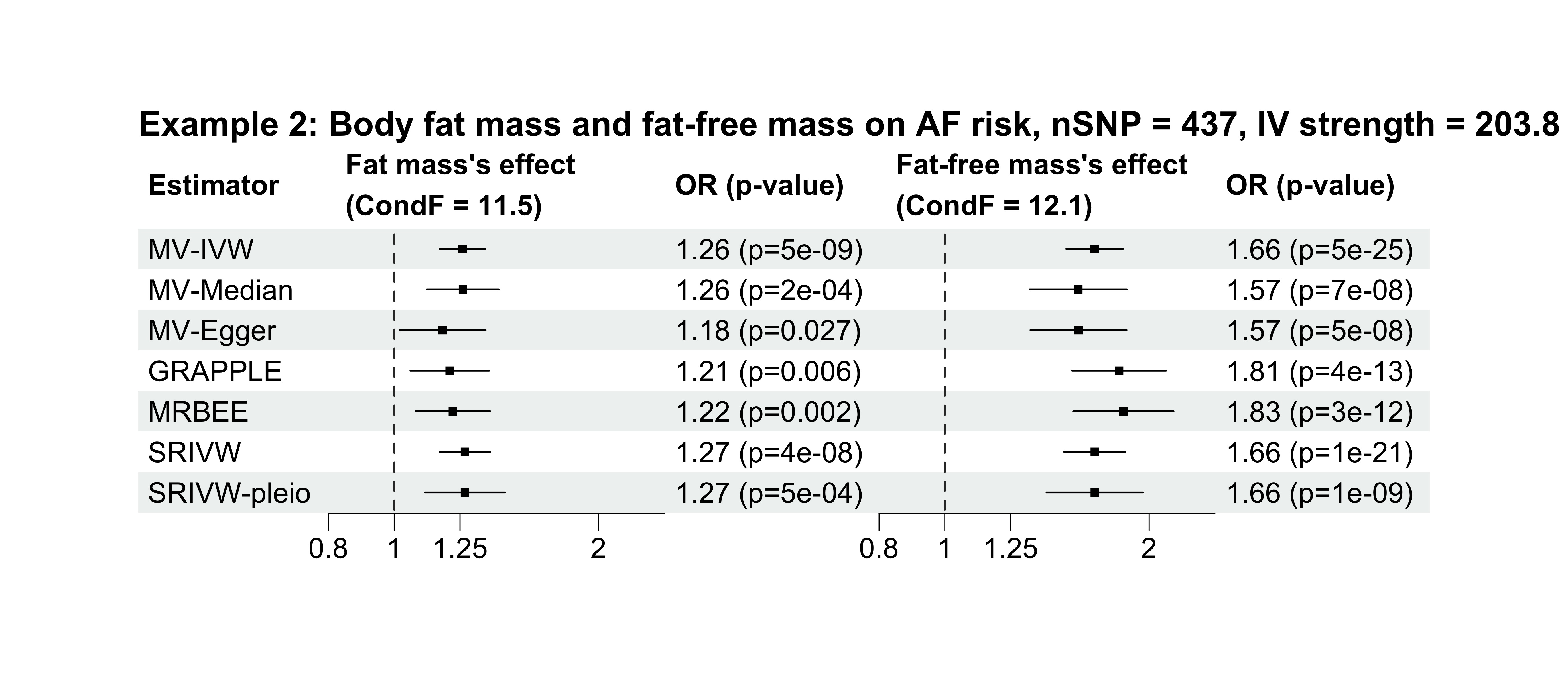}
    \caption{{Estimated effects of body fat mass and body fat-free mass on atrial fibrillation (AF) across different estimators. The odds ratios (ORs) and corresponding p-values are reported for each estimator with the horizontal bars representing 95\% confidence intervals. IV strength is calculated as $\hat \lambda_{\mathrm{min}}/\sqrt{p}$.   Abbreviations: nSNP (number of SNPs); CondF (conditional $F$-statistic).}}\label{fig fat mass}
\end{figure}

\vspace{6mm}
\textbf{Example 3: Adult and childhood BMI on breast cancer}

In this example, we examine the effects of adult and childhood BMI on breast cancer. We used summary statistics from \citeSupp{pulit2019meta} for adult BMI and from \citeSupp{vogelezang2020novel} for childhood BMI. The summary statistics for childhood BMI are available through GWAS Catalog with accession number GCST90002409. After SNP selection and harmonization, we retain 510 SNPs with an IV strength parameter of 19.0. The conditional F-statistics for adult and childhood BMI are 3.5 and 1.9, respectively. This is an example with weak IV strength. The results, shown in the upper row of Figure \ref{bmi breast cancer fig}, reveal that the IVW, MVMR-Median, and MVMR-Egger surprisingly suggest a protective causal effect of adult BMI on breast cancer. In contrast, GRAPPLE and SRIVW, methods that are more robust to weak IVs, indicate a harmful effect. Although we do not know the true direct causal effect of BMI on breast cancer, the IVW estimate is likely directionally biased. Furthermore, IVW suggests that the protective effect of BMI is statistically significant. This supports our theoretical analysis that the IVW method can produce overly narrow confidence intervals, which may yield misleading results when IV strength is weak. 

Regarding the direct effect of childhood BMI, all estimators suggest a protective effect, consistent with previous findings (\citeSupp{richardson2020use, hao2023reassessing}). However, IVW, MVMR-Median, and MVMR-Egger estimates are closer to the null, likely due to weak IV bias.

A significant protective effect of adult BMI on breast cancer by the IVW method is also observed in \citeSupp{wu2024causal} where the authors examined the effects of adult BMI and BMI at age 8 on breast cancer. To further investigate this observation, we re-run our analysis, replacing childhood BMI with 8-year-old BMI, using the same summary statistics for 8-year-old BMI as in \citeSupp{wu2024causal}. The results are presented in the lower row of Figure \ref{bmi breast cancer fig}. Due to the smaller sample size in the GWAS of 8-year-old BMI compared to childhood BMI, the IV strength parameter decreases from 19 to 8.4. While the conditional F-statistic for adult BMI increases by a small amount in this analysis, the IVW method still has a directional bias and an overly narrow confidence interval. Moreover, as overall IV strength decreases, the issue with IVW becomes more pronounced. This additional analysis further supports our claim that examining the conditional F-statistic for a single exposure alone is not sufficient for valid inference of its causal effect in MVMR.

\begin{figure}[h]
    \centering
    \includegraphics[width=\textwidth]{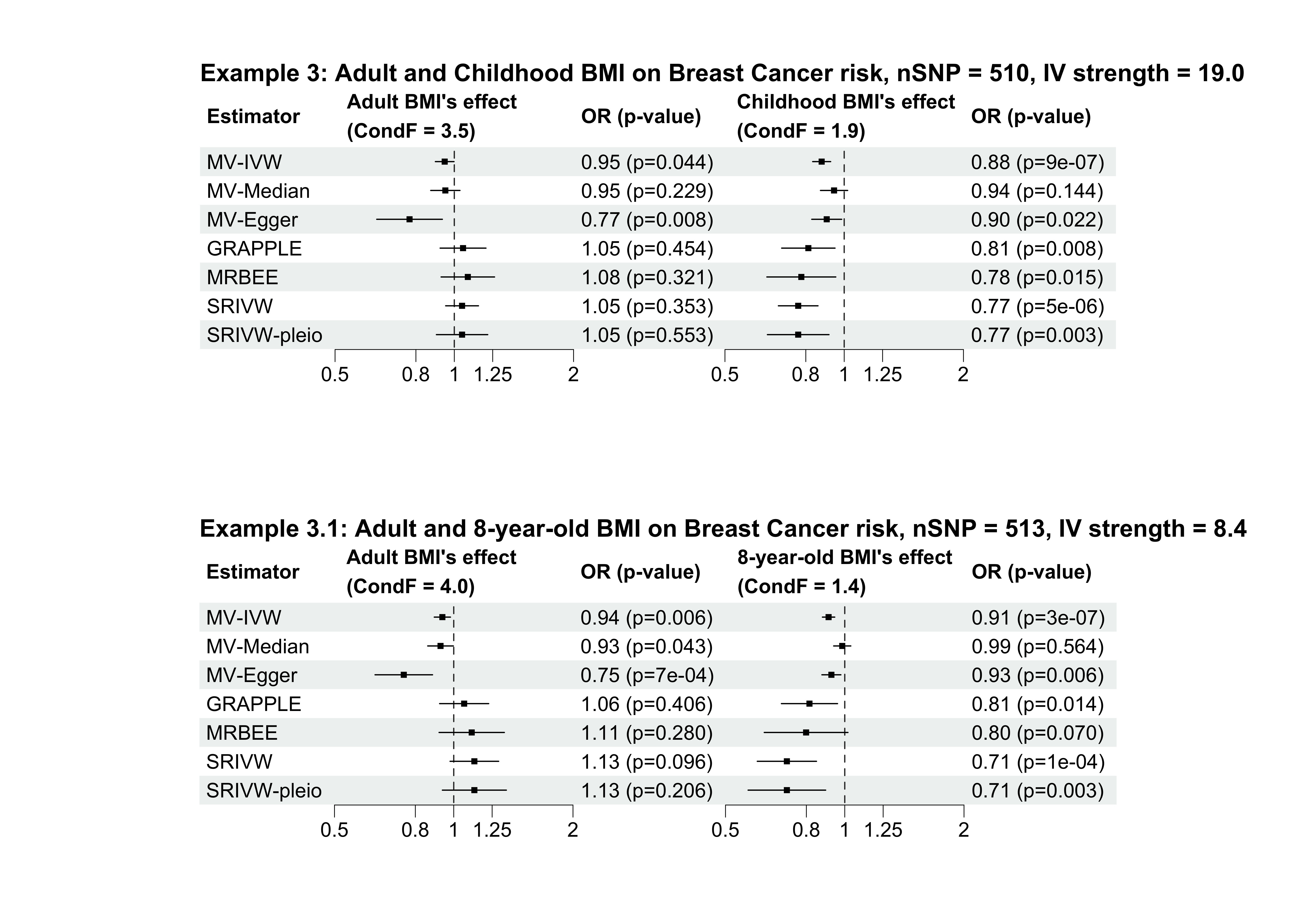}
    \caption{{Estimated effects of adult and childhood body size on breast cancer across different estimators. In the upper row, the analysis uses childhood BMI to represent childhood body size, while in the lower row, 8-year-old BMI is used. The odds ratios (ORs) and corresponding p-values are reported for each estimator with the horizontal bars representing 95\% confidence intervals. IV strength is calculated as $\hat \lambda_{\mathrm{min}}/\sqrt{p}$.  Abbreviations: nSNP (number of SNPs); CondF (conditional $F$-statistic).}}\label{bmi breast cancer fig}
\end{figure}

\vspace{6mm}
\textbf{Example 4: Body fat distribution on CAD}

In this example, we examine the direct causal effects of each individual body fat distribution trait, namely, visceral adipose tissue (VAT), gluteofemoral adipose tissue (GFAT) on CAD in BMI-controlled MVMR analyses. Genetic variants associated with VAT and GFAT are identified from a GWAS conducted by \citeSupp{agrawal2022inherited}; genetic variants associated with CAD are obtained from \citeSupp{nelson2017association}. After SNP selection and harmonization, we end up with 509 SNPs for the VAT analysis, and 511 SNPs for the GFAT analysis. The IV strength parameters are 7.9 and 11.1, respectively, indicating weak IV strength. The results are presented in the top two rows of Figure \ref{fig vat gfat}. 

In the VAT analysis, the conditional F-statistic is 1.3 for VAT and 2.7 for BMI. As observed in Example 3, the confidence intervals by the IVW method are much narrower than those of GRAPPLE, MRBEE and SRIVW, failing to account for the greater uncertainty due to weak IVs. Additionally, the IVW estimates are likely biased, with BMI's estimated effect deviating from null. MVMR-Median and MVMR-Egger have similar issues. The results from MRBEE and SRIVW/SRIVW-pleio suggest that when controlling for VAT, BMI has no direct effect on CAD risk. In comparison, GRAPPLE indicates that BMI still has a direct effect, likely due to its slightly larger point estimate for BMI's direct effect. 

In the GFAT analysis, IVW, MVMR-Median and MVMR-Egger observe similar issues as in the VAT analysis due to weak IV strength. However, unlike in the VAT analysis, their estimates are now closer to the null than those of GRAPPLE, MRBEE and SRIVW. The results from SRIVW suggest that GFAT has a protective direct effect on CAD and BMI has a harmful direct effect.

The observed harmful effect of VAT and protective effect of GFAT are consistent with existing findings (\citeSupp{chen2022effect}).

We also perform an MVMR analysis with VAT and GFAT as exposures, without controlling for BMI. This analysis has greater IV strength and all estimators produce more consistent results. As shown in the last row of Figure \ref{fig vat gfat}, the findings still suggest that VAT may be harmful while GFAT is protective. However, there is likely uncontrolled pleiotropy, as indicated by the results from MVMR-Egger, GRAPPLE, MRBEE and SRIVW-pleio. The notably wider confidence intervals from SRIVW-pleio, compared to SRIVW, demonstrate its ability to account for pleiotropy in practice. 

We attempt to perform an MVMR analysis with BMI as an additional exposure. However, this analysis has an IV strength parameter of -0.9, with conditional F-statistics of 1.1, 1.2, and 2.4 for VAT, GFAT, and BMI, respectively, indicating too weak IV strength. Therefore, we do not proceed further with this analysis.

\begin{figure}[h]
    \centering
    \includegraphics[width=\textwidth]{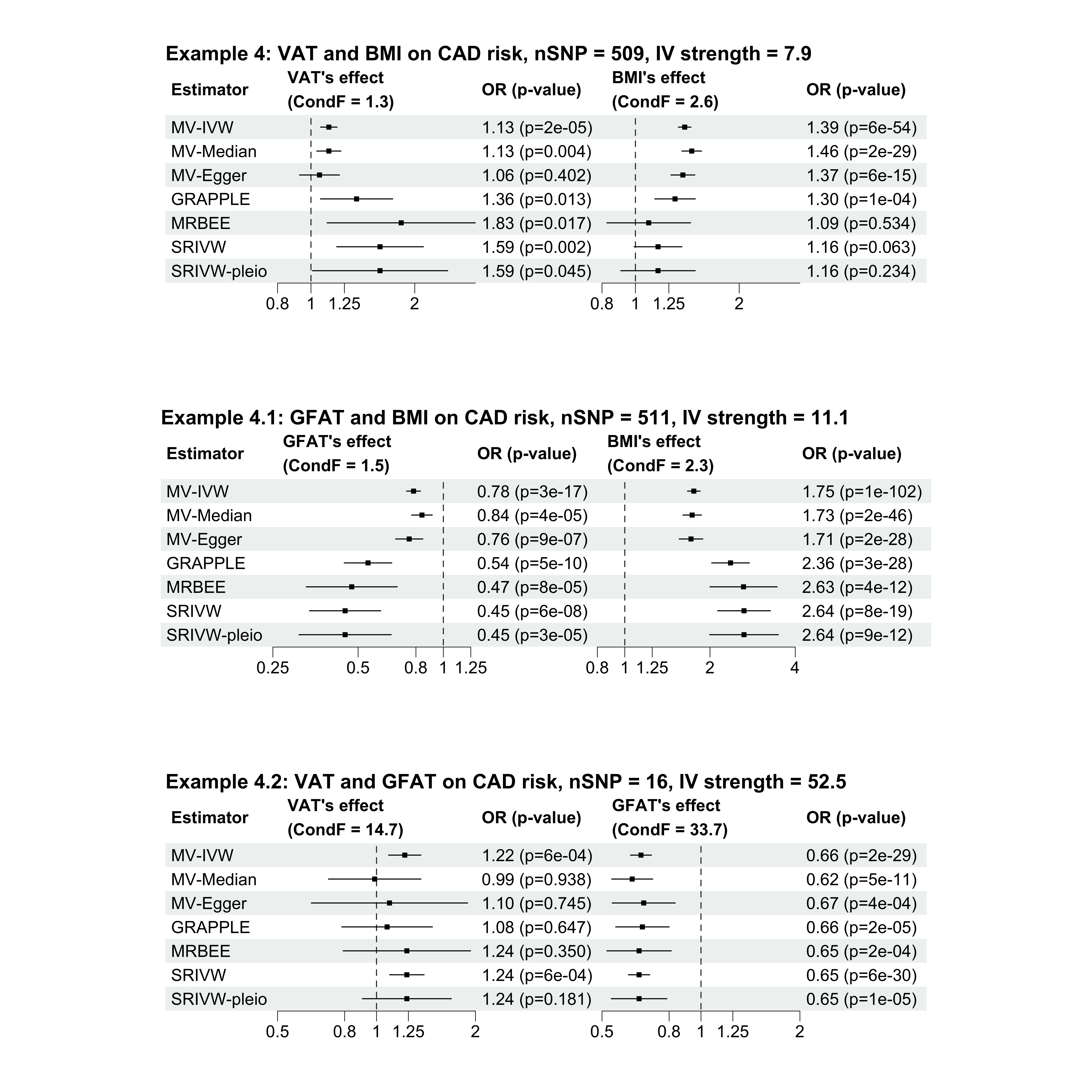}
    \caption{{Estimated effects of visceral adipose tissue (VAT) and gluteofemoral adipose tissue (GFAT) on coronary artery disease (CAD) across different estimators. In the upper two rows, we perform BMI-controlled MVMR analyses to estimate the direct effects of VAT and GFAT, respectively. In the bottom row, we conduct an MVMR analysis with VAT and GFAT as exposures without controlling for BMI. The odds ratios (ORs) and corresponding p-values are reported for each estimator with the horizontal bars representing 95\% confidence intervals. IV strength is calculated as $\hat \lambda_{\mathrm{min}}/\sqrt{p}$. Abbreviations: nSNP (number of SNPs); CondF (conditional $F$-statistic).}}\label{fig vat gfat}
\end{figure}

\vspace{6mm}
\textbf{Example 5: lipoprotein subfraction traits on CAD}

In this example, we estimate the direct effects of 19 lipoprotein subfraction traits on CAD, while controlling for HDL-C, LDL-C, and TG levels. We follow the three-sample design used in \citeSupp{zhao2021Mendelian} and use publicly available GWAS summary statistics to construct the selection, exposure, and outcome datasets. Further details on the list of the 19 lipoprotein subfractions are provided in Section S4.2 of the Supplement.

The upper row of Figure \ref{fig lipid sub} presents the results for M-HDL-P, a medium-sized HDL particle. This analysis uses a p-value threshold of 1e-4, resulting in 299 SNPs with an IV strength parameter of 8.1, indicating weak IV strength. The conditional F-statistics are 1.5 and 3.7 for M-HDL-P and HDL-C, respectively. The lower row of Figure \ref{fig lipid sub} presents the MVMR results when M-HDL-P is excluded (using a more stringent p-value threshold of 5e-8). In this setting, the analysis of the traditional lipid profile has strong IV strength with an IV strength parameter of 525.5. All estimators consistently suggest that both LDL-C and TG have harmful effects on CAD, while HDL-C has little to no effect. However, due to the use of a less stringent p-value threshold and high correlation between M-HDL-P and HDL-C, the inclusion of M-HDL-P introduces more weak IVs and substantially reduces the IV strength. 

Consistent with findings in \citeSupp{zhao2021Mendelian}, results from GRAPPLE, MRBEE and SRIVW indicate that M-HDL-P has a protective direct effect, and that HDL-C appears to have a harmful effect on CAD, controlling for M-HDL-P. Although the point estimates of M-HDL-P and HDL-C from IVW, MVMR-Median and MVMR-Egger are in the same directions, their magnitudes are quite different, likely due to weak IV bias. As observed in Examples 3 and 4, the confidence intervals by the IVW method are much narrower than those from GRAPPLE, MRBEE, and SRIVW.

Similar patterns are observed for the other lipoprotein subfractions, with results presented in Section S4.4 of the Supplement. Due to overly narrow confidence intervals, MV-IVW produces more significant findings than GRAPPLE, MRBEE, and SRIVW, potentially leading to spurious findings. MV-Egger and MV-Median have similar issues. 

Among GRAPPLE, MRBEE, and SRIVW-pleio- estimators that account for balanced horizontal pleiotropy-MRBEE tends to provide the widest confidence intervals, a pattern observed in across all five examples.

\begin{figure}[h]
    \centering
    \includegraphics[width=15cm]{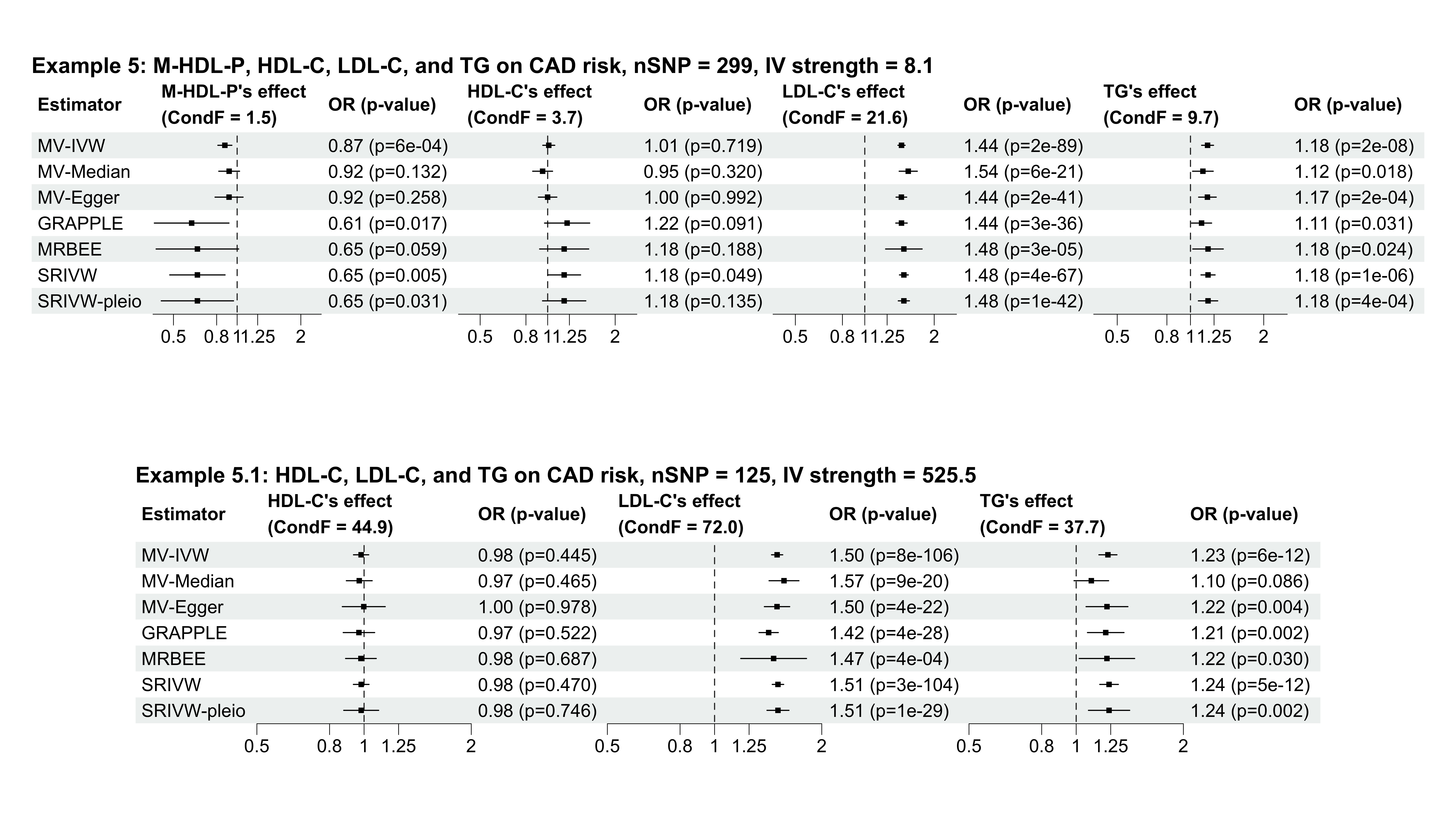}
    \caption{{Estimated effects of medium HDL particle (M-HDL-P) and traditional lipid profile on coronary artery disease (CAD) across different estimators. In the bottom row, we perform an MVMR analysis of HDL-C, LDL-C and TG. In the upper row, we additionally include M-HDL-P as an exposure. The odds ratios (ORs) and corresponding p-values are reported for each estimator with the horizontal bars representing 95\% confidence intervals. IV strength is calculated as $\hat \lambda_{\mathrm{min}}/\sqrt{p}$. Abbreviations: nSNP (number of SNPs); CondF (conditional $F$-statistic).}} \label{fig lipid sub}
\end{figure}

\clearpage

The results for all 19 lipoprotein subfractions are presented in Figure \ref{supp fig: example 6.1}-\ref{supp fig: example 6.5}.

\begin{figure}[h]
    \centering
    \includegraphics[width=\textwidth]{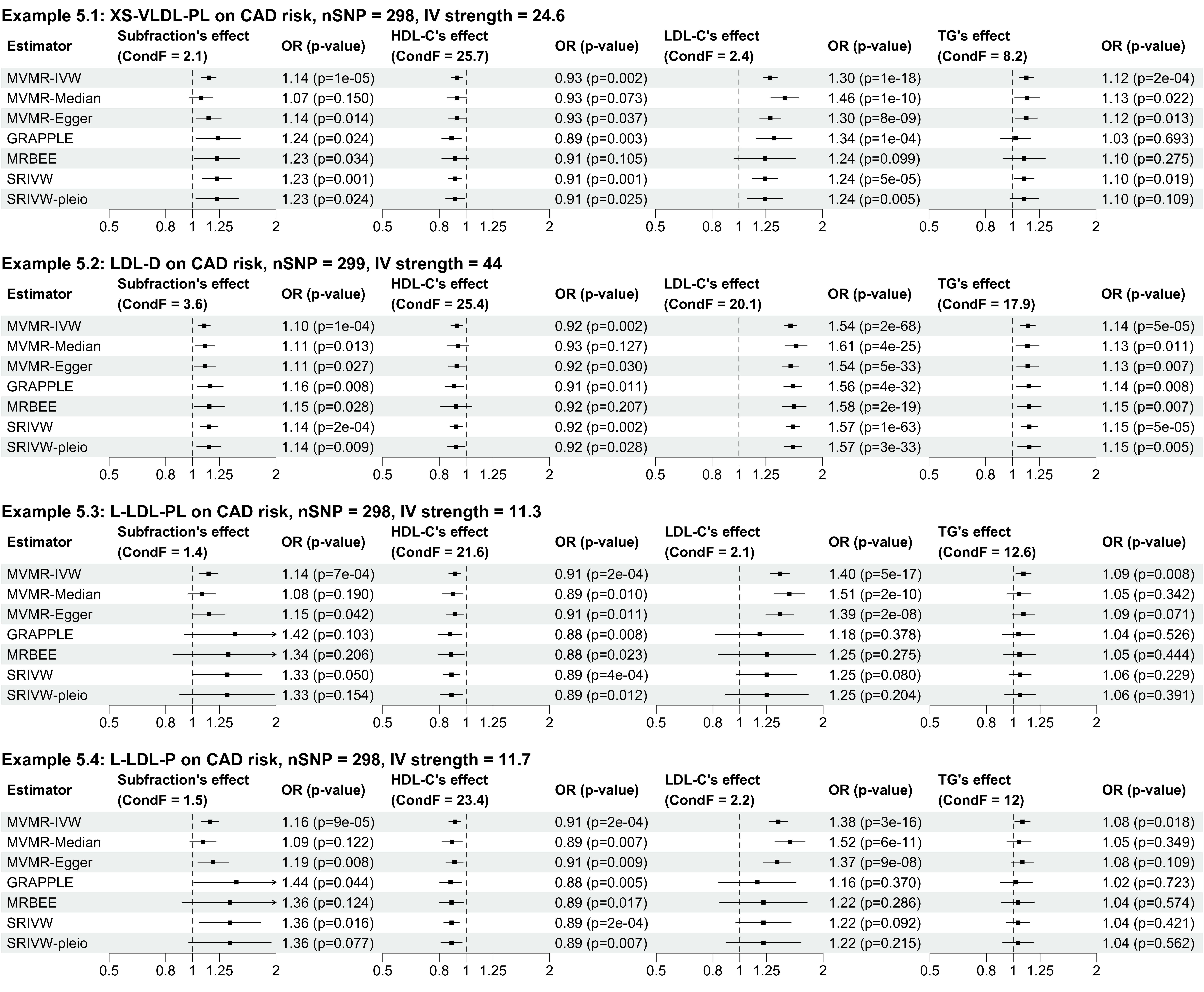}
    \caption{{Estimated effects of lipoprotein subfractions and traditional lipid profile on coronary artery disease from Example 5 across different estimators. The odds ratios (ORs) and corresponding p-values are reported with the horizontal bars representing 95\% confidence intervals. IV strength is calculated as $\hat \lambda_{\mathrm{min}}/\sqrt{p}$, where $\hat \lambda_{\mathrm{min}}$ is the minimum eigenvalue of the sample IV strength matrix and $p$ is the number of SNPs. ``condF'' means conditional $F$-statistic. Abbreviations: CAD (coronary artery disease), nSNP (number of SNPs).}} \label{supp fig: example 6.1}
\end{figure}

\begin{figure}[h]
    \centering
    \includegraphics[width=\textwidth]{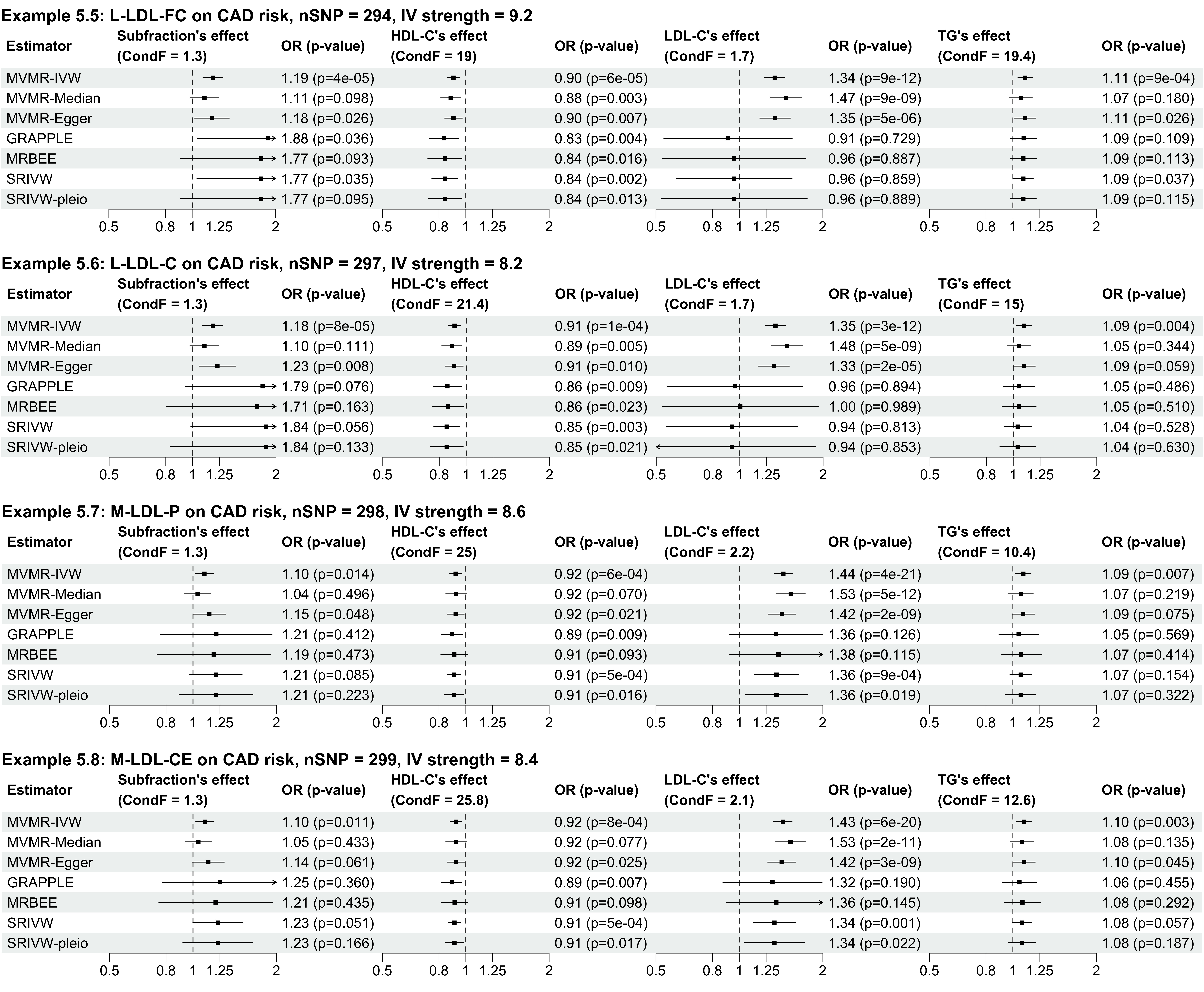}
    \caption{{Estimated effects of lipoprotein subfractions and traditional lipid profile on coronary artery disease from Example 5 across different estimators. The odds ratios (ORs) and corresponding p-values are reported with the horizontal bars representing 95\% confidence intervals. IV strength is calculated as $\hat \lambda_{\mathrm{min}}/\sqrt{p}$, where $\hat \lambda_{\mathrm{min}}$ is the minimum eigenvalue of the sample IV strength matrix and $p$ is the number of SNPs. ``condF'' means conditional $F$-statistic. Abbreviations: CAD (coronary artery disease), nSNP (number of SNPs).}} \label{supp fig: example 6.2}
\end{figure}

\begin{figure}[h]
    \centering
    \includegraphics[width=\textwidth]{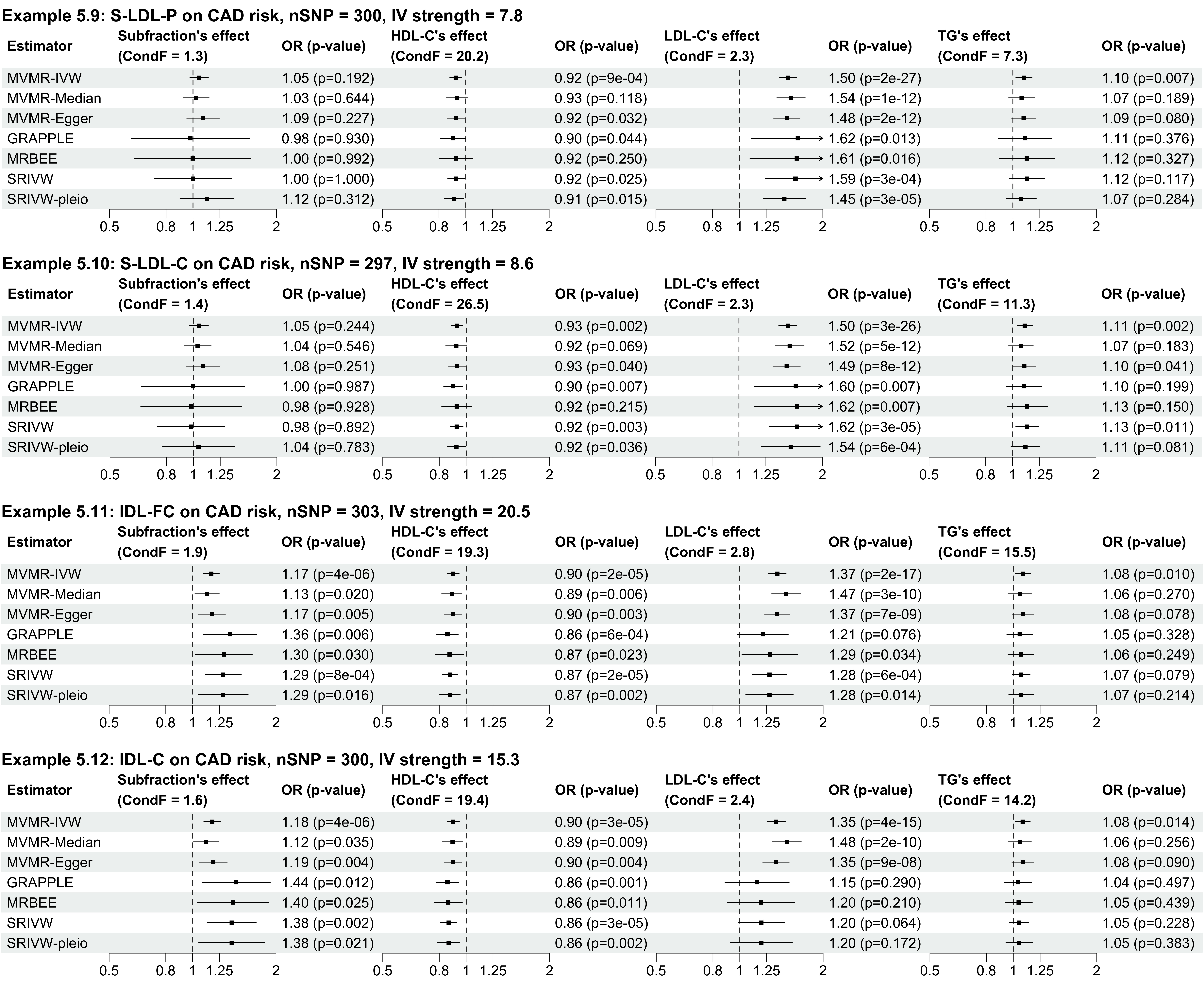}
    \caption{{Estimated effects of lipoprotein subfractions and traditional lipid profile on coronary artery disease from Example 5 across different estimators. The odds ratios (ORs) and corresponding p-values are reported with the horizontal bars representing 95\% confidence intervals. IV strength is calculated as $\hat \lambda_{\mathrm{min}}/\sqrt{p}$, where $\hat \lambda_{\mathrm{min}}$ is the minimum eigenvalue of the sample IV strength matrix and $p$ is the number of SNPs. ``condF'' means conditional $F$-statistic. Abbreviations: CAD (coronary artery disease), nSNP (number of SNPs).}} \label{supp fig: example 6.3}
\end{figure}

\begin{figure}[h]
    \centering
    \includegraphics[width=\textwidth]{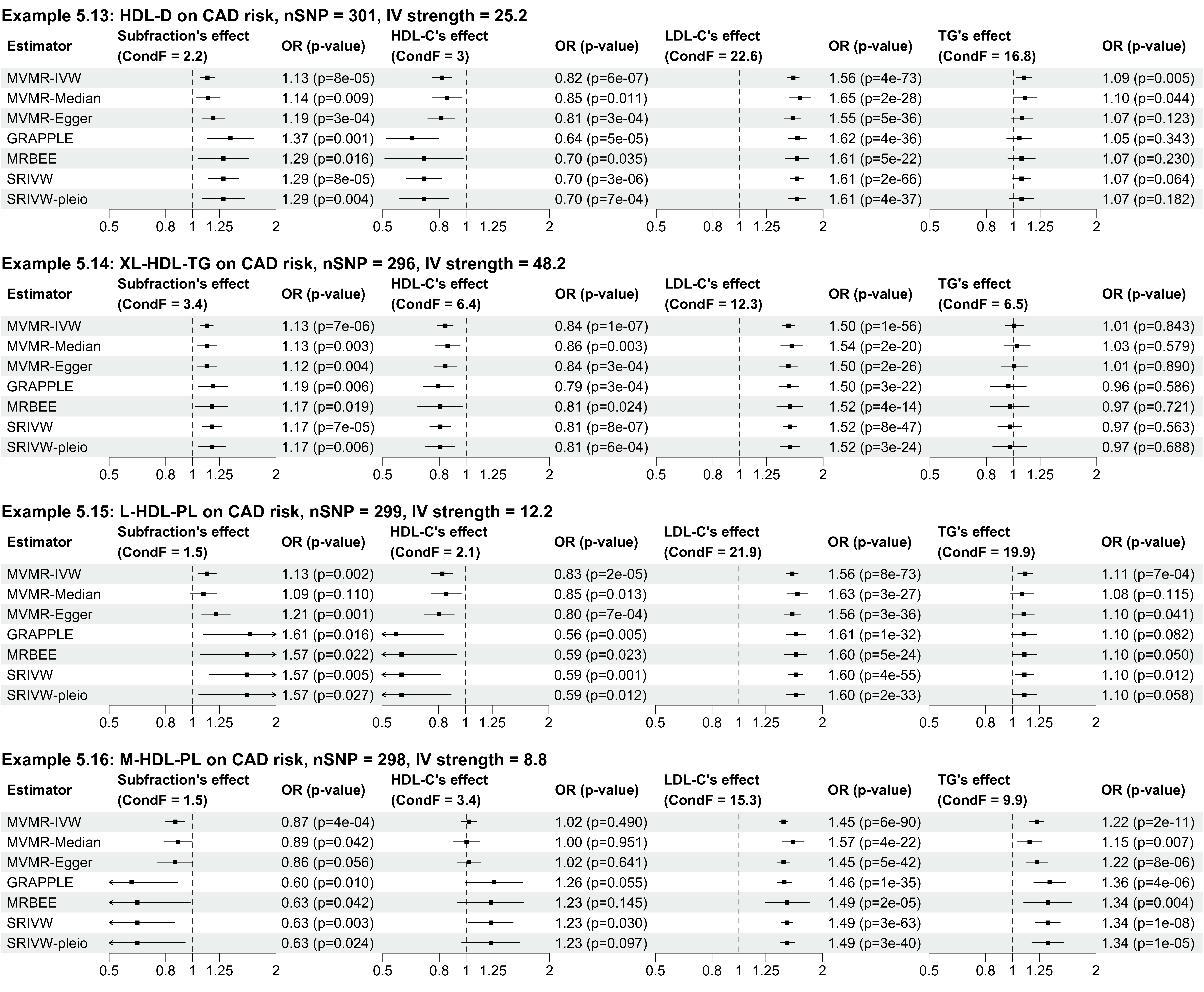}
    \caption{{Estimated effects of lipoprotein subfractions and traditional lipid profile on coronary artery disease from Example 5 across different estimators. The odds ratios (ORs) and corresponding p-values are reported with the horizontal bars representing 95\% confidence intervals. IV strength is calculated as $\hat \lambda_{\mathrm{min}}/\sqrt{p}$, where $\hat \lambda_{\mathrm{min}}$ is the minimum eigenvalue of the sample IV strength matrix and $p$ is the number of SNPs. ``condF'' means conditional $F$-statistic. Abbreviations: CAD (coronary artery disease), nSNP (number of SNPs).}} \label{supp fig: example 6.4}
\end{figure}

\begin{figure}[h]
    \centering
    \includegraphics[width=\textwidth]{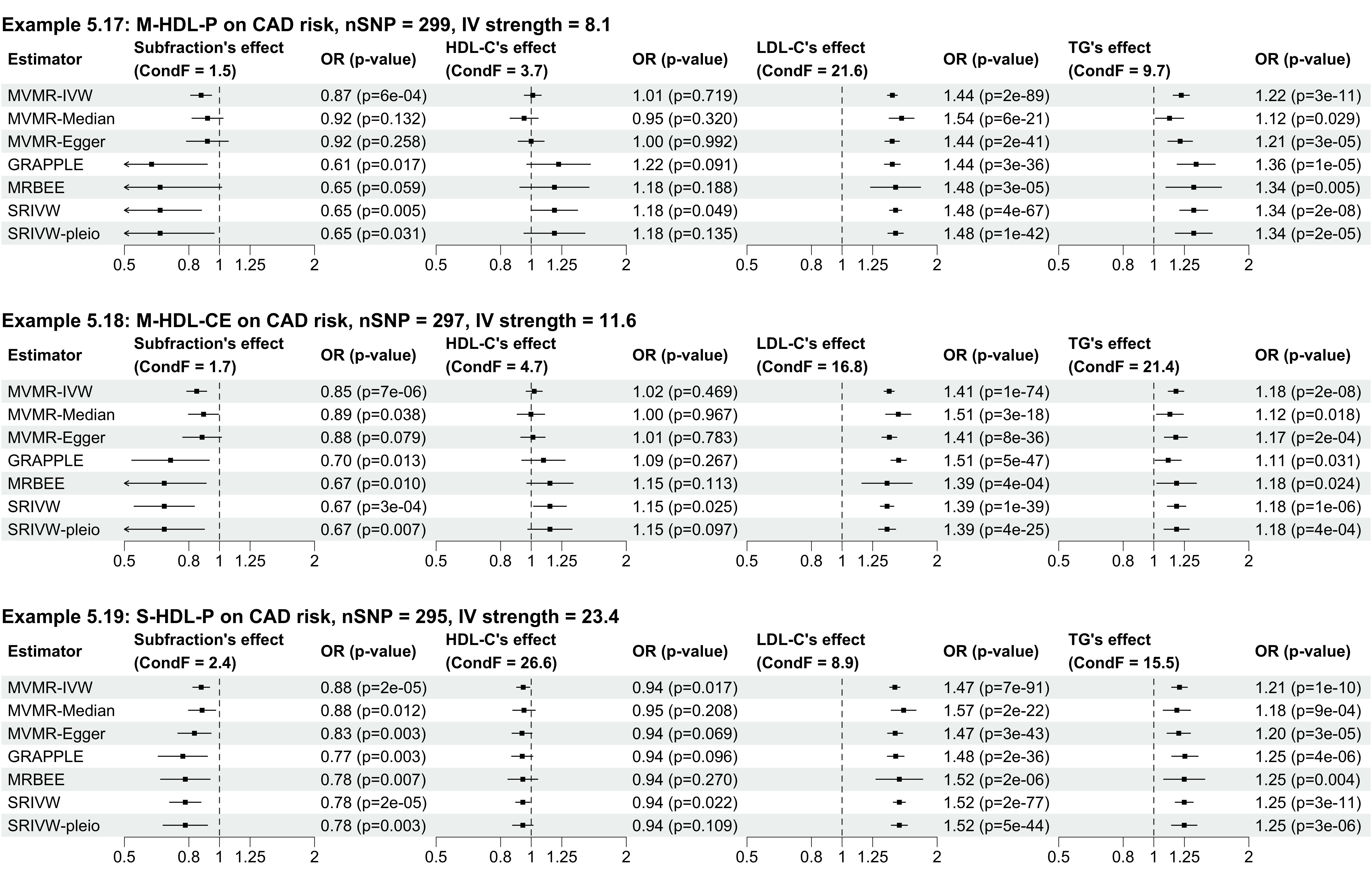}
    \caption{{Estimated effects of lipoprotein subfractions and traditional lipid profile on coronary artery disease from Example 5 across different estimators. The odds ratios (ORs) and corresponding p-values are reported with the horizontal bars representing 95\% confidence intervals. IV strength is calculated as $\hat \lambda_{\mathrm{min}}/\sqrt{p}$, where $\hat \lambda_{\mathrm{min}}$ is the minimum eigenvalue of the sample IV strength matrix and $p$ is the number of SNPs. ``condF'' means conditional $F$-statistic. Abbreviations: CAD (coronary artery disease), nSNP (number of SNPs).}} \label{supp fig: example 6.5}
\end{figure}

\clearpage

\section{Extension to non-independent exposure and outcome datasets}\label{sec: extension overlap}

Our proposed methods can be easily extended to non-independent exposure and outcome datasets. In this section, we sketch this direction of extension.

First, we relax Assumption 1 to allow for correlation between $\hat \Gamma_j$ and $\hat \bgamma_j$. Specifically, we assume
\begin{align}\label{eq: dgp overlap}
\begin{pmatrix}
\hat \bgamma_j \\
\hat \Gamma_j
\end{pmatrix}\sim N\left(\begin{pmatrix}
\bgamma_j \\
\Gamma_j
\end{pmatrix},\Xi_j\right),
\end{align}
where 
\begin{align*}
    \Xi_j = {\rm{{\rm diag}}}(\sigma_{Xj1},...,\sigma_{XjK},\sigma_{Yj})\ \Sigma_{\rm overlap}\ {\rm{{\rm diag}}}(\sigma_{Xj1},...,\sigma_{XjK},\sigma_{Yj}) = \begin{pmatrix}
\Sigma_{Xj} & \Sigma_{XYj} \\
\Sigma_{XYj}^T & \sigma_{Yj}^2
\end{pmatrix} ,
\end{align*}
$\Sigma_{\rm overlap}$ is the positive definite correlation matrix shared across $j$ and $\Sigma_{XYj}$ is a vector representing the known covariance between $\hat \bgamma_j$ and $\hat \Gamma_j$.

We then note that due to the correlation between $\hat \bgamma_j$ and $\hat \Gamma_j$, $E\bigg(\sum_{j}\hat \bgamma_j\hat \Gamma_j\sigma_{Yj}^{-2}\bigg)$ no longer equals $\sum_j M_j \bbeta_0$. Instead, we have
\begin{align}
    E\bigg(\sum_{j}\hat \bgamma_j\hat \Gamma_j\sigma_{Yj}^{-2}\bigg) =  \sum_j\Sigma_{XYj}\sigma_{Yj}^{-2} + \sum_j M_j\bbeta_0
\end{align}
This induces bias to the ``numerator'' of the SRIVW estimator. A straightforward way to correct for this bias is modifying the ``numerator'' to be $\sum_{j}(\hat \bgamma_j\hat \Gamma_j - \Sigma_{XYj})\sigma_{Yj}^{-2}$, leading to the following general formulation of the SRIVW estimator:
\begin{align}\label{eq: SRIVW overlap}
    \hat\bbeta_{\rm SRIVW, \phi} &= \bigg\{ R_{\phi}\bigg(\sum_{j=1}^{p} \hat M_j - V_j\bigg) \bigg\}^{-1}\bigg\{\sum_{j=1}^{p} (\hat \bgamma_j \hat \Gamma_j -\Sigma_{XYj})\sigma_{Yj}^{-2} \bigg\}. 
\end{align} 
with the variance-covariance matrix given by
\begin{align*}    
        {\mathcal{V}}_{\rm SRIVW}  & =  \bigg\{ \sum_{j=1}^{p}M_j \bigg\}^{-1}  \bigg[  \sum_{j=1}^{p}\big\{  (1+ \bm{\beta}_0^TV_j\bm{\beta}_0 - 2\bbeta_0^T W_j) (M_j+V_j)+V_j\bm{\beta}_0\bm{\beta}_0^TV_j \\
        & - V_j\bbeta_0W_j^T - W_j\bbeta_0^TV_j + (1+4\bbeta_0^TW_j)W_jW_j^T\big\} \bigg]  \bigg\{ \sum_{j=1}^{p}M_j \bigg\}^{-1} . \nonumber
\end{align*}
where $W_j = \Sigma_{XYj}\sigma_{Yj}^{-2}$. When there is no overlap between exposure and outcome datasets, i.e., $\Sigma_{XYj} = 0$, the above formula reduces to the formulas in the main article. A consistent estimator of ${\mathcal{V}}_{\rm SRIVW}$ can be obtained by replacing $\sum_{j=1}^{p}M_j$ and $\bbeta_0$ with respectively $R_{\phi}(\sum_{j=1}^{p}\hat M_j - V_j)$ and $\hat \bbeta_{SRIVW, \phi}$.

In practice,  $\Sigma_{Xj}$'s and $\Sigma_{XYj}$'s are typically unknown. We can follow the strategy outlined in \ref{subsec: est sigma xj} to estimate it based on summary data (\citeSupp{Wang2021grapple}). To choose $\phi$ with overlapping datasets, we can minimize the following objective function:
\begin{align}
    J_{\times}(\phi) = \sum_{j=1}^{p} \frac{(\hat \Gamma_j - \hat \bgamma_j^T\hat\bbeta_{\rm SRIVW, \phi})^2}{\sigma_{Yj}^2 + \hat\bbeta_{\rm SRIVW, \phi}^T \Sigma_{Xj} \hat\bbeta_{\rm SRIVW, \phi} - 2 \hat\bbeta_{\rm SRIVW, \phi}^T\Sigma_{XYj}},\ \text{ subject to $\phi \in B_n$}, \nonumber
\end{align}
where $B_n$ can take the same form as described in the main text.

We now perform some simulations to evaluate the performance of this new estimator. We follow the same simulation setup as in the main Table 1 except that we now generate summary statistics from \eqref{eq: dgp overlap} with $\Sigma_{\rm overlap}$ taking the following form.
\begin{align*}
    \Sigma_{\rm overlap} = 
    \begin{bmatrix}
        1 & -0.1 & -0.05 & -0.2\\
     -0.1 &    1 & 0.2   & 0.5\\
    -0.05 &  0.2 & 1     & 0.4\\
    -0.2  & 0.5  &  0.4  & 1
\end{bmatrix}
\end{align*}
We compare the SRIVW estimator presented in Equation (9) of the main paper (ignoring the overlap), the new SRIVW estimator in \eqref{eq: SRIVW overlap} and GRAPPLE which can also handle overlapping datasets (\citeSupp{Wang2021grapple}). 

From Table \ref{tab:overlap}, we can see that ignoring the overlap between exposure and outcome datasets can lead to biased estimates and poor coverage probabilities of 95\% confidence intervals from the standard SRIVW estimator in the presence of weak IVs. In comparison, the new SRIVW estimator effectively corrects for the bias due to the overlap and yields coverage probabilities close to the nominal level of 0.95. GRAPPLE has slightly larger bias and worse coverage probabilities than the new SRIVW estimator.

\begin{table}[h]

\caption{Simulation results for three MVMR estimators in simulation study 1 with 10,000 repetitions, when Factor 1 = (i), Factor 2 = (i), Factor 3 = (i) and when exposure and outcome datasets are overlapping. In other words, $\bbeta_0 = (\beta_{01}, \beta_{02}, \beta_{03}) = (0.8, 0.4, 0)$,  and the first exposure has weaker IV strength than the other two. $p = 145$. $\lambda_{\min}$ is the minimum eigenvalue of the true IV strength matrix $\sum_{j=1}^{p} \Omega_j^{-1} \bgamma_j\bgamma_j^T\Omega_j^{-T}$.}\label{tab:overlap}
\centering
\resizebox{\linewidth}{!}{
\begin{tabular}[t]{c|cccc|cccc|cccc}
\toprule
\multicolumn{1}{c}{ } & \multicolumn{4}{c}{$\beta_{01}$ = 0.8} & \multicolumn{4}{c}{$\beta_{02}$ = 0.4} & \multicolumn{4}{c}{$\beta_{03}$ = 0} \\
\cmidrule(l{3pt}r{3pt}){2-5} \cmidrule(l{3pt}r{3pt}){6-9} \cmidrule(l{3pt}r{3pt}){10-13}
Estimator & Est & SD & SE & CP & Est & SD & SE & CP & Est & SD & SE & CP\\
\midrule
\addlinespace[0.3em]
\multicolumn{13}{l}{\textit{Mean $\hat \lambda_{\min}/\sqrt{p}$ = 103.4, mean conditional $F$-statistics = 9.8, 38.4, 23.5}}\\
\hspace{1em}\cellcolor{gray!6}{SRIVW in (5)} & \cellcolor{gray!6}{0.796} & \cellcolor{gray!6}{0.032} & \cellcolor{gray!6}{0.033} & \cellcolor{gray!6}{0.953} & \cellcolor{gray!6}{0.402} & \cellcolor{gray!6}{0.009} & \cellcolor{gray!6}{0.009} & \cellcolor{gray!6}{0.954} & \cellcolor{gray!6}{0.003} & \cellcolor{gray!6}{0.013} & \cellcolor{gray!6}{0.014} & \cellcolor{gray!6}{0.958}\\
\hspace{1em}SRIVW in (S19) & 0.804 & 0.032 & 0.033 & 0.953 & 0.400 & 0.009 & 0.009 & 0.954 & 0.001 & 0.014 & 0.014 & 0.954\\
\hspace{1em}\cellcolor{gray!6}{GRAPPLE} & \cellcolor{gray!6}{0.798} & \cellcolor{gray!6}{0.030} & \cellcolor{gray!6}{0.031} & \cellcolor{gray!6}{0.950} & \cellcolor{gray!6}{0.400} & \cellcolor{gray!6}{0.009} & \cellcolor{gray!6}{0.009} & \cellcolor{gray!6}{0.954} & \cellcolor{gray!6}{0.000} & \cellcolor{gray!6}{0.013} & \cellcolor{gray!6}{0.014} & \cellcolor{gray!6}{0.957}\\
\addlinespace[0.3em]
\multicolumn{13}{l}{\textit{Mean $\hat \lambda_{\min}/\sqrt{p}$ = 21.7, mean conditional $F$-statistics = 2.9, 28.3, 14.7}}\\
\hspace{1em}SRIVW in (5) & 0.765 & 0.092 & 0.094 & 0.900 & 0.402 & 0.009 & 0.009 & 0.953 & 0.001 & 0.014 & 0.015 & 0.961\\
\hspace{1em}\cellcolor{gray!6}{SRIVW in (S19)} & \cellcolor{gray!6}{0.819} & \cellcolor{gray!6}{0.099} & \cellcolor{gray!6}{0.098} & \cellcolor{gray!6}{0.960} & \cellcolor{gray!6}{0.400} & \cellcolor{gray!6}{0.009} & \cellcolor{gray!6}{0.009} & \cellcolor{gray!6}{0.953} & \cellcolor{gray!6}{0.001} & \cellcolor{gray!6}{0.015} & \cellcolor{gray!6}{0.015} & \cellcolor{gray!6}{0.954}\\
\hspace{1em}GRAPPLE & 0.785 & 0.073 & 0.071 & 0.928 & 0.400 & 0.009 & 0.009 & 0.951 & -0.001 & 0.014 & 0.014 & 0.950\\
\addlinespace[0.3em]
\multicolumn{13}{l}{\textit{Mean $\hat \lambda_{\min}/\sqrt{p}$ = 7.6, mean conditional $F$-statistics = 1.7, 29.3, 17.2}}\\
\hspace{1em}\cellcolor{gray!6}{SRIVW in (5)} & \cellcolor{gray!6}{0.595} & \cellcolor{gray!6}{0.115} & \cellcolor{gray!6}{0.173} & \cellcolor{gray!6}{0.760} & \cellcolor{gray!6}{0.401} & \cellcolor{gray!6}{0.008} & \cellcolor{gray!6}{0.008} & \cellcolor{gray!6}{0.973} & \cellcolor{gray!6}{-0.007} & \cellcolor{gray!6}{0.012} & \cellcolor{gray!6}{0.014} & \cellcolor{gray!6}{0.928}\\
\hspace{1em}SRIVW in (S19) & 0.779 & 0.146 & 0.218 & 0.947 & 0.400 & 0.009 & 0.009 & 0.965 & -0.001 & 0.014 & 0.016 & 0.956\\
\hspace{1em}\cellcolor{gray!6}{GRAPPLE} & \cellcolor{gray!6}{0.763} & \cellcolor{gray!6}{0.145} & \cellcolor{gray!6}{0.136} & \cellcolor{gray!6}{0.890} & \cellcolor{gray!6}{0.400} & \cellcolor{gray!6}{0.009} & \cellcolor{gray!6}{0.009} & \cellcolor{gray!6}{0.957} & \cellcolor{gray!6}{-0.002} & \cellcolor{gray!6}{0.014} & \cellcolor{gray!6}{0.014} & \cellcolor{gray!6}{0.943}\\
\bottomrule
\multicolumn{13}{l}{\footnotesize Abbreviations: Est = estimated causal effect; SD = standard deviation; SE = standard error; CP = coverage probability.}
\end{tabular}}
\end{table}

\spacingset{1.0}
\bibliographystyleSupp{apalike}
\bibliographySupp{reference}

\end{document}